%% file: main.tex
\documentclass[conference,10pt,twocolumn]{IEEEtran}
\usepackage[hidelinks]{hyperref}
\def\thesubsubsectiondis{\Alph{subsection}.\arabic{subsubsection}}
\newcommand{\header}[1]{\subsubsection{#1}}

\usepackage{cite}
\usepackage[pdftex]{graphicx}
\usepackage{amsmath}
\usepackage{array}
\usepackage{url}
\usepackage{verbatim}
\usepackage{xspace}

\usepackage{underscore}           
\usepackage{textcomp}             
\usepackage{prooftree}            
\usepackage{mathtools}            
\usepackage{amssymb}              
\usepackage{enumerate}            
\usepackage{session}

\newcommand{\ie}{i.e.,\xspace}
\newcommand{\eg}{e.g.\xspace}
\newcommand{\cf}{c.f.\xspace}

\usepackage{ifthen}
\newboolean{arxiv} 
\setboolean{arxiv}{true}%
\newcommand{\arxiv}[2]{\ifthenelse{\boolean{arxiv}}{#1}{#2}}

\newcommand{\weaker}{\ensuremath{\preceq}}
\newcommand{\welltyped}{well-typed\xspace}

\begin{document}
\title{Assuming Just Enough Fairness to make \\ Session Types Complete for Lock-freedom}
\author{
\IEEEauthorblockN{Rob van Glabbeek}
\IEEEauthorblockA{
Data61, CSIRO and UNSW, Australia \\
Sydney, Australia 
\\ Email: rvg@cs.stanford.edu
}
\and
\IEEEauthorblockN{Peter H{\"o}fner}
\IEEEauthorblockA{
Australian National University\\
Canberra, Australia
\\ Email: peter.hoefner@anu.edu.au
}
\and
\IEEEauthorblockN{Ross Horne}
\IEEEauthorblockA{Computer Science,
University of Luxembourg \\
Esch-sur-Alzette, Luxembourg
\\ Email: ross.horne@uni.lu
}}

\IEEEoverridecommandlockouts
\arxiv{
\IEEEpubid{\makebox[\columnwidth]{\mbox{To appear in the Proceedings of LICS 2021.}\hfill} \hspace{\columnsep}\makebox[\columnwidth]{ }}
}{
\IEEEpubid{\makebox[\columnwidth]{978-1-6654-4895-6/21/\$31.00~
\copyright2021 IEEE \hfill} \hspace{\columnsep}\makebox[\columnwidth]{ }}
}
\maketitle

\begin{abstract}
We investigate how different fairness assumptions affect results concerning 
\textit{lock-freedom}, a typical liveness property targeted by session type systems.
We fix a minimal session calculus and systematically take into account all known fairness assumptions,
thereby identifying precisely three interesting and semantically distinct notions of lock-freedom, 
all of which having a sound session type system.
We then show that, by using a general merge operator in an otherwise standard approach to global session types, we obtain a session type system complete for the strongest amongst those notions of lock-freedom, which assumes only \emph{justness} of execution paths,
a minimal fairness assumption for concurrent systems.
\end{abstract}

\section{Introduction}
It has long been known that there is an intimate relationship between liveness properties and fairness assumptions.
Seminal work by Owicki and Lamport~\cite{Owicki1982} draws attention to the fact that liveness properties, such as ``each request will eventually be answered'' are indispensable to create correct concurrent programs.

Typically, a liveness property does not hold for all execution paths of a concurrent system: 
imagine two sellers and two buyers: 
\emph{buyer1} repeatedly requests product $A$ from \emph{seller1},
 who is able to sell the product.
Similarly, \emph{buyer2} and \emph{seller2} are able to exchange product $B$.
Assuming that both buyers try to request infinitely many products, there is an infinite execution path where product $A$ is always requested and bought, and $B$ is never sold. When taking all infinite execution paths into consideration, the rudimentary liveness property mentioned by Owicki and Lamport does not hold.
Ranging over all infinite or completed executions -- the default assumption for many model checkers -- 
essentially assumes only that the system as a whole progresses if there is some work to do and there is no deadlock.

When reasoning about starvation-sensitive liveness properties, i.e, properties that avoid situations where a component wants to do something but is denied forever, Owicki and Lamport state explicitly that such liveness properties depend on a fairness assumption.

Assuming that the parties in our example act independently, claiming that the aforementioned liveness property fails is unrealistic, for both sellers should be able to react on any request.
It is reasonable to make
some fairness assumption that ensures that the parties requesting and selling $A$ do not impair the parties involved with $B$. This simple example can be used as a litmus test that any realistic fairness assumption for a concurrent system should pass.

Thus, liveness properties have to be parametrised with a fairness assumption that rules out potential executions of a system.
As the fairness assumption becomes weaker (permitting more executions), the liveness property becomes stronger (systems can do more, so the liveness property is more likely to be rejected).

A reason why there exist different notions of fairness is that some notions are not realistic for some applications. 
For example, an implication of making the strongest of all fairness assumptions might be that
you will phone everyone in your phone book repeatedly, which is unlikely. The minimal assumption \textit{justness}~\cite{GH19} does not entail this, but it does imply that you will not be prevented from having a phone conversation due to unrelated calls between others.  
A recent survey~\cite{GH19} of fairness assumptions classifies dozens of semantically distinct notions by their strength in ruling out potential 
executions.
Thus, for every liveness property, there are dozens of incarnations of that property obtained simply by varying the underlying fairness assumption.

Not all liveness properties obtained by varying fairness assumptions are semantically distinct.
We identify two key reasons why liveness properties coincide: the (fixed) choice of process model and the choice of liveness property.

In this paper, we fix the process model to be a core synchronous session calculus featuring an internal and external choice~\cite{BHR84,DeNicola1987},
which is frequently studied in the context of session types. 
We also fix the liveness properties to follow a scheme for \textit{lock-freedom}~\cite{Padovani2014,Kobayashi2002}, which has emerged as one of the most important liveness properties for multiparty session calculi and related calculi, such as the linear $\pi$-calculus.
Lock-freedom is essentially the absence of starvation, as described above.
Clearly, the choice of the fairness assumption will influence whether a system is lock-free.

\begin{figure*}
\begin{gather*}
\begin{array}{c}
\left(\N_1 \pipar \N_2\right) \pipar \N_3
\equiv
\N_1 \pipar \left( \N_2 \pipar \N_3 \right)
\qquad
\M \pipar \N \equiv \N \pipar \M
\qquad
\M \pipar 0 \equiv \M
\end{array}
\\[12pt]
\begin{prooftree}
\N \equiv \N'
\quad
\N' \goesto{\alpha} \M'
\quad
\M' \equiv \M
\justifies
\N \goesto{\alpha} \M
\end{prooftree}
\qquad
\begin{prooftree}
k \in I
\justifies
\loc{p}{
\textstyle\bigoplus_{i \in I}\, \send{p_i}{\lambda_i}; \PP_i
}
\pipar
\N
\goesto{\tau}
\loc{p}{
\send{\chosen[3pt]{p_i}}{\lambda_i};\PP_k 
}
\pipar
\N
\end{prooftree}
\\[4pt]
\begin{prooftree}
\justifies
\loc{p}{\rec{X}\PP}
\pipar
\N
 \goesto{\tau}
\loc{p}{ \PP\sub{X}{\rec{X}\PP} }
\pipar
\N
\end{prooftree}
\qquad
\begin{prooftree}
k \in I
\justifies
\loc{p_k}{\chosen[3pt]{\send{q}{\lambda_k} ; \PQ} }
\pipar
\loc{q}{\textstyle{\sum_{i \in I}}\, \recv{p_i}{\lambda_i} ; \PP_i }
\pipar
\N
\goesto{\comm{p_k}{\lambda_k}{q}}
\loc{p_k}{\PQ}
\pipar
\loc{q}{\PP_k}
\pipar
\N
\end{prooftree}
\end{gather*}
\caption{The default semantics for networks that we fix for this study.}\label{fig:red}
\end{figure*}

The restriction to session calculi, for which \emph{session type systems} exist, allows us to answer the following question: 
\begin{quote}
For a given fairness assumption, does there exist a session type system that is sound and/or complete, in the sense that a 
network is lock-free if and/or only if it is \welltyped?
\end{quote}
Our systematic study yields the following main contributions.
\begin{enumerate}
\item 
We classify the notions of lock-freedom that arise by taking every notion of fairness in the survey~\cite{GH19} 
and using them to instantiate a parameter in a general scheme for lock-freedom.
The resulting classification includes classic notions of lock-freedom of session calculi found in the literature. Hence it relates these notions as well.
However, we discover that the notion of lock-freedom which arises from \textit{justness} is new to the literature.

\item We introduce a generalisation of the projection mechanism of global types onto threads, which uses the most general possible merge operator.
This solves the problem that session type systems employing global types
without an explicit parallel composition operator 
are incomplete, in the sense that there are lock-free networks that cannot be typed.
This leads to the following main result.
\item 
We prove that our session type system is 
complete for lock-freedom, when assuming justness.
To the best of our knowledge, this is the first completeness result of this kind.
We delineate the scope of our completeness result by showing that completeness does not hold for 
weaker notions of lock-freedom.

\item 
We prove that more notions of lock-freedom coincide when restricting to race-free networks.
Furthermore, race-free networks are sound for all notions of lock-freedom, whenever we assume at least $\textit{justness}$.
\end{enumerate}
Following \cite{Castellani2019b,Severi2019}, we employ session types that abstract from
  the {concrete }types (e.g.\ \texttt{Bool} or \texttt{Nat}) of messages, using labels $\lambda$ instead.
  As a result, systems and types have a fairly similar syntax. 
It is fairly trivial to move from our session type system with labels to one with data and data  types.
\vspace{1ex}

\paragraph*{Structure of the paper}
Section~\ref{sec:calculus} introduces our session calculus and
a spectrum of fairness assumptions, and then systematically classifies the resulting spectrum of lock-freedom properties. Section~\ref{sec:types} presents
our session type system featuring a general merge operator and guarded types, which we prove to be complete with respect to $\Live{\J}$ -- the notion of lock-freedom arising from the assumption of justness -- for all networks.
Section~\ref{sec:sound} considers race-free networks in order to explore the scope of soundness results.
Section~\ref{sec:related} situates our results with respect to notions of lock-freedom from the literature.

\section{The scope: A session calculus, its key fairness notions and liveness properties}\label{sec:calculus}

In this section,\pagebreak[3] we define the session calculus and a scheme for lock
freedom. We also explain various fairness assumptions and illustrate their differences through separating examples.

\subsection{Syntax and semantics for threads and networks} 

Our session calculus features finitely many recursive \textit{threads} that send and receive messages.
Threads, uniquely identified by location names, feature an internal choice $\bigoplus \send{p_i}{\lambda_i}$ between messages labelled $\lambda_i$ sent to locations  $p_i$ (a choice made at run-time entirely by the sending thread), and an external choice $\sum \recv{p_i}{\lambda_i}$ amongst messages received (meaning that the thread is ready to receive 
different messages $\lambda_i$ from $p_i$, but cannot influence which of them will eventually come through).
\[
\begin{array}{rl}
\PP \coloneqq & \End \\
       \mid & \bigoplus_{i \in I}\, \send{p_i}{\lambda_i}; \PP_i \\
       \mid & \sum_{i \in I}\, \recv{p_i}{\lambda_i}; \PP_i \\
       \mid & X \\
       \mid & \rec{X}\PP
\end{array}
\qquad\qquad
\begin{array}{rl}
\N \coloneqq & \loc{p}{ \PP } \\
        \mid & 0 \\
        \mid & \N \pipar \N
\end{array}
\]
The index sets $I$ are finite, and in the case of $\bigoplus_{i\in I}$ also non-empty.
We enforce guarded recursion by excluding threads of the form $\rec{X}X$ or $\rec{X}\rec{Y}\PP$.
If $\loc{p}\PP$ is a sub-expression of a network $\N$, then $p$ is called a 
\emph{location} of $\N$.\linebreak[3]
In a network $\N$, all locations are required to be distinct and all threads closed, meaning
that each occurrence of a variable $X$ is in the scope of a recursion $\rec{X}{\PP}$.
Moreover, in each sub-expression $\send{p_k}{\lambda_k}$ or $\recv{p_k}{\lambda_k}$, the $p_k$ must be a location of $\N$.
We may elide $\End$; we write
$\send{p_1}{\lambda_1}; \PP_1 \oplus \dots \oplus \send{p_n}{\lambda_n}; \PP_n$
for $\bigoplus_{i \in \{1,\dots,n\}}\send{p_i}{\lambda_i}; \PP_i$, and $\recv{p_1}{\lambda_1}; \PP_1 + \dots + \recv{p_n}{\lambda_n}; \PP_n$
for $\sum_{i \in \{1,\dots,n\}}\send{p_i}{\lambda_i}; \PP_i$.
In particular, we write $\recv{p}{\lambda}; \PP$ in case 
$I$ is a singleton set. We follow a recent trend allowing inputs in an external choice to listen to different locations~\cite{Castellani2019b,Jongmans2020}, which allows us to broaden the scope of our investigation.

\vspace{1ex} 

\subsubsection*{A reduction semantics for our session calculus}

The rules for our session calculus, presented in Figure~\ref{fig:red}, are fairly standard.
In this semantics, an output that a thread has committed to can interact synchronously with some input in an external choice. 
Also, recursion is unfolded by a $\tau$-transition and the standard associativity and commutativity of parallel composition can be applied to enable any transition.

A design decision, we will demonstrate to be significant, is that there is a $\tau$-transition for resolving all internal choices. 
To ensure that singleton internal choices perform only one $\tau$-transition (and not a diverging sequence of $\tau$--transitions),
the transition ends in a \emph{network state} that is not a syntactically valid network.
\emph{Network states} are comprised of located \emph{thread states}, which
due to the annotation $\chosensymb$,
are not necessarily threads themselves.

\subsection{Fairness notions for session calculi \label{sec:fairness}}

We now discuss three fairness assumptions for our session calculus.
A fairness assumption restricts the set of complete execution paths, here simply referred to as \textit{paths}.
\vspace{1ex} 

\begin{definition}{Preliminaries}
A path consists of a network state $\N_0$ and a maximal list of transitions $\N_i \goesto{\alpha_i} \N_{i+1}$, permitted by Figure~\ref{fig:red}.%
\vspace{1ex} 
\end{definition}

Maximality ensures that either the list is infinite or the final network state has no outgoing transition, that is, we restrict ourselves to \textit{complete} execution paths. 

A \emph{fairness notion} $\F$ characterises a subset of all paths as the \emph{fair} ones, modelling
executions that we assume can actually occur; we refer to such paths as \emph{$\F$-fair paths}. 
It is required to satisfy the condition of
\emph{feasibility}~\cite{Apt1988}, saying that each finite prefix of a path is also a prefix of a fair path.
One notion of fairness $\G$ is \emph{stronger} than another one $\F$ -- in symbols $\F\weaker \G$ --
if it rules out more paths as unfair.

A network $\N$ \emph{successfully terminates} under a fairness notion $\F$
iff all fair paths successfully terminate, \ie all components of $\N$ eventually take the form $\loc{p}{\End}$.

A {liveness property}, or more generally a \emph{linear-time property}, is formalised as a property $\varphi$ of paths.
It holds for network state $\N_0$ under a certain fairness assumption iff all fair paths
starting in $\N_0$ satisfy $\varphi$.

\vspace{1ex} 
\header{Strong and weak fairness}
In \cite{GH19}, the concepts of \emph{strong and weak fairness} are parametrised by the notion of a \emph{task}.
What a task is may differ from one notion of fairness to another, but for each task it should be
clear when it is \emph{enabled} in a network state, and when a path \emph{engages} in a task.
A task $T$ is said to be \emph{relentlessly} enabled on a path $\pi$ if each suffix of $\pi$ contains a
network state in which $T$ is enabled; it is \emph{perpetually} enabled if it is enabled in all
network states of $\pi$. A path $\pi$ is \emph{strongly fair} if, for each suffix $\pi'$ of
$\pi$, each task that is relentlessly enabled on $\pi'$ is engaged in by $\pi'$.
It is \emph{weakly fair} if, for each suffix $\pi'$ of $\pi$, each task that is perpetually enabled
on $\pi'$ is engaged in by $\pi'$.

\newcommand{\task}{\mathfrak t}
Given a notion of a task $\task$, the concept of strong fairness S$\task$ is always stronger than its weak counterpart W$\task$, \ie
$\mbox{W}\task \weaker \mbox{S}\task$.

In \cite{GH19},  several notions of fairness found in the literature are characterised through
formalising what constitutes a task.
\emph{Fairness of transitions} is obtained by taking the tasks to be the transitions.
Such a task is enabled in a network state $\N$ if $\N$ is the source state of that transition.
A path $\pi$ engages in a transition  if that transition occurs in $\pi$.
\vspace{1ex} 

\begin{fact}{ST}
\textit{Strong fairness of transitions} ($\St\T$) characterises exactly those paths $\pi$ with the property
that whenever a transition is relentlessly enabled on $\pi$ then
the transition must be taken infinitely often on $\pi$;
it rules out all other paths.
\vspace{1ex} 
\end{fact}

In \cite{GH19}, it is shown that for finite-state systems strong fairness of transitions ($\St\T$) is
the strongest feasible notion of fairness.%
\vspace{1ex} 

\begin{example}{binaryext:early}
Consider the following network where a $\buyer$ chooses to talk to or to buy a product from a $\seller$, after which the order is shipped.
\[
\renewcommand{\extend}{\texttt{talk}}%
\begin{array}{rl}
&
\loc{\buyer}{
 \rec{X}\left(
  \send{\seller}{\extend};X
  \oplus 
  \send{\seller}{\buy}
 \right)
}
\\
\pipar
&
\seller
\mbox{\Large\textlbrackdbl}
 \rec{Y}
  \begin{array}[t]{@{}l@{}}
  \left(\recv{\buyer}{\extend} ; Y \right.
  \\~+\left.\recv{\buyer}{\buy} ; \send{\shipper}{\order}
 \right) \mbox{\Large\textrbrackdbl}\end{array}
\\
\pipar
&
\loc{\shipper}{
  \recv{\seller}{\order} 
}
\end{array}
\]
The network successfully terminates when assuming $\St\T$, for in the only infinite execution the
$\tau$-transition belonging to instruction $\send{\seller}{\buy}$ is relentlessly enabled but
never taken.%
\vspace{1ex} 
\end{example}

A notion of task that figures prominently in the literature is that of a \emph{component}.
A component is one of the prime elements in a parallel composition -- in
a network expression it is completely determined by its location.
Each transition involves either one or two components.
A component is \emph{enabled} in a network state iff a transition involving that component is
enabled; a path \emph{engages} in a component iff it contains a transition that involves that component.

We define a function $\comp$ which returns for a transition the set of components participating in the transition.
Each transition labelled $\tau$ involves exactly one component (location) evident from the rule;
each transition labelled $\comm{p}{\lambda}{q}$ involves exactly two components, $p$ and $q$.
This defines strong and weak fairness of components.
\vspace{1ex} 

\begin{fact}{SC}
\textit{Strong fairness of components} ($\St\C$) characterises the paths $\pi$ 
such that, for any location $p$, if there are transitions involving $p$ relentlessly enabled on $\pi$, then a transition that involves $p$ must be taken infinitely often on $\pi$.
\vspace{1ex} 
\end{fact}

\begin{fact}{WC}
A path $\pi$ satisfies \textit{weak fairness of components} ($\W\C$) whenever, 
for every location $p$, if some transition involving $p$ is, from some
state onwards, perpetually enabled,
then a transition that involves $p$ occurs infinitely often in $\pi$.
\vspace{1ex} 
\end{fact}

Under the fairness assumption $\St\C$, Example~\ref{ex:binaryext:early} does not successfully terminate,
for there is an infinite path where, alternately, the buyer performs a $\tau$-transition to select
the left branch of its choice and then the $\buyer$ and $\seller$ talk to each other. Along this
path there is never a transition enabled that involves the $\shipper$; hence that branch need never be taken.
This illustrates that $\SC$ allows strictly more paths than $\St\T$, \ie $\SC\precneqq\St\T$.
\vspace{1ex} 

\begin{example}{sc-j:early}
To see that $\SC$ excludes some paths, consider the following network.
\[
\begin{array}{rl}
&
\loc{\seller}{
 \rec{X}\left(
  \recv{\buyer1}{\order1};X
  +
  \recv{\buyer2}{\order2}
 \right)
}
\\
\pipar
&
\loc{\buyer1}{
 \rec{Y}
  \send{\seller}{\order1} ; Y
}
\\
\pipar
&
\loc{\buyer2}{
  \send{\seller}{\order2}
}
\end{array}
\]
The above network terminates under $\SC$ (albeit in a state where $\buyer1$ has not successfully terminated), for, in any infinite execution, a
transition from $\buyer2$ is relentlessly enabled but never taken.
It does not need to terminate under weak fairness of components,
for no transition is enabled
perpetually due to the $\tau$-transitions  that  unfold the recursion after each communication.
\vspace{1ex} 
\end{example}

Guaranteeing termination in this example seems wrong as the fairness assumption 
constrains the `free will' of the \seller\ in the sense that they have to sell items to $\buyer2$.
Therefore we will introduce a weaker fairness assumption.
\vspace{1ex} 

\header{Justness}
We consider a minimal notion of fairness that guarantees only that concurrent transitions cannot prevent each other from happening. Informally, two transitions are concurrent if no component is involved in both transitions.%
\vspace{1ex} 

\begin{definition}{concurrent}
Two transitions $t$ and $u$ are \emph{concurrent}, notation $t \conc u$, if
$\comp(t) \cap \comp(u)=\emptyset$.
\vspace{1ex} 
\end{definition}
Justness guarantees that once a transition is enabled that stems from a set of parallel components, one (or more) of
these components will eventually partake in a transition.
\vspace{1ex} 

\begin{definition}{just}
A path $\pi$ is \emph{just} whenever,
for every suffix of $\pi$ beginning with state $s$ and for every transition $t$ enabled in state $s$, some transition $u$ occurs in that suffix such that $t \not\conc u$. 
Equivalently, one might say that no enabled transition is denied forever only by concurrent transitions.
The corresponding fairness assumption, which only allows just paths, is called \emph{justness} (\J).
\vspace{1ex} 
\end{definition}

Example \ref{ex:sc-j:early} illustrates that $\J$ is strictly weaker than $\St\C$, \ie \J\ rules out fewer paths.
While this system terminates under $\St\C$, it does not necessarily terminate under \J, for it allows infinite communication between the \seller\ and $\buyer1$.
Although the transition involving $\buyer2$ is relentlessly enabled, it is not ruled out by justness 
since the $\seller$ is involved in both communications.

Justness is however enough to assume that in our leading example at the top of the introduction,
the two concurrent interactions cannot prevent each other from occurring.%
\vspace{1ex} 

\begin{example}{p-jt:early}
More formally, we can model the scenario described at the top of the introduction as follows.
\[
\begin{array}{rl}
&
\loc{\seller1}{
 \rec{X}
  \recv{\buyer1}{\order};X
}
\\
\pipar
&
\loc{\buyer1}{
 \rec{Y}
  \send{\seller1}{\order} ; Y
}
\\
\pipar
&
\loc{\seller2}{
 \rec{Z}
  \recv{\buyer2}{\order};Z
}
\\
\pipar
&
\loc{\buyer2}{
 \rec{W}
  \send{\seller2}{\order} ; W
}
\end{array}
\]
There is no just path where $\seller2$ and $\buyer2$ never act. Indeed, for any just path all components act infinitely often.
\vspace{1ex} 
\end{example}

In general, $\J \weaker \W\C$ holds~\cite{GH19}.
In addition, for our session calculus, justness coincides with weak fairness of components.%
\vspace{1ex} 

\begin{proposition}{WC}
WC coincides with J.
\end{proposition}
\begin{proof}
Let $\pi$ be an infinite path in our network that is not WC-fair. So, on a suffix of $\pi$, a component $p$ is
perpetually enabled, but never taken. In case $p$ is stuck in a state where its next transition is a
$\tau$, then $\pi$ is not just.

In case $p$ is stuck in a state $\chosen[2pt]{\send{q}{\lambda}}; \PP$,
then, for component $p$ to be perpetually
enabled, $q$ must always be in a state $\sum_{i \in I} \recv{p_i}{\lambda_i}; \PP_i$ with $p=p_k$ and
$\lambda=\lambda_k$ for some $k\in I$. Location $q$ must get stuck in such a state, for if $q$
keeps moving, it will at some point reach a state $\rec{ X } \PQ$, which is not of the above form.
Consequently, $\pi$ is not just.

The remaining case is that $p$ is stuck in a state of the form $\sum_{i \in I} \recv{p_i}{\lambda_i}; \PP_i$.
For component $p$ to be enabled, a component $p_k$ with $k\in I$ must be in a state
$\send{p}{\lambda_k}; \PP$. Again it follows that $\pi$ is not just.
\end{proof}
As we will observe later, under a different choice of semantics of our session calculus,
  $\J$ and $\W\C$ do not coincide.

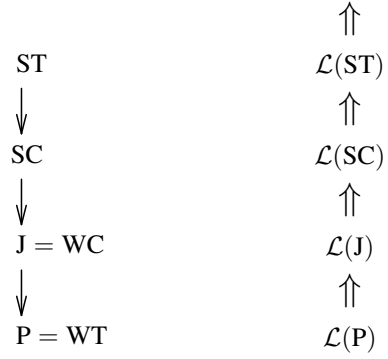
\begin{figure}[t]
\vspace{-8ex}
\hfill\begin{minipage}{0.12\textwidth}
\input{lattice2b}
\centerline{\raisebox{-14ex}{\box\graph}}
\end{minipage}
\hfill\hfill
\begin{minipage}{0.12\textwidth}
\input{lattice3b}
\centerline{\raisebox{-8ex}{\box\graph}}
\end{minipage}
\hfill{}
\vspace{1ex}
\caption{A classification for our session calculus
of fairness assumptions and liveness properties.}
\vspace*{-11pt}
\label{reduced taxonomy}
\end{figure}
\vspace{1ex} 

\header{Further notions of fairness}
If we define \textit{weak fairness of transitions} ($\W\T$), where, as for $\St\T$, the tasks are the individual transitions, then $\W\T$ imposes no restrictions on the completed traces for our languages.
To see why, observe that in any infinite path, no transition is enabled
perpetually due to the $\tau$-transitions for unfolding recursion. This most liberal fairness
assumption, which we denote $\Pr$,\footnote{\textit{On terminology.}
In related work~\cite{GH19}, $\Pr$ stands for ``progress'', the assumption  that a system cannot
spontaneously halt as long as it is neither deadlocked nor successfully terminated. However, the
word ``progress'' is heavily overloaded,  meaning anything from
deadlock-freedom~\cite{Honda2016,Dezani2006} and
lock-freedom~\cite{Dezani2005,Padovani2014,Najm1999} to other liveness properties~\cite{Misra2001},
such as weak and strong normalisation. That means, it refers to some desired property rather than an assumption on paths.
Furthermore, there are related liveness properties such as \textit{global progress} that concerns delegation~\cite{Coppo2016}.
}
only guarantees that the system as a whole will progress if some transition is enabled.

The survey~\cite{GH19} classifies $21$ different notions of fairness, covering all common notions found in the literature.
In our session calculus, many of these notions coincide, so that only $7$ different notions of fairness remain; 
see Appendix\arxiv{~\ref{app:A}}{~A of \cite{GHH21}}.

Here we have presented those that we found to be the most important notions for session calculi -- summarised in 
Figure~\ref{reduced taxonomy}.
Notably, there are strong fairness assumptions strictly between $\St\C$ and $\St\T$.
However, 
every fairness assumption from \cite{GH19} leads to a notion of lock-freedom that coincides with
one based on a fairness assumption defined in this section\arxiv{ (see Appendix~\ref{app:B})}{\cite{GHH21}}.

\subsection{A scheme for lock-freedom\label{ssec:lockfreedom}}
As discussed, a fairness assumption rules out certain paths
for given systems. As lock-freedom considers only the paths of a system
  that can actually be taken,
it depends on the underlying fairness assumption.
Hence, a \emph{scheme for lock-freedom} reads as follows:
\begin{equation}\label{lock scheme}
\mbox{\begin{minipage}{3.1in}
Along any $\F$-fair path, if a component has not successfully terminated, then it must eventually do something.
\end{minipage}}
\end{equation}

We can now formally define our scheme for lock-freedom with respect to a fairness assumption $\F$.
\vspace{1ex} 

\begin{definition}{live}
Let $\F$ be a fairness assumption.
A network $\N$ satisfies liveness property $\Live{\F}$ (for short $\N\models\Live\F$) if, for each $\F$-fair
path $\pi$ starting in $\N$ and each location $p$ of $\N$,
\begin{itemize}
\item either $p$ successfully terminates on $\pi$, or
\item $\pi$ contains infinitely many transitions involving $p$.
\end{itemize}
\end{definition}
Remember that a location $p$ successfully terminates when it is of the form $\loc{p}{\End}$.
The letter $\mathcal{L}$ indicates ``liveness'' or ``lock-freedom''.

We say $\Live\F$ is \emph{stronger} than $\Live\G$, denoted by $\Live\F\Rightarrow\Live\G$, if
$\N\models\Live{\F}\Rightarrow\N\models\Live{\G}$, for all $\N$. It is \emph{strictly stronger}
if moreover $\Live{\G}\not\Rightarrow\Live{\F}$.
In case a fairness assumption $\G$ is stronger than $\F$, then $\Live{\G}$ is a weaker property than $\Live{\F}$.
\vspace{1ex} 

\begin{proposition}{prop2}
$\F\weaker\G$ implies $\Live{\F}\Rightarrow\Live{\G}$,
for fairness assumptions $\F$ and $\G$.
\vspace{1ex} 
\end{proposition}

Intuitively, any path of $\N$ that is lock-free under $\F$ will also be lock-free under \G.
Since $\G$ rules out more paths than $\F$ and since $\mathcal L$ is defined over paths, 
the proof is obvious. 

A network has a \emph{deadlock} (state) if there exists a reachable network state without outgoing transitions that is not successfully
terminated; a network is \emph{deadlock-free} if it does not have a deadlock.

Clearly, $\Live{\St\T}$ implies deadlock-freedom, since every finite execution can be extended to some $\St\T$-fair path, 
using  feasibility.
In networks consisting of one or two parties only, deadlock-freedom coincides with all notions of lock-freedom.
Deadlock-freedom, however, is considered to be insufficient for networks with three or more locations, as those networks may experience starvation:
\emph{starvation} occurs when there is an execution path along which some component wants to perform a task but no task involving that component occurs.

Using the relationship between $\Live{\St\T}$ and deadlock-free\-dom, as well as Proposition~\ref{pr:prop2}, yields the classification of liveness properties on the right-hand side of Figure~\ref{reduced taxonomy}.
Since $\F\precneqq\G$ does not imply that $\Live{\F}$ is strictly stronger than $\Live{\G}$,
we provide separating examples to prove that 
the presented notions of lock-freedom are different.
\vspace{1ex} 

\header{\texorpdfstring{$\Live{\St\T}$}{L(ST)} is strictly stronger than deadlock-freedom} 
\mbox{}\vspace{1ex} 

\begin{example}{livelock}
Consider the following network, where a $\buyer$ purchases goods repeatedly from a $\seller$, 
while a \shipper\ is awaiting an order that is never placed.\pagebreak[3]
\[
\begin{array}{rl}
&
\loc{\buyer}{
 \rec{X}
  \send{\seller}{\buy};X
}
\\
\pipar
&
\loc{\seller}{
 \rec{Y}
  \recv{\buyer}{\buy} ; Y
}
\\
\pipar
&
\loc{\shipper}{
 \recv{\seller}{\order}
}
\end{array}
\]
This network is deadlock-free, for the buyer and seller can always interact;
it does not satisfy $\Live{\St\T}$ as $\shipper$ is not in state $\End$ and is never involved in a transition. 
\end{example}
\vspace{1ex} 

\header{\texorpdfstring{$\Live{\St\C}$}{L(SC)} is strictly stronger than \texorpdfstring{$\Live{\St\T}$}{L(ST)}}
 Consider Example \ref{ex:binaryext:early}. 
We have seen that all $\St\T$-fair paths successfully terminate. In particular, along all fair paths the $\shipper$ performs a transition.
In contrast, there is an infinite \St\C-fair path where the $\shipper$ neither makes a transition nor successfully terminates.

The following example separates $\St\T$ from $\St\C$ without considering termination.
\vspace{1ex} 

\begin{example}{Fu}
Consider the following network, where a $\buyer$ talks alternatingly to two \seller{s}, but talks to each seller for as long as they desire.
\[
\newcommand{\talk}{\texttt{talk}}
\begin{array}{rl}
&
\buyer
\mbox{\Large\textlbrackdbl}
\begin{array}[t]{@{}l}
 \rec{X} 
\begin{array}[t]{@{}l}
 \left(\recv{\seller1}{\talk}; X\right.
   \\
  ~+
\begin{array}[t]{@{}l}
\recv{\seller1}{\extend}; 
\\
  \rec{Z}    
\begin{array}[t]{@{}l}
   \left( \recv{\seller2}{\talk}; Z\right.
\\    
    +\left.
   \left.\recv{\seller2}{\extend}; X\right)\right)\mbox{\Large\textrbrackdbl}
\end{array}
\end{array}
\end{array}
\end{array}
\\
\pipar
&
\loc{\seller1}{
 \rec{V}(
  \send{\buyer}{\talk} ; V \oplus \send{\buyer}{\extend} ; V~)
}
\\
\pipar
&
\loc{\seller2}{
 \rec{W}(
  \send{\buyer}{\talk} ; W \oplus \send{\buyer}{\extend} ; W~)
}
\end{array}
\]
The above network satisfies $\Live{\St\T}$ but not $\Live{\St\C}$, since no location terminates and there are $\St\C$-fair paths on which one of $\seller1$ or $\seller2$ ceases to act, violating the condition that there must be infinitely many transitions stemming from them.
\end{example}
\vspace{1ex} 

\header{\texorpdfstring{\,$\Live{\J}$}{L(J)} is strictly stronger than \texorpdfstring{$\Live{\St\C}$}{L(SC)}}
We consider a variant of Example~\ref{ex:sc-j:early}. 
\vspace{1ex} 

\begin{example}{sc-j:2}\vspace{-1ex}
\[
\begin{array}{@{}rl@{}}
&
\loc{\seller}{
 \rec{X}\left(
  \recv{\buyer1}{\order1};X
  +
  \recv{\buyer2}{\order2};X
 \right)}
\\
\pipar
&
\loc{\buyer1}{
 \rec{Y}
  \send{\seller}{\order1} ; Y
}
\\
\pipar
&
\loc{\buyer2}{
 \rec{Z}
  \send{\seller}{\order2} ; Z
}
\end{array}
\]
The above network satisfies $\Live{\St\C}$, since each location $p$ has a relentlessly enabled communication transition. Hence, $p$ will engage in a communication transition infinitely often.
However, the system does not satisfy $\Live{\J}$,
since there is a $\J$-path where $\buyer1$ never acts. Namely, every communication of $\buyer1$ may be preempted by a communication of $\buyer2$, as both buyers communicate with the same $\seller$. 
\vspace{1ex} 
\end{example}

Although Example~\ref{ex:sc-j:early} separates $\J$ from $\St\C$, we cannot use it as separating 
example for $\Live\J$ and $\Live{\St\C}$. It does not even satisfy $\Live{\St\T}$, since if $\buyer2$ ever acts, then $\buyer1$ never successfully terminates nor engages in infinitely many transitions.
\vspace{1ex} 

\header{\texorpdfstring{$\;\Live{\Pr}$}{L(P)} is strictly stronger than \texorpdfstring{$\Live{\J}$}{L(J)}}
The network of Example \ref{ex:p-jt:early} -- the example from the introduction -- satisfies $\Live{\J}$, since on a just path there are infinitely many transitions stemming from each location.
However, it does not satisfy $\Live{\Pr}$, since there exists a path where two components talk forever, to the exclusion of the other two. 
This example indicates (again) that $\J$ is the minimal realistic fairness assumption.

\begin{figure*}
\begin{gather*}
\begin{prooftree}
\N \equiv \N'
\quad \N' \dgoesto{\alpha} \M'
\quad
\M' \equiv \M
\justifies
\N \dgoesto{\alpha} \M
\end{prooftree}
\qquad\qquad
\begin{prooftree}
\loc{p}{ \PP\sub{X}{\rec{X}\PP} }
\pipar
\N
 \dgoesto{\alpha}
\loc{p}{ \PQ }
\pipar
\N
\justifies
\loc{p}{\rec{X}\PP}
\pipar
\N
 \dgoesto{\alpha}
\loc{p}{ \PQ }
\pipar
\N
\end{prooftree}
\qquad
\\[2ex]
\begin{prooftree}
j \in H \quad k \in I \quad \lambda_k = \lambda_j
\justifies
\loc{p_k}{\textstyle{\bigoplus_{h \in H}}\, \send{q_h}{\lambda_h} ; \PQ_h }
\pipar
\loc{q_j}{\textstyle{\sum_{i \in I}}\, \recv{p_i}{\lambda_i} ; \PP_i }
\pipar
\N
\dgoesto{\comm{p_k}{\lambda_k}{q_j}~}
\loc{p_k}{\PQ_{j}}
\pipar
\loc{q_j}{\PP_k}
\pipar
\N
\end{prooftree}
\end{gather*}
\caption{A reactive semantics without $\tau$-transitions for internal choice or recursion. The definition of $\equiv$ is unchanged.}\label{fig:reactive}
\end{figure*}

\arxiv{In Appendix~\ref{app:B}}{In~\cite{GHH21}}, we analyse further notions of lock-freedom, 
based on other fairness assumptions.

\subsection{Lock-freedom in the literature}

There are two prevalent notions of lock-freedom in the literature, 
which we call Kobayashi lock-freedom and Padovani lock-freedom, acknowledging the authors of key papers where these properties are investigated.
We prove that these two notions relate to $\Live{\St\C}$ and $\Live{\St\T}$, respectively. 
We believe, however, that $\Live{\J}$ is a novel notion of lock-freedom. 
In Section~\ref{sec:related} we discuss further notions.
\vspace{1ex} 

\header{Kobayashi lock-freedom}\label{sec:Kobayashi}

Our scheme (\ref{lock scheme}) for lock-freedom is inspired by a scheme proposed by Kobayashi~\cite{Kobayashi2002} in the setting of the linear $\pi$-calculus, which does not feature operators for choice. 
Our scheme is more general, making it applicable to several calculi. 

Although Kobayashi argues that lock-freedom is para\-metrised by a fairness assumption, 
he settles for exactly one,  called \emph{strong fairness} and attributed to \cite{CS87,Emerson90},
with the stated intention that: ``every process that is able to participate in a communication infinitely often can eventually participate in a communication.''
The intended fairness assumption in~\cite{Kobayashi2002} coincides with $\St\C$.
Almost the same can be said for the formalisation of strong fairness in~\cite{Kobayashi2002}, although
literally speaking the latter is slightly weaker.\footnote{The reason is that Kobayashi's intended requirement that a
component must act is formalised by describing the states right before and right after that component
acts, and stipulating that one must go from the former to the latter. However, in
\cite{Kobayashi2002} there is no way to unambiguously project global states on individual components,
and one can make the prescribed transition without actually involving that component.}
\vspace{1ex} 

\header{\texorpdfstring{\!}{}Padovani lock-freedom coincides with \texorpdfstring{$\Live{\St\T}$}{L(ST)}}
Padovani \cite{Padovani2014} presents a notion of lock-freedom that does not refer explicitly to a fairness assumption. 
Below we use the abbreviation $\proc(p,\N)$ that denotes the unique thread state
$\PP$ such that $\N \equiv \loc{p}{\PP} \pipar \N'$, if $p$ is a location of a network state $\N$.
\vspace{1ex} 

\begin{definition}{lock-free}
$\N$ is \emph{Padovani lock-free} if for each reachable state $\M$ of $\N$, and for each location $p$ of $\M$
  such that $\proc(p,\M)\neq\End$, network $\M$ has an execution path that contains a
  transition involving $p$.
\vspace{1ex} 
\end{definition}
\begin{theorem}{lock-free}
A network is Padovani lock-free iff it satisfies $\Live{\St\T}$. [See Appendix\arxiv{~\ref{app:C}}{~C of \cite{GHH21} } for the proof.]
\end{theorem}

\subsection{Lock-freedom for a reactive semantics}\label{sec:reactive}
This section demonstrates that differences between session calculi, which may appear to be merely stylistic, in fact impact the resulting notions of liveness. 
An alternative semantics, (\eg~\cite{Severi2019,Castellani2019b}), which we call \textit{reactive semantics}, is given in Figure~\ref{fig:reactive}. 
Here, neither unfolding recursion nor making a choice between various send actions induces a $\tau$-transition. 

In Definition~\ref{df:live}, we formally introduced liveness properties for a network, parametrised
by a fairness assumption. 
In fact, the definition also depends on the given semantics. In the remainder, we denote by $\Live{\F}$ a liveness 
property with regard to the semantics of Figure~\ref{fig:red}, and by $\React{$\F$}$ a liveness property 
with regard to the reactive semantics.%
\vspace{1ex} 

\begin{example}{R versus L}
The following network has a deadlock by the default semantics of Figure~\ref{fig:red}. Consequently,
it satisfies none of the properties $\Live{\F}$. Yet, it satisfies $\React{$\F$}$,  for each $\F$.
\[
\begin{array}{rl}
&
\loc{\buyer}{
  \send{\seller}{\buy} \oplus \send{\seller}{\order}
}
\\
\pipar
&
\loc{\seller}{
  \recv{\buyer}{\buy}
}
\end{array}
\]
\end{example}

A similar result to Proposition~\ref{pr:prop2} shows that the 
strength of a fairness assumption partially determines the strength of 
the corresponding liveness property.
\vspace{1ex} 

\begin{proposition}{prop2reactive}
$\F\weaker\G$ implies $\React{\F}\Rightarrow\React{\G}$,
for fairness assumptions $\F$ and $\G$.
\vspace{1ex} 
\end{proposition}

Consequently, a classification of the liveness properties $\React{\F}$, for $\F$ any of the
  fairness assumptions from \cite{GH19}, can be obtained from the classification of these  fairness
  properties (Figure \arxiv{\ref{full taxonomy} in Appendix~\ref{app:A}}{6 in \cite{GHH21}}) by collapsing certain entries,
  just as for the classification of liveness properties $\Live{\F}$ from Figure~\ref{reduced taxonomy}.
  Since the separating examples given for $\Live{\F}$ apply also to $\React{\F}$, we end up with at least
  four different notions $\React{\F}$. However, we expect a lattice that is quite a bit larger, with
  fewer notions coinciding.

As an instance of this,  $\React{\J}$ is strictly stronger than $\React{\W\C}$.
Strictness is shown by the following example.
\vspace{1ex} 

\begin{example}{reactive:early}
The following network presents a \buyer\ who negotiates with $\seller1$ up to a point 
and then decides to order a product with $\seller2$ and inform $\seller1$ about their decision.
\[
\begin{array}[t]{rl}
&
\buyer\mbox{\Large\textlbrackdbl}
 \rec{X} (
  \begin{array}[t]{@{}l}
  \send{\seller1}{\nego} ; X
  \\
  \oplus ~
  \send{\seller2}{\order}; \send{\seller1}{\done})\mbox{\Large\textrbrackdbl}
  \end{array}
\\
\pipar
&
\loc{\seller1}{
 \rec{Y} \left(
  \recv{\buyer}{\nego} ;  Y
  +
  \recv{\buyer}{\done} 
  \right)
}
\\
\pipar
&
\loc{\seller2}{
 \recv{\buyer}{\order}
}
\end{array}
\]
The network successfully terminates under $\React{\W\C}$,
for a transition involving $\seller2$ is perpetually enabled, when appealing to Figure~\ref{fig:reactive}.
It does not need to terminate under justness as the \buyer\ is involved in all transitions.
\vspace{1ex} 
\end{example}

Similar to Example~\ref{ex:sc-j:early}, guaranteeing termination in this example seems wrong as the fairness assumption 
constrains the \buyer's `free will'. Therefore,
the presented results suggest that $\J$ is a more realistic notion than $\W\C$.

\section{Session types and completeness}\label{sec:types}
We now focus on session type systems.
A suitably crafted session type system guarantees liveness properties for a network, 
if the network is \welltyped.
We devise a session type system that is complete for $\Live{\J}$, 
meaning that all lock-free networks can be typed.

\subsection{Global session types, projections and type judgements}

We build on a widely-adopted approach for multiparty session types.
It first defines a global type, describing the interacting behaviour of all parties involved.
In our syntax for global types, communications of the form $\comm{p}{\lambda}{q}$ describe the sending of a message labelled $\lambda$ from location $p$ to $q$, and $\bigboxplus$ indicates a choice over a finite, non-empty index set $I$.
\[
\begin{array}{rlr}
\G \coloneqq & \End & \mbox{(successful termination)} \\
        \mid & \bigboxplus_{i\in I}\ \comm{p}{\lambda_i}{q_i} ; \G_i  & \mbox{(choice of communication)} \\
        \mid & X & \mbox{(recursion variable)} \\
	\mid & \rec{X} \G  & \mbox{(recursion)}
\end{array}
\]
As for our session type calculus we 
exclude types of the form $\rec{X}X$ or $\rec{X}\rec{Y}\G$ to enforce guarded recursion.
Moreover, for $\bigboxplus_{i\in I}\, \comm{p}{\lambda_i}{q_i} ; \G_i$, we assume $p \not= q_i$ for
all $i \mathbin\in I$. That means locations cannot send messages to themselves.
A global type is \emph{closed} whenever it contains no free recursion variables.
The fact that $p$ is the same in every branch of a choice means there is a distinguished choice leader $p$, who makes that choice, but there may be different recipients, as in related work on flexible choices~\cite{Castellani2019b}.

A global session type can be projected to a local view for each location.
We call local types stemming from projections \emph{projection types}.
They are defined almost in the same way as threads of \sect{calculus}:
instead of the construct $\sum_{i \in I} \recv{p_i}{\lambda_i}; \PP_i$ they feature merely its unary case $\recv{p}{\lambda}; \PP$, as well as the \emph{merge} operators $\merge_{i\in I} \PP_i$ .

We define the set of \emph{{participants}}
 of a global type $\G$ recursively:\vspace{-2ex} 
\begin{align*}
\participants{\End} &= \participants{X} = \emptyset 
\\
\participants{\rec{ X } \G} &= \participants{\G} \\
\participants{\textstyle\bigboxplus_{i\in I} \comm{p}{\lambda_i}{q_i} ; \G_i} &=
  {\textstyle\bigcup_{i \in I}}\, \{p,q_i\} \cup \participants{\G_i}
\end{align*}
Given a global session type $\G$ and location $p$,
we define the projection $\proj{p}{\G}$ of $\G$ on $p$ as follows.
\label{def:proj}
\begin{align*}
\proj{p}{ \End } &= \End
\qquad
\proj{p}{ X } = X
\\[2pt]
\proj{p}{ (\rec{X} \G) } &= 
\left\{
\begin{array}{@{}l@{~~}l@{}}
 \End & \begin{array}[t]{@{}l@{}}\mbox{if $p \notin \participants{\G}$}\\
      \mbox{and $\rec{X}\G$ is closed}
      \end{array}
\\
 \rec{X} \left( \proj{p}{ \G }  \right) & \mbox{otherwise}
\\
\end{array}
 \right.
\\[2pt]
\proj{r}{ (\textstyle\bigboxplus_{i\in I}\, \comm{p}{\lambda_i}{q_i} ; \G_i) }
&=
\left\{
\begin{array}{@{}lr@{}}
 \bigoplus_{i \in I} \proj{r}{ (\comm{p}{\lambda_i}{q_i} ; \G_i) }
 & p\mathop= r
 \\[2pt]
 \merge_{i\in I} \proj{r}{ (\comm{p}{\lambda_i}{q_i} ;\G_i) }
 & p\mathop{\not=} r
\end{array}
\right.
\\[2pt]
\proj{r}{(\comm{p}{\lambda}{q} ; \G) }
&=
\left\{
\begin{array}{@{}ll@{}}
 \send{q}{\lambda} ; \left( \proj{r}{ \G } \right)
 & p \mathop= r \\
 \recv{p}{\lambda} ; \left( \proj{r}{ \G } \right)
 & q \mathop= r \\
 \proj{r}{\G}
 & r \mathop{\not\in} \left\{ p, q \right\}
\end{array}
\right.
\end{align*}
The merge operator is interpreted directly through
the judgement relation $\vdash$ between threads and projection types, coinductively defined in Figure~\ref{fig:types}. 
See~\cite{coinduction} for a formal definition of what it means to interpret such rules coinductively. 
Usually, the merge is defined independently from the type judgements;
it is simply an operation that builds a single type from several types, without using an explicit merge primitive.
In the standard approach~\cite{Yoshida2020}, the work of our judgement relation $\vdash$
is split between (a) the aforementioned merge operation, 
(b) a subtyping relation $\leq$ between types~\cite[Definition~6]{Yoshida2020}, 
and (c) a relation $\vdash$ between threads and local session types \cite[Figure~5]{Yoshida2020}. 
Our use of merge as a primitive construct for generating projection types, interpreted through $\vdash$, 
makes merging as general as possible.

\begin{figure}
\begin{gather*}
\begin{prooftree}
\PP\sub{X}{\rec{X}\PP} \vdash \PQ
\justifies
\rec{X}\PP \vdash \PQ
\end{prooftree}
\qquad
\begin{prooftree}
\PP \vdash \PQ\sub{X}{\rec{X}\PQ}
\justifies
\PP \vdash \rec{X}\PQ
\end{prooftree}
\\[8pt]
\begin{prooftree}
\justifies
\End \vdash \End
\end{prooftree}
\qquad
\begin{prooftree}
i \in I
\quad
\PP_i \vdash \PQ_i
\justifies
\textstyle{\sum_{i\in I}\, \recv{p_i}{\lambda_i} ; \PP_i}
\vdash
\recv{p_i}{\lambda_i} ; \PQ_i
\end{prooftree}
\\[8pt]
\begin{prooftree}
I \subseteq J
\quad
\forall i \in I
\quad
\PP_i \vdash \PQ_i
\justifies
\textstyle{\bigoplus_{i\in I}\, \send{p_i}{\lambda_i} ; \PP_i}
\vdash
\textstyle{\bigoplus_{i\in J}\, \send{p_i}{\lambda_i} ; \PQ_i}
\end{prooftree}
\qquad
\begin{prooftree}
\forall i \in I
\quad
\PP \vdash \PQ_i
\justifies
\textstyle{\PP}
\vdash
\textstyle{\merge_{i\in I}\, \PQ_i}
\end{prooftree}
\end{gather*}
\caption{Typing judgements relating threads to projection types.
\label{fig:types}}
\end{figure}

\subsection{Well-typed networks}

The following definition plays the role of a type rule assigning a global type to a network in related systems, 
e.g.,~\cite{Yoshida2020,Severi2019,Denielou2012,Castellani2019b}.
\vspace{1ex} 

\begin{definition}{typing}
\mbox{A network $\N\mathbin=\loc{p_1}{\!\PP_1\!}\mathop\| \loc{p_2}{\!\PP_2\!} \mathop\| \!\dots\! \mathop\| \loc{p_n}{\!\PP_n\!}$}
is \welltyped with respect to a global type $\G$, denoted $\N\vdash\G$,  if $\G$ is closed, 
$\participants{\G} \mathop\subseteq \left\{ p_1, p_2, \hdots, p_n \right\}$,
and $\PP_i \vdash \proj{p_i}{\G}$ for all $i$.

A network $\N$ is \emph{\welltyped} if $\N \mathbin\vdash \G$ for some global type $\G$\hspace{-1pt}.%
\vspace{1ex} 
\end{definition}

\begin{example}{typing}
\renewcommand{\extend}{\texttt{talk}}%
A global type for the network of Example~\ref{ex:binaryext:early}~is
\[
\G = 
\rec{ X } ( 
\begin{array}[t]{@{}l}
\comm{\buyer}{\extend}{\seller}; X
\\
\boxplus
\begin{array}[t]{l}
\comm{\buyer}{\buy}{\seller}; 
\\
\comm{\seller}{\order}{\shipper}; \End).
\end{array}
\end{array}
\]
We have
\begin{align*}
\proj{\buyer}{\G} &= \rec{ X } (\send{\seller}{\extend}; X \oplus \send{\seller}{\buy}; \End) \\
\proj{\seller}{\G} &=  \rec{ X } ( 
 \begin{array}[t]{@{}l}
  \recv{\buyer}{\extend}; X \\
  \sqcap~\recv{\buyer}{\buy}; \send{\shipper}{\order}; \End) 
 \end{array}
\\
\proj{\shipper}{\G} &= \rec{ X } (X \sqcap\recv{\seller}{\order}; \End). 
\end{align*}
With the help of the rules of Figure~\ref{fig:types}, we can derive the following facts, 
using proofs that are not well-founded.
\[\begin{array}{@{}r@{}l@{\ }l@{}}
\rec{ X } (&\send{\seller}{\extend}; X \oplus \send{\seller}{\buy}; \End)
&\vdash
\proj{\buyer}{\G}\\
\rec{Y} (&%
  \recv{\buyer}{\extend} ; Y&\\
  &+~
  \recv{\buyer}{\buy} ; \send{\shipper}{\order};\End
 ) &\vdash \proj{\seller}{\G}\\
\multicolumn{2}{@{}l}{\ensuremath{\recv{\seller}{\order}; \End}} 
&\vdash \proj{\shipper}{\G} 
\end{array}\]
\end{example}
This network is \welltyped.
However, in the literature, it is commonly regarded as not \welltyped, which 
may be due to
the unguarded recursion in $\proj{\shipper}{\G}$
\vspace{1ex} 

\begin{example}{unsound}
\renewcommand{\extend}{\texttt{talk}}%
This network is a restriction of the previous example, where no message is sent to the $\shipper$ on any path.
\[
\begin{array}{rl}
&
\loc{\buyer}{
 \rec{X}\left(
  \send{\seller}{\extend};X
 \right)}
\\
\pipar
&
\seller
\mbox{\Large\textlbrackdbl}
 \rec{Y}(
  \begin{array}[t]{@{}l@{}}
  \recv{\buyer}{\extend} ; Y 
  \\ + ~
  \recv{\buyer}{\buy} ; \send{\shipper}{\order}; \End
 )\mbox{\Large\textrbrackdbl}\end{array}
\\
\pipar
&
\loc{\shipper}{
  \recv{\seller}{\order} ; \End
}
\end{array}
\]
By using the same global type $\G$, we can type the network. 
The following judgement makes use of the rule for internal choice in Figure~\ref{fig:types}, which permits deleting branches, as for most session subtype relations in the literature~\cite{Gay2005,Demangeon2011}.
\\[2mm]
\centerline{$
\rec{X}
  \send{\seller}{\extend};X \vdash \proj{\buyer}{\G}
$}\\[2mm]
The projections to the other locations are the same as in Example~\ref{ex:typing}.

This network is \welltyped and deadlock-free, 
but is not lock-free under any fairness assumption.
That is, this type system is unsound for any notion of lock-freedom we have discussed.
\vspace{1ex} 
\end{example}

Any type system targeting some notion of lock-freedom presented must reject Example~\ref{ex:unsound}.
In this paper, our design decision to ensure soundness is to require that recursion
has to be guarded for projections.\footnote{An alternative design decision for strengthening the
  type system, that we do not pursue here, could be to restrict the type rule for internal
  choice (Figure~\ref{fig:types}) to prevent branches from being deleted (\cf \cite{Castellani2019b}),
  which, combined with our general merge, would allow Example~\ref{ex:typing} to stay in the fold 
  for $\Live{\St\T}$.
}
We thereby disallow the projection $\proj{\shipper}{\G}$ in the above examples, 
thereby rejecting the networks in Examples~\ref{ex:typing} and~\ref{ex:unsound}.
Since Example~\ref{ex:typing} satisfies $\Live{\St\T}$, we 
have to aim for a stronger notion of lock-freedom.
It will be lock-freedom under justness and we will 
prove that our session type system is 
complete for that type of lock-freedom.

\subsection{Guarded type judgements}

We define a variant of well-typedness (Definition~\ref{df:typing}) that enforces each projection type to be guarded.
A projection type $\PP$ is \emph{guarded} iff each occurrence of a variable $X$ within a
subexpression $\rec{X} \PQ$ of $\PP$ occurs within a subexpression  
$\send{p}{\lambda}; \PP$ or $\recv{p}{\lambda}; \PP$.
\vspace{1ex} 

\begin{definition}{guarded typing}
A network $\N$
is \emph{guardedly well-typed} with respect to a global type $\G$, denoted $\N\vdashg\G$,  if 
$\N  \vdash \G$ and all projections $\proj{p}{\G}$ are guarded.
\vspace{1ex} 
\end{definition}

Note that $\G$ is guarded by definition, but this is not sufficient to ensure that $\proj{p}{\G}$ is guarded. 
Examples~\ref{ex:typing} and~\ref{ex:unsound} are \welltyped, but not guardedly well-typed.
\vspace{1ex} 

\begin{example}{SC-typed}
Example~\ref{ex:sc-j:2}, which features a competition between two buyers,
is guardedly well-typed with respect to the following global type.
\[
\G = 
 \rec{X}(
\begin{array}[t]{@{}l@{}}
 \comm{buyer1}{\order1}{\seller} ; 
\\
 \comm{buyer2}{\order2}{\seller} ; X)
\end{array}
\]
All projections are guarded. 
Indeed, any global type without a choice will lead to guarded projections.
The interesting projection relates the thread for the $\seller$ to the projection of the $\seller$.
\begin{eqnarray*}
& \proj{\seller}{\G}
=
 \rec{X}\left(
  \recv{\buyer1}{\order1};
  \recv{\buyer2}{\order2};X
 \right)
\\
 &\rec{X}\left(
  \recv{\buyer1}{\order1};X
  {+}
  \recv{\buyer2}{\order2};X
 \right)
\mathbin\vdash
 \proj{\seller}{\G}
\end{eqnarray*}
The above judgement holds by unfolding the recursions so as to appeal twice to the rule for $\sum$ in Figure~\ref{fig:types}.
\vspace{1ex} 
\end{example}

The following example illustrates that our type system cannot be complete for $\Live{\St\C}$.
\vspace{1ex} 

{
\newcommand{\lambdai}{{a}}
\newcommand{\lambdaii}{{b}}
\newcommand{\lambdaiii}{{c}}
\newcommand{\lambdaiv}{{d}}
\begin{example}{not well typed}
The next network satisfies $\Live{\St\C}$, but not $\Live{\J}$.
\[
\begin{array}{rl}
& \loc{p}{ \rec{X} \left( \send{q}{\lambdai}; X \oplus \send{q}{\lambdaii} ; X \right) } \\
\pipar &
\loc{q}{ \rec{Y} \big( \recv{p}{\lambdai} ; Y + \recv{r}{\lambdaiii} ; \left( \recv{r}{\lambdaiv} ; Y
  + \recv{p}{\lambdaii} ; \recv{r}{\lambdaiv} ; Y \right) \big) }  \\
\pipar &
\loc{r}{ \rec{Z}  \send{q}{\lambdaiii}; \send{q}{\lambdaiv} ; Z}
\end{array}
\]
The network does not satisfies $\Live{\J}$ for there is an infinite just path in which locations $p$ and $q$ constantly communicate via $\comm{p}{\lambdai}{q}$ and $r$ never engages in a communication.
That just path is not a $\St\C$-fair, since the communication $\comm{r}{\lambdaiii}{q}$ involving location $r$ is relentlessly enabled yet never taken. 

The network is not \welltyped, let alone guardedly well-typed,
for each global type must have a subexpression
$\comm{p}{\lambdai}{q}; \G_1 \boxplus \comm{p}{\lambdaii}{q}; \G_2$\;,
and hence must have a reachable state $\M$ in which both transitions $\M \dgoesto{\comm{p}{\lambdai}{q}}$
and $\M \dgoesto{\comm{p}{\lambdaii}{q}~}$ are enabled. Yet there is no such reachable state.%
\vspace{1ex} 
\end{example}
}

\begin{figure*}
\[
\gt(h,\M) \mathbin= \left\{ \begin{array}{@{}l@{~~}l@{}} 
  \gt(h,\M') & \mbox{if $\M \goesto\tau \M'$ for a network $\M'$,  }\\
  \End & \mbox{if $\proc(p,\M) = \End$ for each location $p$ of $\M$,}\\
  \mbox{{\sc deadlock}} & \mbox{if no location is ready in $\M$,}\\
  X_\M & \mbox{if $\M$ occurs in $h$ and $h\upharpoonright \M$ is complete for $\M$,}\\
  \bigboxplus_{i\in I}\, \comm{p}{\lambda_i}{q_i} ; \gt(h_i,\M^p_i)     & 
     \mbox{if $\M$ occurs in $h$, }
     \mbox{$p \mathbin= \ch(h,\M)$ and $\proc(p,\M) = \bigoplus_{i\in I}\send{q_i}{\lambda_i} ; \PP_i$,}\\
  \rec{ X_\M } \bigboxplus_{i\in I} \comm{p}{\lambda_i}{q_i} ; \gt(h_i,\M^p_i)     & 
     \mbox{if $p \mathbin= \ch(\varepsilon,\M)$ and $\proc(p,\M) = \bigoplus_{i\in I}\send{q_i}{\lambda_i} ; \PP_i$.}
  \end{array}\right.
\]
\hfill
where $h_i := h(\M,p,q_i)$, \ie the sequence obtained from $h$ by appending the triple $(\M,p,q_i)$.
\caption{Algorithm for synthesising a global type for a network.
}\label{fig:algorithm}
\end{figure*}

Example~\ref{ex:not well typed} shows that
the strongest completeness result possible is completeness with respect to $\Live{\J}$.
Before turning to our completeness proof in the next section, we demonstrate the power of our general merge operator. 
\vspace{1ex} 

\begin{example}{guarded typing positive}
This network consists of two independent pairs of threads, both of which make a choice repeatedly.
\[
\begin{array}{rl}
& \loc{\buyer1}{ \rec{X} \left( \send{\seller1}{\extend} ; X \oplus \send{\seller1}{\order} \right) } \\
\pipar &
\loc{\seller1}{ \rec{Y} \left( \recv{\buyer1}{\extend} ; Y + \recv{\buyer1}{\order} \right) }  \\
\pipar &
\loc{\buyer2}{ \rec{X} \left( \send{\seller2}{\extend} ; X \oplus \send{\seller2}{\order} \right)} \\
\pipar &
\loc{\seller2}{ \rec{Y} \left( \recv{\buyer2}{\extend} ; Y + \recv{\buyer2}{\order} \right) }
\end{array}
\]
The following is a global type for this example.
\[
\G = 
\rec{X}
\Big( 
\begin{array}[t]{@{}l@{}}
\big(
\begin{array}[t]{@{}l@{}}
 \comm{\buyer1}{\extend}{\seller1} ; 
\\
(
\begin{array}[t]{@{\,}l@{}}
\comm{\buyer2}{\extend}{\seller2} ;
X
\\
\boxplus~
\comm{\buyer2}{\order}{\seller2} ;
\G_Y)\big)
\end{array}
\end{array}
\\
\boxplus~
\comm{\buyer1}{\order}{\seller1} ;
\G_Z
\Big)
\end{array}
\vspace{-1ex}
\]
$
\mbox{with}\quad
\begin{array}[t]{r@{~=~}l}
\G_Y & \rec{Y}
\big(
\begin{array}[t]{l}
\comm{\buyer1}{\extend}{\seller1}; Y 
\\
\boxplus~
\comm{\buyer1}{\order}{\seller1}; \End
~\big) 
\end{array}
\\
\G_Z & 
\rec{Z}
\big(
\begin{array}[t]{l}
\comm{\buyer2}{\extend}{\seller2} ;
Z 
\\
\boxplus ~
\comm{\buyer2}{\order}{\seller2} ; \End
~\big).
\end{array}
\end{array}$\\[1ex]
We can show that $\N \vdashg \G$. 
For example, we have 
\[
\begin{array}{r@{~=~}l}
\proj{\buyer1}\G & \rec{X} \big(
 \begin{array}[t]{l}
  \send{\seller1}{\extend}; \\
  \big(X \sqcap
\rec{Y} (
\begin{array}[t]{l}
\send{\seller1}{\extend}; Y 
\\
\oplus~
 \send{\seller1}{\order}; \End )\big)
\end{array}
 \\
 \oplus~\send{\seller1}{\order}; \End ~\big)\text{ and}
 \end{array} 
\\[13mm]
\multicolumn{2}{l}{
\rec{X} \left( 
  \send{\seller1}{\extend}; X \oplus 
  \send{\seller1}{\order}
   \right) \vdash  \proj{\buyer1} \G
}
\end{array}
\]
where the projection $\proj{\buyer1}{\G}$ on $\buyer1$ is guarded. 
\vspace{1ex} 
\end{example}

The above example is out of scope of most systems for global session types that do not feature an explicit parallel composition
 operator. These systems are incomplete, in the sense that there are lock-free networks that cannot be typed.
 Our session type system overcomes the incompleteness, due our general treatment of merge.
A similar example, which also can be typed by our methodology, is given in~\cite{Scalas2019}. It is used there to demonstrate that there are networks that cannot be typed using established notions of global type without parallel composition~\cite{Lange2015}.
An alternative approach using coinductive projections has been proposed in~\cite{Barbanera2021}.

\subsection{Completeness for lock-freedom under justness}

We now show one of our main results, namely that, for our session type calculus, all lock-free networks can be typed,
when assuming justness. To the best of our knowledge, this is the first completeness result of this kind.
\vspace{1ex} 

\begin{theorem}{guarded typing completeness}
If $\N\models\Live{\J}$, then $\N$ is guardedly well-typed.
\vspace{1ex} 
\end{theorem}

The proof \arxiv{(Appendix \ref{app:E}) }{\cite[Appendix~D]{GHH21} } makes use of an algorithm for synthesising a global type from a network, along with a proof that the algorithm terminates with the correct guarded type.

To express the algorithm we require the following concepts.
A \emph{reachable network}, from a given network $\N$, is a reachable network state $\M$
  that happens to be a network, in the sense that $\proc(p,\N) \neq \chosen{\PP}$
  for all locations $p$ of $\N$. A network $\M$ is \emph{unfolded}
 \label{pg:netunfolded} 
  if there is no network $\M'$ with $\M \goesto\tau \M'$
  (although there may be network states $\M'$ with $\M \goesto\tau \M'$). The unfolding of a network $\M$ is the unique network $\M'$  such that $\M \mathrel{(\goesto\tau)^*} \M'$ and $\M'$ is unfolded.
A location $p$ is \emph{ready} in a network $\N$ if $\proc(p,\N) = \bigoplus_{i\in I}\send{q_i}{\lambda_i} ; \PP_i$,
and for each $i\mathop\in I$ there exists a transition $\N \dgoesto{\comm{p}{\lambda_i}{q_i}~} \N_i$,
using the transition relation of Figure~\ref{fig:reactive}.
Define a \emph{history} as a sequence of triples $(\N,p,q)$ with $\N$ a network and $p,q$ locations of ${\N}$.
A history $h$ is \emph{complete} for a network $\M$ if each location that is ready in $\M$ occurs in $h$.
If $h$ is a history and $\M$ a network expression that occurs in $h$,
then $h \upharpoonright \M$ denotes the suffix of $h$ that starts with the first occurrence of $\M$ in $h$.
Moreover, $h \upharpoonleft \M$ denotes the prefix of $h$ prior to the first occurrence of $\M$ in $h$,
so that $h = (h \upharpoonleft \M)(h \upharpoonright \M)$.
Call a location $p$ \emph{eligible} in a network state $\M$ w.r.t.\ a history $h$ if (a) $p$ is ready in $\M$,
and (b) either $\M$ does not occur in $h$, or $p$ does not occur in $h \upharpoonright \M$.
Finally, {\sc deadlock} is a constant, temporarily added to the syntax of session types.

Our algorithm requires several choices. 
First, we select a fresh variable $X_\M$ for each unfolded network $\M$ that is reachable from $\N$.
We then pick a total order on the finite set of locations of $\N$, referred to as \emph{age},
so that each nonempty set of locations has an oldest element.
Finally, for each reachable network $\M$  and each location $p$ that is ready in $\M$, say
with $\proc(p,\M) = \bigoplus_{i\in I}\send{q_i}{\lambda_i} ; \PP_i$, and for each $i\mathbin\in I$,
pick a net- work $\M^p_i$ such that $\M \dgoesto{\comm{p}{\lambda_i}{q_i}~} \M^p_i$ and $\proc(p,\M^p_i)=\PP_i$.

Our algorithm employs the routine $\gt(h,\M)$, parametrised by the choice of a network $\M$ and a history $h$, as defined in Figure~\ref{fig:algorithm}.
Here, $\ch(h,\M)$ is a partial function that selects, for a given history $h$ and reachable network $\M$,
the oldest location that is eligible in $\M$ w.r.t.\ $h$.
It is defined only when such a location exists. 
The case distinction in the figure is meant to be prioritised, in the sense that
a later-listed option is taken only if none of the higher-listed options apply.
Our algorithm is then defined to yield the global type $\gt(\varepsilon,\N)$, with $\N$ the given network and $\varepsilon$ the empty 
history (sequence).

The intuition for our algorithm, which attempts to construct a global session type $\G$ out of a given network \N,  is as follows.
Since it is essential that $\G$ induces ongoing progress of all unterminated locations in the
network, we keep track of the history $h$ of communications incorporated in $\G$ until the ``construction front'' at network state $\M$.
Here the routine $\gt(h,\M)$ specifies the next communication-choice that will be incorporated in $\G$. 
The first clause in Figure~\ref{fig:algorithm}, where the $\tau$-transition must unfold recursion, says
that all recursions should be unfolded before attempting the remaining case distinctions. 
The second clause says that we can safely terminate upon reaching a state in which the threads of
all locations have terminated.
The last two clauses specify a choice leader $p$ and extend $\G$ with the send actions of $p$ in state $\M$.
Here $p$ must be a location that is ready in $\M$; if such a $p$ does not exist 
the failed attempt is reported by including the constant {\sc deadlock} in the attempted session type $\G$ (Clause~3).
Since the syntactic expression $\G$ must be finite, 
each branch that does not reach $\End$ needs to loop back to a previous stage in the construction of $\G$, at some point.
To facilitate looping back, we attach a recursion variable $X_\M$ to each stage we might want to
loop back to, namely to each first occurrence of a network state $\M$ in our history. This explains the difference
between Clauses 5 and 6. Clause 4 says that we can safely loop back to a previous stage if it
involves the same current network state $\M$, and between that previous stage and the present all
locations that are ready in $\M$ already had a turn. If Clause 4 does not apply, then $\M$ must have
eligible locations w.r.t.\ $h$, and Clause 5 or 6 picks the oldest such location, to make
sure that in the end all eligible locations get a turn.
\vspace{1ex} 

\begin{example}{algo}
  Applying this algorithm to the network of \ex{binaryext:early} yields the type of \ex{typing},
  but with a spurious recursive anchor $\mu Z$ right before $\comm{\seller}{\order}{\shipper}$.
  Applied to \ex{sc-j:early} it fails with possible output
  \[\begin{array}{l}
  \rec{ X } \comm{\buyer1}{\order1}{\seller};\\ \comm{\buyer2}{\order2}{\seller}; \mbox{\sc deadlock}.
  \end{array}\]
  For \ex{p-jt:early} it yields the following correct type.\vspace{-2pt}
  \[\rec{ X }( \comm{\buyer1\!}{\order}{\seller1}; \comm{\buyer2}{\order}{\seller2}; X)\]
For \ex{livelock} it yields the type
  $\rec{ X }\comm{\buyer}{\buy}{\seller}; X$;
  this type is incorrect for that network $\N$. This does not contradict the proof of \thm{guarded typing  completeness}, since $\N\not\models\Live{\J}$.

For \ex{sc-j:2} the algorithm yields the type of \ex{SC-typed}.
For \ex{R versus L} it fails with output \mbox{\sc deadlock}.\\
For \ex{reactive:early} it yields the following correct type.
  \[\begin{array}{l}
  \rec{ X } (\comm{\buyer}{\nego}{\seller1}; X \\
  \boxplus\, \comm{\buyer}{\order}{\seller2}; \rec{ Z } \comm{\buyer}{\done}{\seller1}; \End)
  \end{array}\]
  For \ex{guarded typing positive} it yields the type of \ex{guarded typing positive}.
\vspace{1ex} 
\end{example}

\begin{observation}{WC complete}
An immediate consequence of Proposition~\ref{pr:WC} and Theorem~\ref{thm:guarded typing completeness},
is that guardedly well-typed networks are complete for $\Live{\W\C}$, suggesting that a
carefully selected notion of weak fairness 
is suitable for some session calculi.
\vspace{1ex} 

\begin{corollary}{WC complete}
If $\N\models\Live{\W\C}$ then $\N$ is guardedly well-typed.\vspace{1ex}
\end{corollary}

In contrast, recall that Example~\ref{ex:R versus L} satisfies $\React{\Pr}$;
yet it is not (guardedly) well-typed. This shows that there is no corresponding
completeness result for any notion of lock-freedom $\React{$\mathcal{F}$}$,
where $\mathcal{F}$ is some notion of fairness.
This is an argument for why we emphasise the semantics in Figure~\ref{fig:red} rather than the one in Figure~\ref{fig:reactive}.
\end{observation}

\section{Race-freedom and soundness\label{sec:sound}}

We have established that completeness holds, with respect to $\Live{\J}$.
Hence, if we can model-check $\Live{\J}$, we know we can always synthesise a global type for a network.
In this section, we consider soundness, meaning that a network is lock-free if it is (guardedly) well-typed.
To complement our completeness result, we target $\Live{\J}$ and prove soundness for 
guardedly well-typed networks that are additionally race-free.
The insight of this section is that soundness can be achieved when making the minimal fairness assumption justness.

\vspace{1ex} 

\begin{definition}{race-free}
A network state $\N$ has a \emph{race} whenever 
 $\N \goesto{\comm{p}{\lambda}{r}~} \N'$ and $\N \goesto{\comm{q}{\mu}{r}~} \N''$ with either $p\neq q$ or
 $\N' \mathbin{\neq} \N''$.
A network is \emph{race-free} if it has no reachable network state with a race.
\vspace{1ex} 
\end{definition}

Figure~\ref{reduced taxonomy} implies that $\Live{\J}$ is the strongest lock-freedom property we can get, 
with the exception of $\Live{\Pr}$.
Our guarded type system cannot be sound for the latter notion of lock-freedom, 
for Example~\ref{ex:p-jt:early} is guardedly well-typed and race-free, but does not satisfy $\Live{\Pr}$.

Our example to distinguish $\J$ and $\St\C$ features races.
Indeed, there is no race-free network separating $\J$ from $\St\C$,  as confirmed by the following proposition.
\vspace{1ex} 

\begin{proposition}{race-free collapse}
On race-free networks,  $\J$ coincides with $\St\C$, for our session calculus in Figure~\ref{fig:red}.
\end{proposition}
\begin{proof}
Let $\pi$ be an infinite path in a network that is not SC-fair. So, on $\pi$, a component $p$ is
infinitely often enabled, but never taken. 

In case $p$ is stuck in a state where its next transition is a
$\tau$, then $\pi$ is not just.

In case $p$ is stuck in a state $\chosen[3pt]{\send{q}{\lambda}; P}$, then, in the first state on $\pi$ on which
$p$ is enabled, $q$ must be in a state $\sum_{i \in I} \recv{p_i}{\lambda_i}; \PP_i$ with $p=p_k$ and
$\lambda=\lambda_k$ for some $k\in I$. If $q$ remains in this state throughout $\pi$, then $\pi$ is
not just. If $q$ leaves this state via a transition of $\pi$ that does not involve $p$, then the
state where this happens must be a race, and the network is not race-free.

In the remaining case, $p$ is stuck in a state $\sum_{i \in I} \recv{p_i}{\lambda_i}; \PP_i$.
For $p$ to be enabled, a component $p_k$ with $k\in I$ must be in a state
$\send{p}{\lambda_k}; P$, which reduces this case to the previous one.
\end{proof}
\begin{observation}{reactive race}
In contrast to Proposition~\ref{pr:race-free collapse}, for
 the session calculus in Figure~\ref{fig:reactive}, 
$\J$ and $\St\C$ do not coincide, even for race-free networks (race-freedom also needs to be reformulated for that semantics).
Example~\ref{ex:reactive:early} illustrates this fact.
\vspace{1ex} 
\end{observation}

Since Examples \ref{ex:binaryext:early}, \ref{ex:p-jt:early} and~\ref{ex:livelock} are race-free, the remaining four notions
$\Live{\Pr}$, $\Live{\St\C}$, $\Live{\St\T}$ and deadlock-freedom from Figure~\ref{reduced taxonomy} are all different for race-free networks. Consequently,
 for race-free networks, the only collapse of Figure~\ref{reduced taxonomy} is $\Live{\J} \Leftrightarrow \Live{\St\C}$.
\vspace{1ex} 

\begin{example}{reactive}
The following network is race-free, guardedly well-typed and satisfies $\Live{\J}$.
It is a variant of Example~\ref{ex:reactive:early}, which is also race-free, but does not satisfy $\Live{\J}$.
\[
\begin{array}[t]{rl}
&
\buyer\mbox{\Large\textlbrackdbl}
 \rec{X} \big(
  \begin{array}[t]{@{}l}
  ( \send{\seller1}{\order1} ; \send{\seller2}{\extend} ; X )
  \\
  \oplus ~
  (\send{\seller2}{\order2}; \send{\seller1}{\done})\big)\ \mbox{\Large\textrbrackdbl}
  \end{array}
\\
\pipar
&
\loc{\seller1}{
 \rec{Y} (
  \recv{\buyer}{\order1} ;  Y
  +
  \recv{\buyer}{\done})\ 
}
\\
\pipar
&
\seller2\mbox{\Large\textlbrackdbl}
 \rec{Z}(
  \recv{\buyer}{\extend} ; Z
  + 
  \recv{\buyer}{\order2}) ~ \mbox{\Large\textrbrackdbl}
\end{array}
\]
This network is \welltyped, for we can use the following global type.\vspace{-1ex}
\[
\G =
\rec{X} \big( 
 \begin{array}[t]{@{}l@{}}
 (
 \begin{array}[t]{@{}l}
 \comm{\buyer}{\order1}{\seller1} ;
 \\
 \comm{\buyer}{\extend}{\seller2} ; X)
 \end{array}
 \\
 \boxplus 
     \begin{array}[t]{l}
     (\comm{\buyer}{\order2}{\seller2} ; \\
     \phantom{(}\comm{\buyer}{\done}{\seller1} )\big)
     \end{array}
 \end{array}
\]
\end{example}

This example illustrates that it is possible to send messages to several locations via an internal
choice in a race-free way. 
\arxiv{A simple way to prevent races, adopted by many session type systems, \eg~\cite{Denielou2012,Ghilezan2019,Yoshida2020}, is to ensure that each location listens to only one other location at any time. 
To achieve this we can impose the following syntactic restriction.
\vspace{1ex} 

\begin{definition}{inputs}
A network $\N$ is \textit{syntactically race-free}, 
if, for every sub-expression of the form
$\sum_{i \in I} \recv{p_i}{\lambda_i}; \PP_i$, we have\linebreak[4] $p_i \mathbin= p_j$
for all $i, j \mathbin\in I$, and $\lambda_i \mathbin{\neq} \lambda_j$ for all $i, j \mathbin\in I$ with $i \mathbin{\neq} j$.
\vspace{1ex} 
\end{definition}

Restricting external choices this way was never intended to be a complete criteria for race-freedom. 
However, it is a cheap linear syntactic property to check, whereas checking for race-freedom is as complex as checking for deadlock-freedom.
\vspace{1ex} 

\begin{example}{syntax}
The following example is race-free, but not syntactically race-free.
\[
\begin{array}[t]{rl}
&
\buyer\mbox{\Large\textlbrackdbl}
 \rec{X} (
  \send{\seller1}{\order1}  ; X
  \oplus 
  \send{\seller2}{\order2})\mbox{\Large\textrbrackdbl}
\\
\pipar
&
\seller1 \mbox{\Large\textlbrackdbl}
 \rec{Y}(
  \begin{array}[t]{@{}l@{}}
  \recv{\buyer}{\order1} ; \send{\seller2}{\extend} ; Y
  \\
  +~
  \recv{\seller2}{\done})\mbox{\Large\textrbrackdbl}
  \end{array}
\\
\pipar
&
\seller2\mbox{\Large\textlbrackdbl}
 \rec{Z} (
  \begin{array}[t]{@{}l@{}}
  \recv{\seller1}{\extend} ; Z \\
  +~
  \recv{\buyer}{\order2} ; \send{\seller1}{\done} ) \mbox{\Large\textrbrackdbl}
  \end{array}
\end{array}
\]
\end{example}
It is also guardedly well-typed and satisfies $\Live{\J}$. 
Soundness results are stronger if race-free networks are considered, rather than syntactically race-free networks.%
}{
Many session type systems, \eg~\cite{Denielou2012,Ghilezan2019,Yoshida2020},
ensure that, for every sub-expression of the form
$\sum_{i \in I} \recv{p_i}{\lambda_i}; \PP_i$, we have\linebreak[4] $p_i \mathbin= p_j$
for all $i, j \mathbin\in I$, and $\lambda_i \mathbin{\neq} \lambda_j$ for all $i, j \mathbin\in I$ with $i \mathbin{\neq} j$.
This syntactically guarantees race-freedom, but excludes the above example.
\vspace{1ex} 
}

The reverse of \thm{guarded typing completeness} does not hold.
Hence we cannot expect a soundness result for all networks.
It is not even the case that guardedly well-typed networks are deadlock-free.
\vspace{1ex} 
\begin{example}{mini}
The following network is guardedly well-typed, but neither race-free nor deadlock-free, 
and hence certainly not $\Live{\J}$.\pagebreak[3]
\[
\begin{array}{rl}
&
\loc{\buyer1}{
 \send{\seller}{\buy1}; \End
}
\\
\pipar
&
\loc{\buyer2}{
 \send{\seller}{\buy2}; \End
}
\\
\pipar
&
\seller\mbox{\Large\textlbrackdbl}
 \begin{array}[t]{@{}l@{}}
( \recv{\buyer1}{\buy1} ;
 \recv{\buyer2}{\buy2};\End)
 \\
 +~
 ( \begin{array}[t]{@{}l@{}}
 \recv{\buyer2}{\buy2} ;
 \recv{\buyer1}{\buy1} ;\\
 \send{\buyer1}{\order};\End) \mbox{\Large\textrbrackdbl}
 \end{array}
\end{array}
\end{array}
\]
The global type for this example is the following.
\[
\G = \comm{\buyer1}{\buy1}{\seller};
\comm{\buyer2}{\buy2}{\seller}; \End
\]
In particular,
$
 \proj{\seller}{\G} = \recv{\buyer1}{\buy1} ; \recv{\buyer2}{\buy2}
$ and
\[
 \begin{array}[t]{l}
 \recv{\buyer1}{\buy1} ;
 \recv{\buyer2}{\buy2}
 \\
 +~
 \recv{\buyer2}{\buy2} ;
 \recv{\buyer1}{\buy1} ;
 \send{\buyer1}{\order} \vdashg \proj{\seller}{\G}
\end{array}
\]
which holds due to the rule for $\sum$ in Figure~\ref{fig:types}, permitting branches of an external choice to be removed.
\vspace{1ex} 

\end{example}
This example may suggest that the culprit preventing soundness is the flexible external choice, \ie the subtype relation $\vdashg$.
However, even if $\vdashg$ would be almost the identity relation,
with each merge on types corresponding to an external choice of
the corresponding threads, 
there would be guardedly well-typed networks that are not deadlock-free.
\vspace{1ex} 

\begin{example}{not deadlock-free}
Consider the following network.
\[
\begin{array}{rl}
&
\loc{p}{
 (\send{s}{a}; \send{t}{a}; \send{r}{d}) \oplus (\send{s}{b}; \send{t}{b})
}
\\
\pipar
&
\loc{r}{
 (\recv{s}{c}; \recv{t}{e}; \recv{p}{d}) + (\recv{t}{e}; \recv{s}{c})
}
\\
\pipar
&
\loc{s}{
 \recv{p}{a}; \send{r}{c} + \recv{p}{b}; \send{r}{c} 
}
\\
\pipar
&
\loc{t}{
 \recv{p}{a}; \send{r}{e} + \recv{p}{b}; \send{r}{e} 
}
\end{array}
\]
Using the global type
\[
\G = \begin{array}[t]{l}
     (\comm{p}{a}{s}; \comm{p}{a}{t}; \comm{s}{c}{r}; \comm{t}{e}{r}; \comm{p}{d}{r}; \End) \\
     \boxplus ~
      \comm{p}{b}{s}; \comm{p}{b}{t}; \comm{t}{e}{r}; \comm{s}{c}{r}; \End
     \end{array}
\]
yields projections that are identical to the network threads,
\eg
\[\proj {p} \G  = (\send{s}{a}; \send{t}{a}; \send{r}{d}) \oplus \send{s}{b}; \send{t}{b}.\]
So, $\N \vdashg \G$ follows by using the identity as subtyping relation.
Yet, this network is not deadlock-free, for the execution
 $\comm{p}{b}{s}; \comm{s}{c}{r}; \comm{p}{b}{t}; \comm{t}{e}{r}$
 reaches a deadlock with hanging input $\recv{p}{d}$ in location $r$.%
\footnote{
This is a counterexample to subject reduction in previous work that allows multiple recipients in an external choice~\cite{Castellani2019b,Castellani2020}.
Those soundness results can be restored by restricting to race-free networks.}
\vspace{1ex} 
\end{example}
Examples~\ref{ex:mini} and~\ref{ex:not deadlock-free} are both excluded by
 race-freedom.
We now prove another main result, namely that the converse of \thm{guarded typing completeness} holds for race-free networks.
\vspace{1ex} 

\begin{theorem}{guarded typing}
If $\N$ is guardedly well-typed and race-free, then $\N\models\Live{\J}$.
\vspace{1ex} 
\end{theorem}

The proof [\arxiv{see Appendix~\ref{app:F}}{{Appendix F of \cite{GHH21}}}]
hinges on the following \emph{session fidelity} result,
for which we appeal to race-freedom:

For race-free network states $\N$, if $\N\goesto{\comm{p}{\lambda}{q}\,}\M$
and $\N \vdashg \G$ then there exists $\G'$ such that
$\G \dgoesto{\comm{p}{\lambda}{q}} \G'$ and $\M \vdashg \G'$.
Here, $\dgoesto\alpha$ is a transition relation on global types, defined for this purpose.
Session fidelity strengthens subject reduction, by insisting that the form of $\G'$ reflects the transition.

From this statement we conclude that
each reachable network state $\N'$ along any just path
$\pi$ is guardedly well-typed. For every location $p$, either there are infinitely many transitions along $\pi$ involving $p$,
or there exists a suffix $\pi'$ of $\pi$ stemming from network state $\N'$, such that location $p$ has no further transition involving $p$. 
Using justness and the fact that $\N'$ is guardedly well-typed one can show that
in the latter case $p$ has successfully terminated.

Using Theorems~\ref{thm:guarded typing completeness} and~\ref{thm:guarded typing}, and Proposition~\ref{pr:race-free collapse}
leads to a soundness and a completeness result for our type system.
\vspace{1ex} 

\begin{corollary}{characterise SC}
For race-free network $\N$, 
$\N$ is guardedly well-typed iff $\N$ satisfies $\Live{\St\C}$.
\end{corollary}
\vspace{1ex} 

{
\newcommand{\lambdai}{a}
\newcommand{\lambdaii}{b}
\begin{observation}{guarded}
Projections are defined such that recursion maps to $\End$ whenever $p \not\in \participants{\G}$ and $\mu X.\G$ is closed (see Page \pageref{def:proj}).
The latter condition plays an essential role for soundness. To see why, consider the following global type.
\[
\rec{X} \comm{p}{\lambdai}{q} ; \rec{Y} \comm{r}{\lambdaii}{q} ; X
\]
The anchor with variable $Y$ should, intuitively, be useless.
However, if we were to exclude condition ``$\mu X.\G$ is closed'', the above global type would type the following network.
\[
\begin{array}{rl}
 &\loc{p}{   \rec{X} \send{q}{\lambdai} ; \End }
\\
\pipar
&
\loc{q}{   \rec{X} \recv{p}{\lambdai} ; \rec{Y} \recv{r}{\lambdaii} ; X }
\\
\pipar
&
\loc{r}{   \rec{X} \rec{Y} \send{q}{\lambdaii} ; X }
\end{array}
\]
This race-free network reaches a deadlock right after communications
$\comm{p}{\lambdai}{q} ; \comm{r}{\lambdaii}{q}$. 
The above closedness-condition
resolves this issue, and can be used to correct papers on global types featuring recursion binders.
\end{observation}
}

\section{Related and future work on lock-freedom}\label{sec:related}

While our completeness result is the first of its kind, 
there are several soundness results for type systems with respect to some notion of lock-freedom, e.g.,~\cite{Padovani2014b,Dezani2005,Padovani2014,Severi2019,Carbone2014},
the most closely related of which we draw attention to in this section.
We also situate related work on lock-freedom with regard to our classification and point to future challenges.
\vspace{1ex} 

\paragraph{Strong lock-freedom}
Severi and Dezani-Ciancaglini\linebreak[4] propose a notion of \textit{strong lock-freedom}~\cite{Severi2019}
that coincides with $\React{\J}$ for race-free networks.
They employ a reactive semantics.
The authors impose a restriction on paths, ensuring that all concurrent transitions proceed in lockstep. Their assumption is not a fairness assumption, 
as defined here, as it does not satisfy feasibility.
However, it does have the effect of assuming justness, up to permutations of transitions, for race-free networks.
By Observation~\ref{obs:WC complete}, a completeness result along the lines of Theorem~\ref{thm:guarded typing completeness} cannot hold for strong lock-freedom. 
Strong lock-freedom cannot be lifted directly to networks with races.
A minimal change to their definitions requiring a maximal number of enabled locations to act in every step would extend their definition to networks with races; we did not analyse this extension.
Their use of coinductive syntax, rather than binders, is an alternative for avoiding the soundness problem in Observation~\ref{obs:guarded} that is common in the literature.
\vspace{1ex} 

\paragraph{Further lock-freedom schemes}
Carbone, Dardha and Montesi translate Kobayashi's scheme for lock-freedom
 to a session calculus where both internal and external choices are with a single location~\cite{Carbone2014}.
Their scheme is instantiated with $\SC$\,\footnote{The authors do not provide a definition of fairness, but cite Kobayashi~\cite{Kobayashi2002} instead. Kobayashi's definition does not lift immediately to the calculus of~\cite{Carbone2014}. However, the authors appear to intend $\SC$.}
and coincides with $\Live{\SC}$, restricted to their calculus.
Their scheme inherits the ambiguity discussed in Section~\ref{sec:Kobayashi}.
It assumes a semantics intermediate to those we study in this work, where internal choice is like in Figure~\ref{fig:red}, but recursion and singleton internal choice are reactive as in Figure~\ref{fig:reactive}.
This makes their approach weaker than ours, if lifted directly to our calculus with flexible choices:
Example~\ref{ex:sc-j:2} is lock-free under their scheme instantiated with assumption $\J$ (or even $\Pr$), but does not satisfy $\Live{\J}$ (or even $\React{\J}$) in our scheme.
Note that their work concerns binary session types with delegation, which we do not  consider.

Scalas and Yoshida propose the notions of $\textsc{Live}$, $\textsc{Live+}$, and $\textsc{Live++}$~\cite{Scalas2019}.
The first, $\textsc{Live}$, follows the scheme of Padovani, hence coincides with $\Live{\St\T}$.
The second, $\textsc{Live+}$, is essentially another formulation of $\Live{\SC}$. 
The third, \textsc{Live++}, coincides with $\Live{\Pr}$, hence is unsound for session calculi since it rejects key examples such as Example~\ref{ex:p-jt:early}.
The definitions of \cite{Scalas2019} are arguably less portable than Definition~\ref{df:lock-free}, since their definitions refer to specific language features.
\vspace{1ex} 

\paragraph{Asynchronous session calculi}
An evaluation of lock-freedom for asynchronous calculi, where queues are inserted between communicating threads, requires separate attention.
The asynchronous analogue to our session calculus is an infinite-state system. 
Therefore, $\St\T$ is no longer the strongest fairness assumption; this is now \emph{full fairness} ($\Fu$)~\cite{GH19}.
At the other end of the spectrum, there are also complications when defining concurrency of transitions (Definition~\ref{df:concurrent}).
It is a design decision whether a thread is treated as a single component along with its queues,
and whether enqueue and dequeue events for the same queue are dependent or concurrent. 
Consequently, synchrony/asynchrony and the spectrum of fairness assumptions are not entirely perpendicular dimensions when defining notions of lock-freedom. 
Fairness plays an essential role in related work on preciseness of subtyping for asynchronous calculi~\cite{Ghilezan2021}, which is further evidence that fairness assumptions require scrutiny here.
\vspace{1ex} 

\paragraph{Synthesis and multiparty compatibility}
The body of literature on synthesising global types from multiparty compatible local
types~\cite{Lange2012,Malo2013,Lange2015} plays  a complementary role to our synthesis results, used to establish completeness.
Usually, it is immediate that networks inhabiting a global type are multiparty compatible. 
Hence, we expect that a corollary of Theorem~\ref{thm:guarded typing completeness} is that $\Live{\J}$ implies multiparty compatibility, for some notion of multiparty compatibility.
If we further assume a synthesis result showing that multiparty compatible networks are guardedly well-typed -- under conditions such as race-freedom -- then that notion of multiparty compatibility coincides with $\Live{\J}$.
Such a result for our global type system, does not quite follow immediately from synthesis results in the literature, since Example~\ref{ex:guarded typing positive} would require parallel composition in related work. 
The formal development of multiparty compatibility is left as future work.
\vspace{1ex} 
 
\paragraph{Fair subtyping and weak normalisation}
A fair subtyping relation has been defined for session types as the largest relation over threads that preserves weak normalisation~\cite{Padovani16b}. Weak normalisation is the property that, at any point during an execution, it is not inevitable that a network will not successfully terminate.
Weak normalisation is strictly stronger than Padovani's notion of lock-freedom (\df{lock-free}) -- since the possibility of all components to successfully terminate entails the possibility of all components performing some enabled action -- but is incomparable to liveness properties stronger than $\Live{\St\C}$, including $\Live{\J}$ --
Example~\ref{ex:binaryext:early} is weakly normalising, but does not satisfy $\Live{\St\C}$. 
Consequently, the proposed notion of fair subtyping does not quite fit lock-freedom. Investigating a notion of fair subtyping that is adequate for $\Live{\St\T}$ rather than weak normalisation, and also identifying a session type system complete for $\Live{\St\T}$, as hinted at in the discussion surrounding Example~\ref{ex:unsound}, is future work. In particular, we do not claim that $\Live{\J}$ is the only notion of lock-freedom that can be characterised by some session type system.

\section{Conclusion}
In this paper, we have systematically classified the notions of lock-freedom that arise by taking
every fairness assumption listed in a recent survey~\cite{GH19}.
Based on our comprehensive analysis, we are compelled to put forward a notion of fairness suitable for session calculi:
\textit{justness} (Definition~\ref{df:just}), and its resulting notion of lock-freedom $\Live{\J}$, which we propose to call \textit{``just lock-freedom''}.
Through a generalisation of the classical merge operation on local session types, we have devised a session type system that is complete for just lock-freedom. 
Moreover, race-free networks are sound for just lock-freedom.
Justness is always reasonable to assume, since it does not constrain the `free will' of participants (\cf Examples~\ref{ex:sc-j:early} and~\ref{ex:reactive:early}), while ensuring that concurrent transitions do not constrain each other (\cf Examples~\ref{ex:p-jt:early} and~\ref{ex:guarded typing positive}).

A strength of our results is that completeness (Theorem~\ref{thm:guarded typing completeness}) holds for networks with flexible choice, in which branches of the same choice operator may involve different locations.
Completeness suggests a methodology for session calculi that allows us to pass straight from any network satisfying our realistic notion of lock-freedom $\Live{\J}$ to a global session type.
The methodology would be to directly model check that a network satisfies $\Live{\J}$, and then use the algorithm in Figure~\ref{fig:algorithm} to synthesise a global type for that network.
This methodology works even for networks featuring races. 

Interestingly, there are no previous results synthesising global types directly from lock-freedom. Indeed, Example~\ref{ex:guarded typing positive}, which satisfies almost all notions of lock-freedom in the literature, is known to be out of scope of related session type systems based on global types without explicit parallel composition.
While this incompleteness issue in related work is partly due to the less general merge operator employed in those systems, 
another reason that enables us to obtain the completeness result in
Theorem~\ref{thm:guarded typing completeness} is our
  scrutiny of the role of fairness assumptions. 
  In fact, Example~\ref{ex:not well typed} shows that completeness of our session type system cannot be attained when assuming strong fairness of components.
Furthermore, even small variations in the choice of semantics for the transition system can affect
fairness assumptions significantly, weakening corresponding notions of lock-freedom (see Observation~\ref{obs:WC complete}).
Indeed, amongst all notions of lock-freedom considered in this paper, only $\Live{\J}$ yields both completeness for all networks and soundness  for race-free networks (Theorem~\ref{thm:guarded typing}).
\vspace{1ex plus 2pt}

\paragraph*{Acknowledgement}
The key question addressed in this paper arose in conversation with Ilaria Castellani, Mariangiola Dezani-Ciancaglini, and Paola Giannini, to whom we are grateful for their generous feedback on this work.

\bibliographystyle{IEEEtran}
\bibliography{sessions}

\arxiv{
\cleardoublepage
\appendix
\input{appA}
\input{appB}
\input{appC}
\input{appE}
\input{appF}
}{}
\end{document}

%% file: lattice2b.tex
\expandafter\ifx\csname graph\endcsname\relax
   \csname newbox\expandafter\endcsname\csname graph\endcsname
\fi
\ifx\graphtemp\undefined
  \csname newdimen\endcsname\graphtemp
\fi
\expandafter\setbox\csname graph\endcsname
 =\vtop{\vskip 0pt\hbox{%
    \graphtemp=.5ex
    \advance\graphtemp by 1.548in
    \rlap{\kern 0.119in\lower\graphtemp\hbox to 0pt{\hss \makebox[4pt][l]{$\textrm{P} = \W\T$}\hss}}%
    \graphtemp=.5ex
    \advance\graphtemp by 1.071in
    \rlap{\kern 0.119in\lower\graphtemp\hbox to 0pt{\hss \makebox[4pt][l]{$\J = \W\C$}\hss}}%
\pdfliteral{
q [] 0 d 1 J 1 j
0.576 w
0.576 w
10.368 -95.544 m
8.568 -102.744 l
S
6.768 -95.544 m
8.568 -102.744 l
S
8.568 -85.68 m
8.568 -102.744 l
S
Q
}%
    \graphtemp=.5ex
    \advance\graphtemp by 0.595in
    \rlap{\kern 0.119in\lower\graphtemp\hbox to 0pt{\hss \makebox[8pt][l]{$\St\C$}\hss}}%
\pdfliteral{
q [] 0 d 1 J 1 j
0.576 w
10.368 -61.272 m
8.568 -68.472 l
S
6.768 -61.272 m
8.568 -68.472 l
S
8.568 -51.408 m
8.568 -68.472 l
S
Q
}%
    \graphtemp=.5ex
    \advance\graphtemp by 0.119in
    \rlap{\kern 0.119in\lower\graphtemp\hbox to 0pt{\hss \makebox[4pt][l]{$\St\T$}\hss}}%
\pdfliteral{
q [] 0 d 1 J 1 j
0.576 w
10.368 -27 m
8.568 -34.2 l
S
6.768 -27 m
8.568 -34.2 l
S
8.568 -17.136 m
8.568 -34.2 l
S
Q
}%
    \hbox{\vrule depth1.548in width0pt height 0pt}%
    \kern 0.238in
  }%
}%

%% file: lattice3b.tex
\expandafter\ifx\csname graph\endcsname\relax
   \csname newbox\expandafter\endcsname\csname graph\endcsname
\fi
\ifx\graphtemp\undefined
  \csname newdimen\endcsname\graphtemp
\fi
\expandafter\setbox\csname graph\endcsname
 =\vtop{\vskip 0pt\hbox{%
    \graphtemp=.5ex
    \advance\graphtemp by 1.905in
    \rlap{\kern 0.000in\lower\graphtemp\hbox to 0pt{\hss $\Live{\textrm{P}}$\hss}}%
    \graphtemp=.5ex
    \advance\graphtemp by 1.429in
    \rlap{\kern 0.000in\lower\graphtemp\hbox to 0pt{\hss $\Live{\textrm{J}}$\hss}}%
    \graphtemp=.5ex
    \advance\graphtemp by 1.667in
    \rlap{\kern 0.000in\lower\graphtemp\hbox to 0pt{\hss {\Large $\Uparrow$}\hss}}%
    \graphtemp=.5ex
    \advance\graphtemp by 1.190in
    \rlap{\kern 0.000in\lower\graphtemp\hbox to 0pt{\hss {\Large $\Uparrow$}\hss}}%
    \graphtemp=.5ex
    \advance\graphtemp by 0.714in
    \rlap{\kern 0.000in\lower\graphtemp\hbox to 0pt{\hss {\Large $\Uparrow$}\hss}}%
    \graphtemp=.5ex
    \advance\graphtemp by 0.238in
    \rlap{\kern 0.000in\lower\graphtemp\hbox to 0pt{\hss {\Large $\Uparrow$}\hss}}%
    \graphtemp=.5ex
    \advance\graphtemp by 0.952in
    \rlap{\kern 0.000in\lower\graphtemp\hbox to 0pt{\hss $\Live{\textrm{SC}}$\hss}}%
    \graphtemp=.5ex
    \advance\graphtemp by 0.476in
    \rlap{\kern 0.000in\lower\graphtemp\hbox to 0pt{\hss $\Live{\textrm{ST}}$\hss}}%
    \graphtemp=.5ex
    \advance\graphtemp by 0.000in
    \rlap{\kern 0.000in\lower\graphtemp\hbox to 0pt{\hss \textit{deadlock-freedom}\hss}}%
    \hbox{\vrule depth1.905in width0pt height 0pt}%
    \kern 0.000in
  }%
}%

%% file: appA.tex
\subsection{Classifying Fairness Notions for our Session Calculus\label{app:A}}

In Section~\ref{sec:fairness}, we have presented strong and weak fairness 
of transitions, and of components. We have also introduced the concepts of progress and 
justness. In this appendix, we discuss further fairness assumptions and relate them 
to each other. The notions and our classifications are based on the survey~\cite{GH19}.

The classification from \cite{GH19} defines $3 \times 6 = 18$ fairness assumptions
$xy$ with $x\mathop\in \{\St,\W,\J\}$ and $y\mathop\in \{\C,\Gr,\In,\Sy,\A,\T\}$.
This is a shorthand for ``$x$ fairness of $y$'', with $x\mathop\in \{$strong, weak, $\J$-$\}$
and $y\mathop\in \{$components, groups of components, instructions,
synchronisations, actions, transitions$\}$.

Strong and weak fairness were defined in \sect{fairness}; these notions are parametrised by the
concept of a task. The parameter $y$ above governs the choice of tasks.
\emph{Fairness of transitions}, in which each transition constitutes a task,
was already defined in \sect{fairness}.

Remember that a \emph{component} of a network expression is one of its locations, 
and that ${\it comp}$ is a function that associates with each transition the set of
one or two components that are involved in that transition.
Two transitions $t$ and $u$ are \emph{concurrent}, notation $t \conc u$, iff
$\comp(t) \cap \comp(u)=\emptyset$.

In \emph{fairness of components}, the components constitute the tasks.
Component $p$ is \emph{enabled} in a network state iff a transition $t$ with $p \in \comp(t)$ is
enabled; a path \emph{engages} in a component $p$ iff it contains a transition $t$ with $p \in \comp(t)$.

In \emph{fairness of groups of components}, the tasks are the \emph{sets} (or \emph{groups}) of components.
A group $G$ is \emph{enabled} in a network state iff a transition $t$ with $G = \comp(t)$ is
enabled; a path \emph{engages} in $G$ iff it contains a transition $t$ with $G = \comp(t)$.

Next to transitions and components, there is the concept of \emph{instructions}.
Let $\cal{I}$ be the set of all occurrences of subexpressions $\lambda_k!p_k$, $\lambda_k?p_k$  or
$\mu X$ in a network expression $\N$.\linebreak These subexpressions are called \emph{instructions}.
Each transition labelled $\tau$ stems from exactly one
instruction, and each transition labelled $\comm{p}{\lambda}{q}$ stems from exactly two instructions.
This yields the function $\instr$, which associates with each transition the set of
one or two instructions that gave rise to~it.

In \emph{fairness of instructions}, the instructions constitute the tasks.
  Instruction $I$ is \emph{enabled} in a network state iff a transition $t$ with $I \in \instr(t)$ is
  enabled; a path \emph{engages} in an instruction $I$ iff it contains a transition $t$ with $I \in \instr(t)$.

In \emph{fairness of synchronisations}, the tasks are the \emph{sets} of instructions, called
\emph{synchronisations}.
A synchronisation $Z$ is \emph{enabled} in a network state iff a transition $t$ with $Z = \instr(t)$ is
enabled; a path \emph{engages} in a synchronisation $Z$ iff it contains a transition $t$ with $Z = \instr(t)$.

  In \emph{fairness of actions}, the tasks are the \emph{actions}, or transition labels.
  An action $a$ is \emph{enabled} in a network state iff a transition labelled $a$ is
  enabled; a path \emph{engages} in an action $a$ iff it contains a transition $t$ labelled $a$.

For each of these notions of a task, \cite{GH19} also defines \emph{J-fairness}.
A task $T$ is said to be enabled \emph{during} a transition $u$ from network state $\N$ to $\N'$ if
$T$ is enabled in $\N$ through a transition $t$ that is concurrent with $u$ (\ie $t \conc u$).
Task $T$ is said to be \emph{continuously} enabled if it is enabled in all
network states of $\pi$ and during all transitions of $\pi$. Now a path $\pi$ is \emph{J-fair} if,
for each suffix $\pi'$ of $\pi$, each task that is continuously enabled on $\pi'$ is engaged in by $\pi'$.

Besides the 18 fairness assumptions defined above, the authors of~\cite{GH19} also consider progress (P) and
  justness (J), already defined in \sect{fairness}, as well as \emph{full fairness} (Fu),
  \emph{extreme fairness} (Ext), \emph{probabilistic fairness} (Pr), and
  \emph{strong weak fairness of instructions} (SWI). For finite-state systems (which include the
  networks in our session calculus) Fu, Ext and Pr coincide with ST\@\cite{GH19}.
  For this reason, there is no need to define these concepts here. Regarding SWI,
  say that an instruction $I$ is \emph{requested} in a network state $\N$ if it is enabled in one of
  the treads in $\N$, even if it not enabled by $\N$ itself, due to lack of a synchronisation partner.
  Now a path $\pi$ is \emph{SWI-fair} if, for each suffix $\pi'$ of $\pi$, each instruction that is
  perpetually requested and relentlessly enabled on $\pi'$ is engaged in by $\pi'$.

The following properties from \cite{GH19} trivially hold (for a given network $\N$).

\begin{enumerate}[(1)]
\item For each synchronisation $Z\subseteq \I$, and for each network state $\N$,
  there is at most one transition $t$ with $\instr(t)\mathbin=Z$ that is enabled in~$\N$.%
  \label{unique synchronisation}
\item $\I$ is finite.
  \label{finite}
\item There is a function $\cmp\!:\I\rightarrow\Ce$, where $\Ce$ is the set of components or locations
  in the network, such that
  $\comp(t)\mathbin=\{\cmp(I)\mathbin{\mid} I\mathop\in \instr(t)\}$ for all transitions $t$.
  \label{cmp}
\item\label{swi1} If an instruction $I$ is enabled in a state $\N$, it is also requested.
\item\label{swi2} If instruction $I$ is requested in network state $\N$ and $u$ is a
transition from $\N$ to $\N'$ such that $\cmp(I) \notin \comp(u)$, then $I$ is still requested in $\N'$.
\item
  If $t\mathbin{\conc} u$ with $\source(t)\mathbin=\source(u)$, then $\exists v\mathbin\in\Tr$
  with $\source(v)\mathbin=\target(u)$ and $\instr(v)\mathbin=\instr(t)$.\label{conc}%
\end{enumerate}
Given this, the classification of fairness assumptions from \cite{GH19} applies to the current setting as well,
although some of these assumptions could coincide.
The resulting lattice is shown in Figure~\ref{full taxonomy}.
Here the numbers on the edges refer to the above conditions, when these are needed for the
  indicated comparison in strength.

When using labelling in the style of
CCS, all our transitions would be labelled $\tau$, and as a consequence, $\J\A$,  $\W\A$ and $\St\A$
would collapse with P\@. But here the labelling is quite different.%
\vspace{1ex}

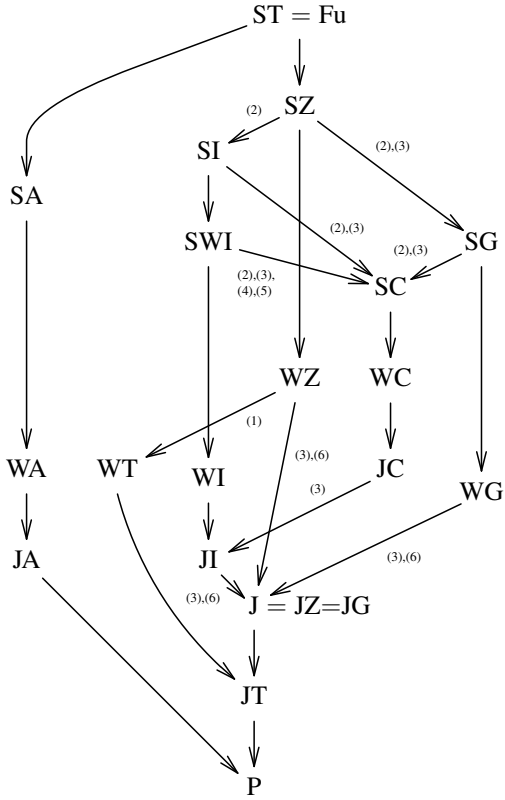
\begin{figure}[th]
\input{lattice}
\centerline{\raisebox{1ex}{\box\graph}}
\vspace{2ex}
\caption{A classification of progress, justness and fairness assumptions for finite-state systems \cite{GH19}}
\label{full taxonomy}
\end{figure}

\header{SC, SA, WZ and SWI are not as strong as SI}\mbox{}\\
The following network terminates when assuming SI, for in the only infinite execution the
$\tau$-transition belonging to instruction $\send{\seller}{\buy}$ is infinitely often enabled but
never taken. Termination is not guaranteed when assuming $\St\C$, $\St\Gr$, $\W\Sy$, SWI or $\St\A$.
\vspace{-2ex}

\begin{equation*}
\begin{array}{rl}
&
\loc{\buyer}{
 \rec{X}\left(
  \send{\seller}{\extend};X
  \oplus 
  \send{\seller}{\buy}
 \right)
}
\\
\pipar
&
\loc{\seller}{
 \rec{Y}\left(
  \recv{\buyer}{\extend} ; Y\right.\\
  &\hspace{10.5ex}\left.+\
  \recv{\buyer}{\buy} ; \send{\shipper}{\order}
 \right)
}
\\
\pipar
&
\loc{\shipper}{
  \recv{\seller}{\order} 
}
\end{array}
\end{equation*}
\vspace{1ex} 

\header{WC, WG, WI, WZ and SA are not as strong as SC}\mbox{}\\
The following network terminates when assuming SC, for in any infinite execution a
transition from $\buyer2$ is infinitely often enabled but never taken.
It does not surely terminate when assuming $\W\C$, for this transition is not perpetually enabled, due to
the $\tau$-transitions of $\seller$.
Neither is termination guaranteed when assuming $\W\Gr$, $\W\In$ or $\W\Sy$.
Furthermore, it does not surely terminate when assuming $\St\A$, because $\buyer2$ may be stuck before doing its
initial $\tau$-transition.\vspace{-2ex}

\begin{equation}\label{eg:sc-j}
\begin{array}{@{}r@{\ \,}l@{}}
&
\loc{\seller}{
 \rec{X}\left(
  \recv{\buyer1}{\order1};X
  +
  \recv{\buyer2}{\order2}
 \right)
}\hspace*{-.7em}
\\
\pipar
&
\loc{\buyer1}{
 \rec{Y}
  \send{\seller}{\order1} ; Y
}
\\
\pipar
&
\loc{\buyer2}{
  \send{\seller}{\order2}
}
\end{array}
\end{equation}
\vspace{1ex} 

\header{Collapsing Fairness Assumptions}\mbox{}
\vspace{1ex} 

\begin{proposition}{SG}
SG and SC coincide.
WG is weaker than~WC.
\end{proposition}
\begin{proof}
Let $\pi$ be an infinite path in our network that is not SG-fair.  One case is that an interaction
between $p$ and $q$ is infinitely often enabled, but never taken. By taking a suffix, we may assume
this interaction is enabled in the first state of $\pi$. W.l.o.g., let $p$ be the sending
party. Then process $p$ must be in a state of the form $\send{q}{\lambda}; P$, and it remains in that
state for the rest of $\pi$.  It follows that also component $p$ is infinitely often enabled, but
never taken.  Hence $\pi$ is not SC-fair.

The other case is that a single-component task (thus consisting of $\tau$-transitions)
is infinitely often enabled, but never taken. Also in this case it follows that $\pi$ is not SC-fair.

The other statement is obtained in the same way.
\end{proof}

\begin{proposition}{SZ}
SZ and SI coincide.
WZ is weaker than WI.
\end{proposition}
\begin{proof}
Let $\pi$ be an infinite path in our network that is not SZ-fair.  One case is that a
synchronisation between $p$ and $q$ is infinitely often enabled, but never taken. By taking a
suffix, we may assume this interaction is enabled in the first state of $\pi$. W.l.o.g., let $p$ be
the sending party. Then process $p$ must be in a state of the form $\send{q}{\lambda}; P$, and it
remains in that state for the rest of $\pi$.  It follows that also this specific instruction of $p$
is infinitely often enabled, but never taken.  Hence $\pi$ is not SI-fair.

The other case is that a single-instruction task (thus consisting of $\tau$-transitions)
is infinitely often enabled, but never taken. Also in this case it follows that $\pi$ is not SI-fair.

The other statement is obtained in the same way.
\pagebreak[3]
\end{proof}

\begin{proposition}{SA}
SA is weaker than SC.
WA is weaker than WC.
\end{proposition}
\begin{proof}
Let $\pi$ be an infinite path in our network that is not SA-fair. So on $\pi$ a transition label
$\alpha$ is infinitely often enabled, but never taken. Then $\alpha\neq\tau$, because it is easy to
show that each infinite path contains infinitely many $\tau$-transitions; in fact, on any 
path from a network state $\tau$-transitions make for at least half of all transitions. So $\alpha$ has the form
$\comm{p}{\lambda}{q}$.  In the first state of $\pi$ on which $\alpha$ is enabled, the process $p$
must be in a state of the form $\send{q}{\lambda}; P$, and it remains in that state for the rest of
$\pi$. For simplicity, we may assume that $\alpha$ is enabled in the first state of $\pi$; otherwise
we simply take a suffix.  Hence the instruction $\send{q}{\lambda}$ is perpetually requested, yet
never taken. Moreover, since $\alpha$ is infinitely often enabled, so is component $p$. Yet no
action from this component occurs on $\pi$. Hence $\pi$ is not SC-fair.

The other statement is obtained in the same way.
\end{proof}

\begin{proposition}{WC appendix}
WC coincides with J (and thus also with JC and JI).
\end{proposition}
\begin{proof}
Let $\pi$ be an infinite path in our network that is not WC-fair. So on $\pi$ a component $p$ is
perpetually enabled, but never taken. In case $p$ is stuck in a state where its next transition is a
$\tau$, then $\pi$ is not just.

In case $p$ is stuck in a state $\send{q}{\lambda}; \PP$, then, for component $p$ to be perpetually
enabled, $q$ must always be in a state $\sum_{i \in I} \recv{p_i}{\lambda_i}; \PP_i$ with $p=p_k$ and
$\lambda=\lambda_k$ for some $k\in I$. Also process $q$ must get stuck in such a state, for if $q$
keeps moving, it will at some point reach a state $\mu X. \PQ$, which is not of the above form.
Consequently, $\pi$ is not just.

The remaining case is that $p$ is stuck in a state of the form $\sum_{i \in I} \recv{p_i}{\lambda_i}; \PP_i$.
For component $p$ to be enabled, a component $p_k$ with $k\in I$ must be in a state
$\send{p}{\lambda_k}; \PP$. Again it follows that $\pi$ is not just.
\end{proof}

\begin{proposition}{WI}
WI coincides with J.
\end{proposition}
\begin{proof}
Let $\pi$ be an infinite path in our network that is not WI-fair. So on $\pi$ an instruction
$\lambda!q$, $\lambda?q$ or $\mu X$ from a process $p$ is perpetually enabled, but never taken.
In case of an instruction $\mu X$, $\pi$ is not just.

In case of an instruction $\lambda!q$, where the $\tau$-transition belonging to this transition is
never taken, $p$ must be stuck in a state $\bigoplus_{i \in I} \send{p_i}{\lambda_i}; \PP_i$ with
$q=p_k$ and $\lambda=\lambda_k$ for some $k\in I$; for if $p$ performed one of the other branches,
the instruction would (temporarily) cease to be enabled. Again $\pi$ is not just.

In case $p$ is stuck in a state $\send{q}{\lambda}; P$, then, for that instruction $\lambda!q$ to be
perpetually enabled, $q$ must always be in a state $\sum_{i \in I} \recv{p_i}{\lambda_i}; \PP_i$ with
$p=p_k$ and $\lambda=\lambda_k$ for some $k\in I$. Also process $q$ must get stuck in such a state,
for if $q$ keeps moving, it will at some point reach a state $\mu X. \PQ$, which is not of the above
form.  Consequently, $\pi$ is not just.

The remaining case is that $p$ is stuck in a state of the form $\sum_{i \in I} \recv{p_i}{\lambda_i}; \PP_i$.
For component $p$ to be enabled, a component $p_k$ with $k\in I$ must be in a state
$\send{p}{\lambda_k}; \PP$. Again it follows that $\pi$ is not just.
\end{proof}

\begin{proposition}{WA}
WA coincides with JA.
\end{proposition}
\begin{proof}
Let $\pi$ be an infinite path in our network that is not WA-fair. So on $\pi$ (possibly after taking
a suffix) a transition label $\alpha$ is perpetually enabled, but never taken. Then
$\alpha\mathbin{\neq}\tau$, as in the proof of \pr{SA}.  So $\alpha$ has the form
$\comm{p}{\lambda}{q}$.  In the first state of $\pi$, the process $p$ must be in a state of the form
$\send{q}{\lambda}; P$, and it remains in that state for the rest of $\pi$. Since $\alpha$ is
perpetually enabled, $q$ must always be in a state $\sum_{i \in I} \recv{p_i}{\lambda_i}; \PP_i$ with
$p=p_k$ and $\lambda=\lambda_k$ for some $k\in I$. Also process $q$ must get stuck in such a state,
for if $q$ keeps moving, it will at some point reach a state $\mu X. \PQ$, which is not of the above
form.  Consequently, the action $\alpha$ is continuously enabled on $\pi$, and $\pi$ is not JA-fair.
\end{proof}

\begin{proposition}{WT}
WT coincides with P.
\end{proposition}
\begin{proof}
Since our syntax does not allow self-loops (considering that unfolding recursion takes a
$\tau$-transition) on no infinite path a transition can be perpetually enabled.
\end{proof}

\begin{proposition}{SWI}
SWI coincides with SC.
\end{proposition}
\begin{proof}
Let $\pi$ be an infinite path in our network that is not SWI-fair. So on $\pi$ an instruction
$\lambda!q$, $\lambda?q$ or $\mu X$ from a process $p$ is perpetually requested and infinitely often
enabled, but never taken. In case of an instruction $\mu X$, $\pi$ is not just, and thus certainly
not SC-fair.

In case of an instruction $\lambda!q$, where the $\tau$-transition belonging to this transition is
never taken, $p$ must be stuck in a state $\bigoplus_{i \in I} \send{p_i}{\lambda_i}; \PP_i$ with
$q=p_k$ and $\lambda=\lambda_k$ for some $k\in I$; for if $p$ performed one of the other branches,
the instruction would (temporarily) cease to be enabled. Again $\pi$ is not just.

If  $p$ is stuck in a state $\send{q}{\lambda}; P$, then component $p$ is infinitely often
enabled, but never taken. Hence $\pi$ is not SC-fair.

In case of an instruction $\lambda?q$, $p$ must be stuck in a state
$\sum_{i \in I} \recv{p_i}{\lambda_i}; \PP_i$ with $q=p_k$ and $\lambda=\lambda_k$ for some $k\in I$;
if it leaves this state, it reaches a state in which that very instruction $\lambda?q$ is no longer
requested. Again  component $p$ is infinitely often
enabled, but never taken. Hence $\pi$ is not SC-fair.
\end{proof}

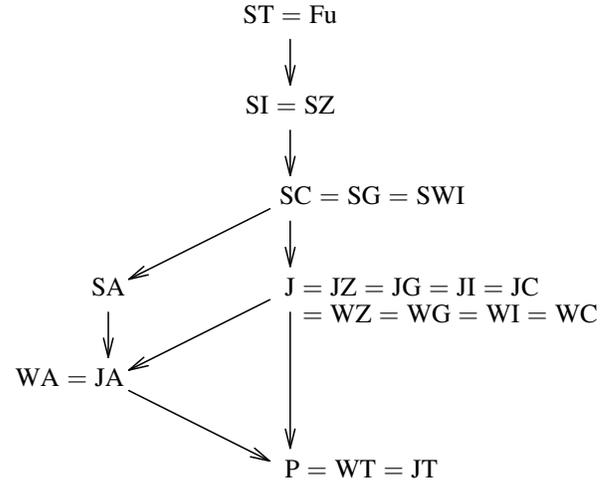
\begin{figure}[t]
\hspace{-10ex}\input{lattice2}
\centerline{\raisebox{1ex}{\box\graph}}
\vspace{2ex}
\caption{
{ A classification of fairness assumptions for our session calculus }
}
\label{appendix taxonomy}
\end{figure}

\header{\texorpdfstring{$\Pr$, $\W\T$ and $\St\A$}{P, WT and SA} are not as strong as J}
In the following network the accountant terminates when assuming $\J$,
but when merely assuming $\Pr$ or $\W\T$ this is not guaranteed, since there is no single
\accountant-transition that is perpetually enabled.
It does not surely terminate when assuming $\St\A$ either, because $\accountant$ may be stuck before
doing its initial $\tau$-transition.\vspace{-2ex}

\begin{equation*}\label{eg:jt-j}
\begin{array}{rl}
&
\loc{\trader1}{
 \rec{X}
  \send{\trader2}{\order}; \recv{\trader2}{\order}; X
}
\\
\pipar
&
\loc{\trader2}{
 \rec{Y}
  \recv{\trader1}{\order} ;   \send{\trader1}{\order} ;  Y
}
\\
\pipar
&
\loc{\accountant}{
 \send{\auditor}{\report}
}
\\
\pipar
&
\loc{\auditor}{
 \recv{\accountant}{\report} 
}
\end{array}
\end{equation*}
The following network shows exactly the same, due to the initial $\tau$-transition of $\buyer$.

\begin{equation}\label{eg:p-jt}
\begin{array}{rl}
&
\loc{\seller}{
 \rec{X}
  \recv{\buyer}{\order};X
}
\\
\pipar
&
\loc{\buyer}{
 \rec{Y}
  \send{\seller}{\order} ; Y
}
\\
\pipar
&
\loc{\accountant}{
 \send{\auditor}{\report}
}
\\
\pipar
&
\loc{\auditor}{
 \recv{\accountant}{\report} 
}
\end{array}
\end{equation}
\vspace{1ex} 

\header{P is not as strong as JA}
In Network (\ref{eg:p-jt}), the path in which the $\tau$-action of $\accountant$ occurs, but the
$\report$-action does not, is progressing, but not JA-fair.
\vspace{1ex} 

\header{J and WA are not as strong as SA}
In Network (\ref{eg:sc-j}), the path in which the $\tau$-action of $\buyer2$ occurs, but the
$\order2$-action does not, is just, as well as WA-fair, but not SA-fair.%
\vspace{1ex} 

\header{SZ is not as strong as ST}
The network in \ex{p-jt:early} from the introduction has 4 states and 8 transitions.
ST insists that in a fair run each of these transitions occurs,
whereas SZ allows a run that alternates regularly between a $\buyer1/\seller1$- and an
$\buyer2/\seller2$-interaction.
\vspace{1ex} 

It follows that our classification collapses as indicated in Figure~\ref{appendix taxonomy}.

%% file: lattice.tex
\expandafter\ifx\csname graph\endcsname\relax
   \csname newbox\expandafter\endcsname\csname graph\endcsname
\fi
\ifx\graphtemp\undefined
  \csname newdimen\endcsname\graphtemp
\fi
\expandafter\setbox\csname graph\endcsname
 =\vtop{\vskip 0pt\hbox{%
    \graphtemp=.5ex
    \advance\graphtemp by 3.214in
    \rlap{\kern 1.310in\lower\graphtemp\hbox to 0pt{\hss \makebox[4pt][l]{$\J = \J\Sy {=} \J\Gr$}\hss}}%
    \graphtemp=.5ex
    \advance\graphtemp by 3.690in
    \rlap{\kern 1.310in\lower\graphtemp\hbox to 0pt{\hss JT\hss}}%
\pdfliteral{
q [] 0 d 1 J 1 j
0.576 w
0.576 w
96.12 -249.84 m
94.32 -257.04 l
S
92.52 -249.84 m
94.32 -257.04 l
S
94.32 -239.976 m
94.32 -257.04 l
S
Q
}%
    \graphtemp=.5ex
    \advance\graphtemp by 4.167in
    \rlap{\kern 1.310in\lower\graphtemp\hbox to 0pt{\hss P\hss}}%
\pdfliteral{
q [] 0 d 1 J 1 j
0.576 w
96.12 -284.112 m
94.32 -291.312 l
S
92.52 -284.112 m
94.32 -291.312 l
S
94.32 -274.32 m
94.32 -291.312 l
S
Q
}%
    \graphtemp=.5ex
    \advance\graphtemp by 1.071in
    \rlap{\kern 0.119in\lower\graphtemp\hbox to 0pt{\hss SA\hss}}%
    \graphtemp=.5ex
    \advance\graphtemp by 2.500in
    \rlap{\kern 0.119in\lower\graphtemp\hbox to 0pt{\hss WA\hss}}%
\pdfliteral{
q [] 0 d 1 J 1 j
0.576 w
10.368 -164.088 m
8.568 -171.288 l
S
6.768 -164.088 m
8.568 -171.288 l
S
8.568 -85.68 m
8.568 -171.288 l
S
Q
}%
    \graphtemp=.5ex
    \advance\graphtemp by 2.976in
    \rlap{\kern 0.119in\lower\graphtemp\hbox to 0pt{\hss JA\hss}}%
\pdfliteral{
q [] 0 d 1 J 1 j
0.576 w
10.368 -198.36 m
8.568 -205.56 l
S
6.768 -198.36 m
8.568 -205.56 l
S
8.568 -188.568 m
8.568 -205.56 l
S
84.312 -287.496 m
88.128 -293.832 l
S
81.792 -290.016 m
88.128 -293.832 l
S
14.616 -220.32 m
88.128 -293.832 l
S
Q
}%
    \graphtemp=.5ex
    \advance\graphtemp by 2.500in
    \rlap{\kern 0.595in\lower\graphtemp\hbox to 0pt{\hss WT\hss}}%
\pdfliteral{
q [] 0 d 1 J 1 j
0.576 w
82.728 -252.216 m
86.904 -258.336 l
S
80.352 -254.952 m
86.904 -258.336 l
S
86.904 -258.336 m
81.648 -253.656 l
S
81.609176 -253.668686 m
62.826176 -236.043581 49.385617 -213.485717 42.826165 -188.577505 c
S
Q
}%
    \graphtemp=.5ex
    \advance\graphtemp by 0.833in
    \rlap{\kern 1.071in\lower\graphtemp\hbox to 0pt{\hss SI\hss}}%
    \graphtemp=.5ex
    \advance\graphtemp by 1.310in
    \rlap{\kern 1.071in\lower\graphtemp\hbox to 0pt{\hss SWI\hss}}%
\pdfliteral{
q [] 0 d 1 J 1 j
0.576 w
78.912 -78.408 m
77.112 -85.608 l
S
75.312 -78.408 m
77.112 -85.608 l
S
77.112 -68.544 m
77.112 -85.608 l
S
Q
}%
    \graphtemp=.5ex
    \advance\graphtemp by 2.548in
    \rlap{\kern 1.071in\lower\graphtemp\hbox to 0pt{\hss WI\hss}}%
\pdfliteral{
q [] 0 d 1 J 1 j
0.576 w
78.912 -167.544 m
77.112 -174.744 l
S
75.312 -167.544 m
77.112 -174.744 l
S
77.112 -102.888 m
77.112 -174.744 l
S
Q
}%
    \graphtemp=.5ex
    \advance\graphtemp by 2.976in
    \rlap{\kern 1.071in\lower\graphtemp\hbox to 0pt{\hss JI\hss}}%
\pdfliteral{
q [] 0 d 1 J 1 j
0.576 w
78.912 -198.36 m
77.112 -205.56 l
S
75.312 -198.36 m
77.112 -205.56 l
S
77.112 -192.024 m
77.112 -205.56 l
S
Q
}%
    \graphtemp=.5ex
    \advance\graphtemp by 3.152in
    \rlap{\kern 1.043in\lower\graphtemp\hbox to 0pt{\hss \tiny(\ref*{cmp}),(\ref*{conc})\hss}}%
\pdfliteral{
q [] 0 d 1 J 1 j
0.576 w
86.76 -221.328 m
90.576 -227.736 l
S
84.168 -223.92 m
90.576 -227.736 l
S
82.008 -219.168 m
90.576 -227.736 l
S
Q
}%
    \graphtemp=.5ex
    \advance\graphtemp by 0.595in
    \rlap{\kern 1.548in\lower\graphtemp\hbox to 0pt{\hss SZ\hss}}%
\pdfliteral{
q [] 0 d 1 J 1 j
0.576 w
92.16 -54.504 m
84.888 -56.088 l
S
90.576 -51.264 m
84.888 -56.088 l
S
103.752 -46.656 m
84.888 -56.088 l
S
Q
}%
    \graphtemp=.5ex
    \advance\graphtemp by 0.595in
    \rlap{\kern 1.310in\lower\graphtemp\hbox to 0pt{\hss \tiny(\ref*{finite})\hss}}%
    \graphtemp=.5ex
    \advance\graphtemp by 0.786in
    \rlap{\kern 2.024in\lower\graphtemp\hbox to 0pt{\hss \,\tiny(\ref*{finite}),(\ref*{cmp})\hss}}%
    \graphtemp=.5ex
    \advance\graphtemp by 2.024in
    \rlap{\kern 1.548in\lower\graphtemp\hbox to 0pt{\hss WZ\hss}}%
\pdfliteral{
q [] 0 d 1 J 1 j
0.576 w
113.256 -129.816 m
111.456 -137.016 l
S
109.656 -129.816 m
111.456 -137.016 l
S
111.456 -51.408 m
111.456 -137.016 l
S
Q
}%
    \graphtemp=.5ex
    \advance\graphtemp by 2.405in
    \rlap{\kern 1.614in\lower\graphtemp\hbox to 0pt{\hss \tiny(\ref*{cmp}),(\ref*{conc})\hss}}%
\pdfliteral{
q [] 0 d 1 J 1 j
0.576 w
60.912 -172.944 m
53.712 -174.6 l
S
59.328 -169.776 m
53.712 -174.6 l
S
102.24 -150.336 m
53.712 -174.6 l
S
99.144 -216.216 m
95.976 -222.912 l
S
95.616 -215.496 m
95.976 -222.912 l
S
109.728 -154.152 m
95.976 -222.912 l
S
Q
}%
    \graphtemp=.5ex
    \advance\graphtemp by 2.224in
    \rlap{\kern 1.310in\lower\graphtemp\hbox to 0pt{\hss \tiny(\ref*{unique synchronisation})\hss}}%
    \graphtemp=.5ex
    \advance\graphtemp by 1.548in
    \rlap{\kern 2.024in\lower\graphtemp\hbox to 0pt{\hss SC\hss}}%
    \graphtemp=.5ex
    \advance\graphtemp by 1.214in
    \rlap{\kern 1.786in\lower\graphtemp\hbox to 0pt{\hss \,\tiny(\ref*{finite}),(\ref*{cmp})\hss}}%
    \graphtemp=.5ex
    \advance\graphtemp by 1.333in
    \rlap{\kern 2.119in\lower\graphtemp\hbox to 0pt{\hss \,\tiny(\ref*{finite}),(\ref*{cmp})\hss}}%
\pdfliteral{
q [] 0 d 1 J 1 j
0.576 w
134.064 -100.44 m
138.744 -106.2 l
S
131.904 -103.32 m
138.744 -106.2 l
S
84.024 -65.16 m
138.744 -106.2 l
S
Q
}%
    \graphtemp=.5ex
    \advance\graphtemp by 1.462in
    \rlap{\kern 1.310in\lower\graphtemp\hbox to 0pt{\hss \,\tiny(\ref*{finite}),(\ref*{cmp}),\hss}}%
    \graphtemp=.5ex
    \advance\graphtemp by 1.548in
    \rlap{\kern 1.310in\lower\graphtemp\hbox to 0pt{\hss \,\tiny(\ref*{swi1}),(\ref*{swi2})\phantom,\hss}}%
\pdfliteral{
q [] 0 d 1 J 1 j
0.576 w
130.824 -105.624 m
137.304 -109.224 l
S
129.888 -109.152 m
137.304 -109.224 l
S
88.848 -96.48 m
137.304 -109.224 l
S
Q
}%
    \graphtemp=.5ex
    \advance\graphtemp by 2.024in
    \rlap{\kern 2.024in\lower\graphtemp\hbox to 0pt{\hss WC\hss}}%
\pdfliteral{
q [] 0 d 1 J 1 j
0.576 w
147.528 -129.816 m
145.728 -137.016 l
S
143.928 -129.816 m
145.728 -137.016 l
S
145.728 -120.024 m
145.728 -137.016 l
S
Q
}%
    \graphtemp=.5ex
    \advance\graphtemp by 2.500in
    \rlap{\kern 2.024in\lower\graphtemp\hbox to 0pt{\hss JC\hss}}%
\pdfliteral{
q [] 0 d 1 J 1 j
0.576 w
147.528 -164.088 m
145.728 -171.288 l
S
143.928 -164.088 m
145.728 -171.288 l
S
145.728 -154.296 m
145.728 -171.288 l
S
Q
}%
    \graphtemp=.5ex
    \advance\graphtemp by 2.581in
    \rlap{\kern 1.643in\lower\graphtemp\hbox to 0pt{\hss \tiny(\ref*{cmp})\hss}}%
\pdfliteral{
q [] 0 d 1 J 1 j
0.576 w
92.16 -208.8 m
84.888 -210.384 l
S
90.576 -205.56 m
84.888 -210.384 l
S
138.024 -183.816 m
84.888 -210.384 l
S
Q
}%
    \graphtemp=.5ex
    \advance\graphtemp by 1.310in
    \rlap{\kern 2.500in\lower\graphtemp\hbox to 0pt{\hss SG\hss}}%
\pdfliteral{
q [] 0 d 1 J 1 j
0.576 w
160.704 -105.912 m
153.504 -107.568 l
S
159.12 -102.744 m
153.504 -107.568 l
S
172.368 -98.136 m
153.504 -107.568 l
S
168.336 -83.304 m
173.016 -89.064 l
S
166.176 -86.184 m
173.016 -89.064 l
S
118.296 -48.024 m
173.016 -89.064 l
S
Q
}%
    \graphtemp=.5ex
    \advance\graphtemp by 2.619in
    \rlap{\kern 2.500in\lower\graphtemp\hbox to 0pt{\hss WG\hss}}%
\pdfliteral{
q [] 0 d 1 J 1 j
0.576 w
181.8 -172.656 m
180 -179.856 l
S
178.2 -172.656 m
180 -179.856 l
S
180 -102.888 m
180 -179.856 l
S
Q
}%
    \graphtemp=.5ex
    \advance\graphtemp by 2.938in
    \rlap{\kern 2.095in\lower\graphtemp\hbox to 0pt{\hss \tiny(\ref*{cmp}),(\ref*{conc})\hss}}%
\pdfliteral{
q [] 0 d 1 J 1 j
0.576 w
107.856 -225.216 m
100.584 -226.728 l
S
106.272 -221.976 m
100.584 -226.728 l
S
172.296 -192.312 m
100.584 -226.728 l
S
Q
}%
    \graphtemp=.5ex
    \advance\graphtemp by 0.119in
    \rlap{\kern 1.548in\lower\graphtemp\hbox to 0pt{\hss $\St\T = \Fu$\hss}}%
\pdfliteral{
q [] 0 d 1 J 1 j
0.576 w
113.256 -27 m
111.456 -34.2 l
S
109.656 -27 m
111.456 -34.2 l
S
111.456 -17.136 m
111.456 -34.2 l
S
10.368 -61.272 m
8.568 -68.472 l
S
6.768 -61.272 m
8.568 -68.472 l
S
91.512 -12.024 m
50.04 -27.432 l
21.83904 -37.90944 8.568 -46.94112 8.568 -55.656 c
8.568 -68.472 l
S
Q
}%
    \hbox{\vrule depth4.167in width0pt height 0pt}%
    \kern 2.619in
  }%
}%

%% file: lattice2.tex
\expandafter\ifx\csname graph\endcsname\relax
   \csname newbox\expandafter\endcsname\csname graph\endcsname
\fi
\ifx\graphtemp\undefined
  \csname newdimen\endcsname\graphtemp
\fi
\expandafter\setbox\csname graph\endcsname
 =\vtop{\vskip 0pt\hbox{%
    \graphtemp=.5ex
    \advance\graphtemp by 1.548in
    \rlap{\kern 1.071in\lower\graphtemp\hbox to 0pt{\hss \makebox[4pt][l]{$\J = \J\Sy = \J\Gr = \J\In = \J\C$}\hss}}%
    \graphtemp=.5ex
    \advance\graphtemp by 1.690in
    \rlap{\kern 1.071in\lower\graphtemp\hbox to 0pt{\hss \makebox[4pt][l]{$\phantom{\J} = \W\Sy = \W\Gr = \W\In = \W\C$}\hss}}%
    \graphtemp=.5ex
    \advance\graphtemp by 2.500in
    \rlap{\kern 1.071in\lower\graphtemp\hbox to 0pt{\hss \makebox[4pt][l]{$\textrm{P} = \W\T = \J\T$}\hss}}%
\pdfliteral{
q [] 0 d 1 J 1 j
0.576 w
0.576 w
78.912 -164.088 m
77.112 -171.288 l
S
75.312 -164.088 m
77.112 -171.288 l
S
77.112 -120.024 m
77.112 -171.288 l
S
Q
}%
    \graphtemp=.5ex
    \advance\graphtemp by 1.548in
    \rlap{\kern 0.119in\lower\graphtemp\hbox to 0pt{\hss SA\hss}}%
    \graphtemp=.5ex
    \advance\graphtemp by 2.024in
    \rlap{\kern 0.119in\lower\graphtemp\hbox to 0pt{\hss \makebox[12pt][r]{$\W\A = \J\A$}\hss}}%
\pdfliteral{
q [] 0 d 1 J 1 j
0.576 w
10.368 -129.816 m
8.568 -137.016 l
S
6.768 -129.816 m
8.568 -137.016 l
S
8.568 -120.024 m
8.568 -137.016 l
S
23.616 -140.184 m
16.344 -141.84 l
S
21.96 -137.016 m
16.344 -141.84 l
S
69.48 -115.272 m
16.344 -141.84 l
S
63.72 -171.288 m
69.336 -176.112 l
S
62.136 -174.528 m
69.336 -176.112 l
S
16.272 -149.544 m
69.336 -176.112 l
S
Q
}%
    \graphtemp=.5ex
    \advance\graphtemp by 0.595in
    \rlap{\kern 1.071in\lower\graphtemp\hbox to 0pt{\hss $\St\In = \St\Sy$\hss}}%
    \graphtemp=.5ex
    \advance\graphtemp by 1.071in
    \rlap{\kern 1.071in\lower\graphtemp\hbox to 0pt{\hss \makebox[8pt][l]{$\St\C = \St\Gr = \textrm{SWI}$}\hss}}%
\pdfliteral{
q [] 0 d 1 J 1 j
0.576 w
78.912 -61.272 m
77.112 -68.472 l
S
75.312 -61.272 m
77.112 -68.472 l
S
77.112 -51.408 m
77.112 -68.472 l
S
78.912 -95.544 m
77.112 -102.744 l
S
75.312 -95.544 m
77.112 -102.744 l
S
77.112 -85.68 m
77.112 -102.744 l
S
23.616 -105.912 m
16.344 -107.568 l
S
21.96 -102.744 m
16.344 -107.568 l
S
69.48 -81 m
16.344 -107.568 l
S
Q
}%
    \graphtemp=.5ex
    \advance\graphtemp by 0.119in
    \rlap{\kern 1.071in\lower\graphtemp\hbox to 0pt{\hss $\St\T = \Fu$\hss}}%
\pdfliteral{
q [] 0 d 1 J 1 j
0.576 w
78.912 -27 m
77.112 -34.2 l
S
75.312 -27 m
77.112 -34.2 l
S
77.112 -17.136 m
77.112 -34.2 l
S
Q
}%
    \hbox{\vrule depth2.500in width0pt height 0pt}%
    \kern 1.190in
  }%
}%

%% file: appB.tex
\subsection{Collapsing notions of lock-freedom}\label{app:B}
\begin{proposition}{SA=P live}
$\Live{\St\A}$ coincides with $\Live{\textrm{P}}$.
\end{proposition}
\begin{proof}
Using the results depicted in Figure~\ref{appendix taxonomy} 
and Proposition~\ref{pr:prop2}, $\Live{P}\Rightarrow\Live{\St\A}$.

Hence it suffices to show that if a network has a progressing path $\pi$ that lacks the property of
\df{live}, then it has an $\St\A$-fair path $\rho$ that lacks this property.
In case $\pi$ is finite, we choose $\rho$ as $\pi$. 

In case $\pi$ is infinite, we define the $\St\A$-unfairness of $\pi$ as the number of different labels $\alpha$ such that 
label $\alpha$ is infinitely often enabled on $\pi$, but from some point onwards never taken.
This must be a finite number,  and if it is 0 then $\pi$ is $\St\A$-fair.
It now suffices to show that if the $\St\A$-unfairness of $\pi$ is positive, then we can modify
$\pi$ into a path $\pi'$ whose $\St\A$-unfairness is strictly smaller, and that still lacks the
property of \df{live}.

Let $\alpha$ be infinitely often enabled on $\pi$, but from some point onwards never taken.
As pointed out in the proof of \pr{SA}, $\alpha\mathbin{\neq}\tau$.
So $\alpha$ has the form $\comm{p}{\lambda}{q}$. 
In the first state of $\pi$ on which $\alpha$ is enabled, but not taken past that state, the
process $p$ must be in a state of the form $\send{q}{\lambda}; P$, and it remains in that state for
the rest of $\pi$. Now $\pi$ can be modified into $\pi'$ by skipping the last $\tau$-transition
belonging to the instruction $\send{q}{\lambda}; P$ of $p$. This strictly decreases its $\St\A$-unfairness.

Since $\pi$ fails the property of \df{live}, there must be a location $p$ such that $p$ does not
terminate on $\pi$, and $\pi$ contains only finitely many  transitions that stem from component $p$.
Now $p$ does not terminate on $\pi'$ either, and also $\pi'$ contains only finitely many  transitions
that stem from component $p$.
\end{proof}

\begin{proposition}{ST live}
$\Live{\St\T}$ coincides with $\Live{\St\In}$.
\vspace{1ex} 
\end{proposition}
\emph{Proof sketch:} 
Proposition~\ref{pr:prop2} implies $\Live{\St\In} \Rightarrow \Live{\St\T}$.

For the other direction it suffices to show that any $\St\In$-fair path $\pi$ that lacks the property of
\df{live} can be converted into an $\St\T$-fair path $\rho$ that lacks this property.
This can be achieved by swapping concurrent transitions. 
\hfill \rule{7pt}{7pt}
\vspace{1ex} 

Consequently, the 7 different fairness assumptions
collapse to 4 different liveness properties, displayed in Figure~\ref{liveness taxonomy}.

\begin{figure}[t]
\input{lattice3}
\centerline{\raisebox{1ex}{\box\graph}}
\vspace{2ex}
\caption{A classification of liveness properties}
\label{liveness taxonomy}
\end{figure}
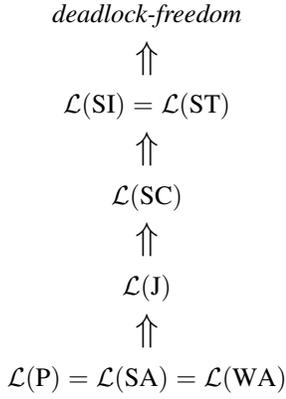

%% file: lattice3.tex
\expandafter\ifx\csname graph\endcsname\relax
   \csname newbox\expandafter\endcsname\csname graph\endcsname
\fi
\ifx\graphtemp\undefined
  \csname newdimen\endcsname\graphtemp
\fi
\expandafter\setbox\csname graph\endcsname
 =\vtop{\vskip 0pt\hbox{%
    \graphtemp=.5ex
    \advance\graphtemp by 1.905in
    \rlap{\kern 0.000in\lower\graphtemp\hbox to 0pt{\hss $\Live{\textrm{P}} = \Live{\St\A} = \Live{\W\A}$\hss}}%
    \graphtemp=.5ex
    \advance\graphtemp by 1.429in
    \rlap{\kern 0.000in\lower\graphtemp\hbox to 0pt{\hss $\Live{\textrm{J}}$\hss}}%
    \graphtemp=.5ex
    \advance\graphtemp by 1.667in
    \rlap{\kern 0.000in\lower\graphtemp\hbox to 0pt{\hss {\Large $\Uparrow$}\hss}}%
    \graphtemp=.5ex
    \advance\graphtemp by 1.190in
    \rlap{\kern 0.000in\lower\graphtemp\hbox to 0pt{\hss {\Large $\Uparrow$}\hss}}%
    \graphtemp=.5ex
    \advance\graphtemp by 0.714in
    \rlap{\kern 0.000in\lower\graphtemp\hbox to 0pt{\hss {\Large $\Uparrow$}\hss}}%
    \graphtemp=.5ex
    \advance\graphtemp by 0.238in
    \rlap{\kern 0.000in\lower\graphtemp\hbox to 0pt{\hss {\Large $\Uparrow$}\hss}}%
    \graphtemp=.5ex
    \advance\graphtemp by 0.952in
    \rlap{\kern 0.000in\lower\graphtemp\hbox to 0pt{\hss $\Live{\textrm{SC}}$\hss}}%
    \graphtemp=.5ex
    \advance\graphtemp by 0.476in
    \rlap{\kern 0.000in\lower\graphtemp\hbox to 0pt{\hss $\Live{\textrm{SI}} = \Live{\textrm{ST}}$\hss}}%
    \graphtemp=.5ex
    \advance\graphtemp by 0.000in
    \rlap{\kern 0.000in\lower\graphtemp\hbox to 0pt{\hss \textit{deadlock-freedom}\hss}}%
    \hbox{\vrule depth1.905in width0pt height 0pt}%
    \kern 0.000in
  }%
}%

%% file: appC.tex
\subsection[Padovani's lock-freedom coincides with  L(ST)]
           {Padovani's lock-freedom coincides with  $\Live{\St\T}$}\label{app:C}

This appendix contains a proof of Theorem~\ref{thm:lock-free}: a network is Padovani lock-free iff it satisfies $\Live{\St\T}$.
\begin{proof}
Suppose $\N\models\Live{\St\T}$. We show it is lock-free. Let $\M$ be a reachable state of $\N$.
Take a path from $\N$ to $\M$, and extend it to an $\St\T$-fair path $\pi$. This is possible by
Theorem 6.1 of \cite{GH19}, saying that $\St\T$-fairness is feasible.
The suffix $\pi'$  of $\pi$ starting at $\M$ satisfies the property required by
Padovani's lock-freedom (\df{lock-free}), that is, for a given location $p$ of $\M$
  such that $\proc(p,\M)\neq\End$, $\pi'$ contains a transition involving $p$. 
Interestingly, the choice of the path required by \df{lock-free} turns out to be
  independent of $p$.

Now suppose $\N$ is lock-free, and let $\pi$ be an $\St\T$-fair path.
Let $p$ be a location of $\N$ such that $\pi$ does not contain a state of the form $\N' \pipar \loc{p}{\End}$.
We need to show that $\pi$ contains infinitely many transitions that stems from component $p$.

For each state $\M$ on $\pi$, let $d(\M)>0$ be the length of the shortest path from $\M$ that contains
a transition from component $p$; such a shortest path exists since $\N$ is lock-free.
Since $\N$ is a finite-state system, there is a state $\M$ that occurs infinitely often on $\pi$.
In case $d(\M)>1$ there must be a transition $\M \goesto\pi \M'$ such that $d(\M')<d(\M)$.
Since this transition is enabled on $\pi$ infinitely often, and $\pi$ is ST-fair, this transition
must be taken infinitely often, and hence also $\M'$ occurs infinitely often in $\pi'$. So by a
trivial induction there is a state $\M''$ with $d(\M'')=1$ that occurs infinitely often in $\pi$.
This state has an outgoing transition that stems from component $p$. Since this transition is
enabled on $\pi$ infinitely often, and $\pi$ is ST-fair, it must be taken infinitely often.
\end{proof}

%% file: appE.tex
\subsection{Proof of completeness\label{app:E}}
This appendix contains a proof of Theorem~\ref{thm:guarded typing completeness}: 
if $\N\models\Live{\J}$, then $\N$ is guardedly well-typed.

\begin{proof} The proof is staged as a series of claims.
\vspace{1ex plus 2pt} 

\noindent
\textit{Claim 1:} Let $\M$ be unfolded.\footnote{See Page~\pageref{pg:netunfolded} for a definition.}
If $\N \dgoesto{\comm{p}{\lambda}{q}~} \M$ then $\proc(q,\N) = \sum_{i \in I} \recv{p_i}{\lambda_i} ; \PQ_i$
with $p=p_j$, $\proc(q,\M) = \PQ_j$ and $\lambda=\lambda_j$ for some $j \in I$.
Moreover, if $r \notin \{p,q\}$ then $\proc(r,\N) = \proc(r,\M)$.
\vspace{1ex plus 2pt}

\noindent
\textit{Proof:} Directly from the definition of the reactive semantics (Figure \ref{fig:reactive}).
\hfill \rule{7pt}{7pt}
\vspace{1ex plus 2pt}

\noindent
\textit{Claim 2:} The algorithm $\gt$ (Figure~\ref{fig:algorithm}) always terminates.
\vspace{1ex plus 2pt}

\noindent
\textit{Proof:} In a run on which $\gt$ does not terminate, along at least one branch an unbounded
history $h$ is created. Since there are only finitely many reachable states, some state $\M$ must
occur unboundedly in $h$. Each time this state is encountered, except for the first time, the
fifth clause of $\gt$ applies. However, each time a different location $p$ that did not already
occur in $h \upharpoonright \M$ is added to $h$. Since there are only finitely many locations, this cannot go on
forever. \hfill \rule{7pt}{7pt}
\vspace{1ex plus 2pt}

An expression $\gt(h,\N)$ may have free occurrences of variables $X_\M$.
It is easy to check that if $X_\M$ has a free occurrence\pagebreak[4] in $\gt(h,\N)$, then $\M$ occurs in $h$.
We define a closed version $\gt^*(h,\N)$ of $\gt(h,\N)$,
obtained from $\gt(h,\N)$ by unfolding recursion.
The definition proceeds by induction on the length of $h$. 
Here, $\textrm{fv}(\G)$ denotes
the set of free recursion variables in a global type expression $\G$.
\[
\gt^*(h,\N) 
:= \gt(h,\N)\msub{X_\M}{\gt^*(h \upharpoonleft \M,\M)}{\!X_\M \in \textrm{fv}( \gt(h,\N) )\! }\]
Note that $\gt(\varepsilon,\N) = \gt^*(\varepsilon,\N)$.
By induction, $\gt^*(h,\N)$ is a closed session type expression.

$\gt(h,\N)$, and hence also $\gt^*(h,\N)$, always yields a valid global session type expression,
except that it may contain the constant \textsc{deadlock}.
\vspace{1ex plus 2pt}

\noindent
\textit{Claim 3:} If $\gt(h,\N)$ contains the constant \textsc{deadlock}, then $\N$ is not deadlock-free.
\vspace{1ex plus 2pt}

\noindent
\textit{Proof:} If $\gt(h,\N)$ contains the constant \textsc{deadlock}, then for some
unfolded network $\M$
reachable from $\N$ and for some extension $h'$ of $h$ we have $\gt(h',\M) = \mbox{\sc deadlock}$.
It suffices to show that $\M$ has a deadlock.
Since no location is ready in $\M$, for each location
$p$ of $\M$ with $\proc(p,\M) = \bigoplus_{i\in I}\send{q_i}{\lambda_i} ; \PP_i$ -- let us call such a
location \emph{active} in $\M$ -- there exists an \mbox{$i \mathbin\in I$} such that
there is no transition $\M \mathbin{\dgoesto{\comm{p}{\lambda_i}{q_i}~}} \M_i$. Let $i_p$ be this $i$.
Now $\M$ admits a sequence of $\tau$-transitions to a state $\M'$ in which
$\proc(p,\M') = \chosen[3pt]{\send{q_{i_p}}{\lambda_{i_p}} ; \PP_{i_p}}$ for each $p$ active in $\M$,
and no further $\tau$-transitions are possible from $\M'$. 
The only transitions that $\M'$ could
possibly do must have a label $\comm{p}{\lambda_{i_p}}{q_{i_p}}$ for some $p$ active in $\M$, yet none of
these transitions are actually possible. Hence $\M'$ is a deadlock. \hfill \rule{7pt}{7pt}
\vspace{1ex plus 2pt}

\noindent
Call a history $h$ \emph{reachable} from $\N$ iff all networks $\M$ that occur in $h$ are reachable
from $\N$.
\vspace{1ex plus 2pt}

\noindent
\textit{Claim 4:} If $\N$ is deadlock-free, and $h$ and $\M$ are reachable from $\N$, then
$\gt^*(h,\M)$ does not contain \textsc{deadlock}.
\vspace{1ex plus 2pt}

\noindent
\textit{Proof:} If $\gt^*(h,\M)$ contains the constant \textsc{deadlock}, then either $\gt(h,\M)$ contains \textsc{deadlock}, or $\gt^*(h',\M')$ does, for a proper prefix $h'$ of $h$ and a network $\M'$
that occurs in $h$. The previous claim and a simple induction on the length of $h$ finish the proof.
 \hfill \rule{7pt}{7pt}
\vspace{1ex plus 2pt}

\noindent
\textit{Claim 5:}  Let $\N \models\Live{\J}$, let $r$ be a location of $\N$, and $h$ and $\M$ be
reachable from $\N$, with $\M$ unfolded. \\
If $r \notin \participants{\gt^*(h,\M)}$, then $\proc(r,\M) = \End$.
\vspace{1ex plus 2pt}

\noindent
\textit{Proof:} For each pair $(\ell,\NL)$ of a history $\ell$ and a network state $\NL$, both reachable
from $\N$, such that $\proc(r,\NL) \mathbin{\neq} \End$, we select a unique successor pair $(\ell',\NL')$ as follows,
inspired by the definition of $\gt(\ell,\NL)$.
In case $\NL$ is not unfolded, we pick a network $\NL'$ with $\NL \goesto{\tau}\NL'$ and take $\ell' := \ell$.
Otherwise, in case $\NL$ does not occur in $\ell$ or $\ell\upharpoonright \NL$ is incomplete for $\NL$, let $p:=\ch(\ell,\NL)$ and 
$\proc(p,\NL) = \bigoplus_{i\in I}\send{q_i}{\lambda_i} ; \PP_i$. Now pick a $k\in I$
and take $\ell' := \ell_k$ and $\NL' := \NL^p_k$ (as in the definition of $\gt(\ell,\NL)$).
Finally, if $\NL$ is unfolded, $\NL$ occurs in $\ell$, and $\ell\!\upharpoonright\! \NL$ is\\
complete for $\NL$, we take the unique successor pair of $(\ell\!\upharpoonleft\! \NL, \NL)$.

If $(\ell',\NL')$ is the successor of $(\ell,\NL)$ then surely there is a transition $\NL \goesto\tau\NL'$
or $\NL \dgoesto{\comm{p}{\lambda_k}{q_k}~} \NL'$ with $p=\ch(\ell,\NL)$ or $p=\ch(\ell\upharpoonleft \NL,\NL)$
and $k$ as chosen above.\footnote{Each transition $\dgoesto{\comm{p}{\lambda}{q}~}$ can be split
into two transitions $\goesto{\tau}\goesto{\comm{p}{\lambda}{q}~}$. }
Combining those transitions yields for each pair $(\ell,\NL)$ a unique path $\pi(\ell,\NL)$
starting from $\NL$, which is either infinite or ends in a state $\NL'$ with $\proc(r,\NL') = \End$.
Here we use Claim~4.
Moreover, in case $\pi(\ell,\NL)$ is infinite, by (the proof of) Claim 2 it must have a suffix $\pi(\ell',\NL')$
such that $\NL'$ is an unfolded network expression, $\NL'$ occurs in $\ell'$, and
$\ell'\upharpoonright \NL'$ is complete for $\NL'$. Hence that suffix is a simple loop.
By construction, the pair $(\ell',\NL')$ must be unique; call $\pi(\ell',\NL')$ the \emph{loop suffix} of $\pi(\ell,\NL)$. 

First assume that $\pi(h,\M)$ does contain a transition that involves component $r$.
Considering that $\M$ is unfolded, we have $\proc(r,\M) \neq \mu X. \PP$. Hence, this transition must
have the label $\comm{p}{\lambda}{q}$ with $p=r$ or $q=r$.
A simple induction shows that $r \in \participants{\gt^*(h,\M)}$.

Henceforth, we assume that $\pi(h,\M)$ contains no transition involving component $r$.
First assume that $\pi(h,\M)$ ends in a state $\M'$ with $\proc(r,\M') \mathbin= \End$.
Since $\pi(h,\M)$ contains no transition involving $r$,
$\proc(r,\M) \mathbin= \End$.
Finally, assume that $\pi(h,\M)$ is infinite. It suffices to derive a contradiction.

Let $\pi(\ell,\NL)$ be the loop suffix of $\pi(h,\M)$.
Now $\pi(\ell,\NL)$ is infinite and contains no transition that involves component $r$.
Moreover, $\NL$ is reachable from $\N$.

Suppose that there is a location $p$ with $\proc(p,\NL)= \bigoplus_{i\in I}\send{q_i}{\lambda_i} ; \PP_i$,
no transition in $\pi(\ell,\NL)$ involves component $p$ from one of the components $q_i$ for
$i \in I$, and $\NL \dgoesto{\comm{p}{\lambda_i}{q_i}~}$ for all $i \in I$.
In that case $p$ is ready in $\NL$, and $\ell \upharpoonright \NL$ is incomplete for $\NL$.
This contradicts the definition of the loop suffix. 

It follows that for each location $p$ with
$\proc(p,\NL)= \bigoplus_{i\in I}\send{q_i}{\lambda_i} ; \PP_i$,
and such that no transition in $\pi(\ell,\NL)$ involves component $p$,
there exists an $k \in I$ such that
either infinitely many transitions in $\pi(\ell,\NL)$ involve component $q_k$,
or $\NL \;\not\!\!\!\dgoesto{\comm{p}{\lambda_k}{q_k}~}$.
Let $\pi'$ be the infinite path obtained from $\pi(\ell,\NL)$ by transforming all states $\NL'$ in
this path in the same way, namely by replacing, for all locations $p$ as above,
$\proc(p,\NL)= \bigoplus_{i\in I}\send{q_i}{\lambda_i} ; \PP_i$ by the appropriate
$\chosen[2pt]{\send{q_k}{\lambda_k} ; \PP_k}$. By construction this path is just.

The path $\pi'$ is the suffix of a path $\pi''$ that starts in $\N$.
This path contains only finitely many transitions that involve component $r$, and no
state $\N' \pipar \loc{r}{\End}$. Consequently, $\N$ does not satisfies $\Live{\J}$.
\hfill \rule{7pt}{7pt}
\vspace{1ex plus 2pt}

\noindent
Now assume that $\N$ satisfies $\Live{\J}$. Then $\N$ is deadlock-free, and hence $\gt(\varepsilon,\N)$
yields a valid global session type, by Claim~3. To prove that $\N$ is well-typed w.r.t.\ $\gt(\varepsilon,\N)$,
it suffices to show that $\N \vdash \gt(\varepsilon,\N)$.
In fact, we prove a stronger claim, namely that for all histories $h$ and networks $\M$
that are both reachable from $\N$ we have $\M \vdash \gt^*(h,\M)$.

By construction, $\participants{\gt(h,\M)}$ contains only locations of $\M$, and hence of $\N$.
Therefore, the same holds for $\participants{\gt^*(h,\M)}$.
Thus, it remains to establish that $\proc(r,\M) \vdash \proj{r}{\gt^*(h,\M)}$
for all $h$ and $\M$ reachable from $\N$, and all locations $r$ of $\N$.
We do this by coinduction \cite{coinduction}.
We make a case distinction on the shape of $\gt(h,\M)$, and apply induction on $h$, and a nested
induction on the number of recursion-unfolding $\goesto\tau$-transitions possible from $\M$.
Pick $h$, $\M$ and $r$ in the following.
\begin{itemize}
\item Suppose that $\M \goesto\tau \M'$ for a network $M'$, that is,
  $\proc(p,\M)=\rec{X}\PP$ and $\proc(p,\M')= \PP\sub{X}{\rec{X}\PP}$
  for some location $p$ of $\N$. 
  
  In case $r\mathop=p$, using the first rule for $\vdash$, we derive
  $\proc(r,\M) \mathbin= \rec{X}\PP \mathbin\vdash \proj{r}{\gt^*(h,\M)}$ from
  $P\sub{X}{\rec{X}\PP} \mathbin= \proc(r,\M') \mathbin\vdash \proj{r}{\gt^*(h,\M')} \mathbin= \proj{r}{\gt^*(h,\M)}$,
  and the latter is a coinduction hypothesis.

  In case $r \mathbin{\neq} p$, then $\proc(r,\M) \mathbin= \proc(r,\M')$ and
  $\proj{r}{\gt^*(h,\M)} \mathbin= \proj{r}{\gt^*(h,\M')}$.
  Since $\M'$ admits fewer recursion-unfolding $\goesto\tau$-transitions than $\M$, by induction, 
  $\proc(r,\M') \vdash \proj{r}{\gt^*(h,\M')}$.
\end{itemize}
  In the remainder, we assume that $\M$ is already unfolded.
  \begin{itemize}
\item Let $\gt(h,\M) = \End$. Then $\gt^*(h,\M) = \proc(r,\M) =\End$.
  Now the third rule for $\vdash$ yields $\proc(r,\M) \vdash \proj{r}{\gt^*(h,\M)}$.
\item Let $\gt(h,\M) = X_\M$. Then $\M$ occurs in $h$
  and we have $\gt^*(h,\M) \mathbin= \gt^*(h \upharpoonleft \M,\M)$.
  By induction, since $h \upharpoonleft \M$ is strictly shorter than $h$,
  $\proc(r,\M) \vdash \proj{r}{\gt^*(h,\M)}$.
\item Let 
$\begin{array}[t]{@{}r@{}l@{}}
\gt(h,\M)   &= \bigboxplus_{i\in I} \comm{p}{\lambda_i}{q_i} ; \gt(h_i,\M^p_i), \text{ so}\\
\gt^*(h,\M) &= \bigboxplus_{i\in I} \comm{p}{\lambda_i}{q_i} ; \gt^*(h_i,\M^p_i).
\end{array}$\\
  If $p\mathbin=r$, then
  $\proj{r}{\gt^*(h,\M)} \mathbin= \bigoplus_{i\in I} \send{q_i}{\lambda_i} ; (\proj{r}{\gt^*(h_i,\M^p_i)})$.
  By the coinduction hypothesis, for each $i \in I$ we may assume
  $\proc(p,\M^p_i) \vdash \proj{p}{\gt^*(h_i,\M^p_i)}$.
  Moreover, $\proc(p,\M) = \bigoplus_{i\in I}\send{q_i}{\lambda_i} ; \PP_i$ with $\PP_i = \proc(p,\M^p_i)$.
  Now apply the fifth proof rule for $\vdash$.

  If $p \mathbin{\neq} r$, $\proj{r}{\gt^*(h,\M)} \mathbin=\merge_{i\in I} \proj{r}{ (\comm{p}{\lambda_i}{q_i} ;\gt(h_i,\M^p_i)) }$.
  Pick $k \in I$. Applying the last rule for $\vdash$, we need to show that
  $\proc(r,\M) \vdash \proj{r}{ (\comm{p}{\lambda_k}{q_k} ;\gt(h_k,\M^p_k)) }$.

  If $r \neq q_k$, 
  $\proj{r}{ (\comm{p}{\lambda_k}{q_k} ;\gt(h_k,\M^p_k)) } = \proj{r}{\gt(h_k,\M^p_k)}$,
  and $\proc(r,\M) = \proc(r,\M^p_k)$ by Claim~1. Moreover, $\proc(r,\M^p_k) \vdash \proj{r}{\gt(h_i,\M^p_i)}$
  can be used as coinduction hypothesis.

  If $r=q_k$ then
  $$\proj{r}{ (\comm{p}{\lambda_k}{q_k} ;\gt(h_k,\M^p_k)) } = \recv{p}{\lambda_k}; \proj{r}{\gt(h_k,\M^p_k)}.$$
  Moreover, by Claim~1,
  $\proc(r,\M) = \sum_{h \in H} \recv{p_h}{\lambda_h} ; \PQ_h$
  with $p=p_j$, $\proc(r,\M^p_k) = \PQ_h$ and $\lambda_k=\lambda_j$ for some $j \in H$.
  Using $\proc(r,\M^p_k) \vdash \gt(h_k,\M^p_k)$ as coinduction hypothesis,
  $\proc(r,\M) \vdash \recv{p}{\lambda_k}; \proj{r}{\gt(h_k,\M^p_k)}$
  follows by application of the fourth rule for $\vdash$.

\item Let $\gt(h,\M) = \mu X_\M. \bigboxplus_{i\in I} \comm{p}{\lambda_i}{q_i} ; \gt(h_i,\M^p_i)$.
  Let $\gt^*(h,\N,X_\M)$ be defined as $\gt^*(h,\N)$, except that the free variable $X_\M$ does not
  get unfolded. Then $\gt^*(h,\N) = \gt^*(h,\N,X_\M)\sub{X_\M}{\gt^*(h \upharpoonleft \M,\M)}$.
  Hence $\gt^*(h,\M) = \mu X. \bigboxplus_{i\in I} \comm{p}{\lambda_i}{q_i} ; \gt^*(h_i,\M^p_i,X_\M)$.

  First let $r \notin \participants{\gt^*(h,\M)}$. Then $\proj{r}{\gt^*(h,\M)} = \End$.
  By Claim~5, $\proc(r,\M) = \End$. Consequently, $\proc(r,\M) \vdash \proj{r}{\gt^*(h,\M)}$ via the
  third rule for $\vdash$.

  Now if $r \in \participants{\gt^*(h,\M)}$
  then $\proj{r}{\gt^*(h,\M)} = \mu X. \big(\proj{r}{(\bigboxplus_{i\in I} \comm{p}{\lambda_i}{q_i} ; \gt^*(h_i,\M^p_i,X_\M))}\big)$.
  Since $\M$ does not occur in $h$ we have $h_i \upharpoonleft \M = h$.
  Hence \[\proj{r}{\gt^*(h_i,\M_i^p)}
  \begin{array}[t]{@{\,=\,}l@{}}
  \proj{r}{\left(\gt^*(h_i,\M^p_i,X_\M)\sub{X_\M}{\gt^*(h_i \!\upharpoonleft \M,\M)\!}\!\right)\!} \\
  \proj{r}{\left(\gt^*(h_i,\M^p_i,X_\M)\sub{X_\M}{\gt^*(h,\M)}\right)} \\
  \proj{r}{\gt^*(h_i,\M^p_i,X_\M)}\sub{X_\M}{\proj{r}{\gt^*(h,\M)}}.
  \end{array}\]
  In order to obtain $\proc(r,\M) \vdash \proj{r}{\gt^*(h,\M)}$, by the second rule for $\vdash$ it
  suffices to establish $\proc(r,\M)\vdash$
  \[
  \begin{array}[t]{@{}l@{}}
  \proj{r}{(\bigboxplus_{i\in I} \comm{p}{\lambda_i}{q_i} ; \gt^*(h_i,\M^p_i,X_\M))}\sub{X_\M}{\proj{r}{\gt^*(h,\M)}} \\[3pt]
  \mbox{} = \proj{r}{(\bigboxplus_{i\in I} \comm{p}{\lambda_i}{q_i} ; \gt^*(h_i,\M^p_i))}.
  \end{array}\]
  This proceeds exactly as in the previous case.
\end{itemize}

\noindent
This shows that $\N$ is well-typed w.r.t.\ $\gt(\varepsilon,\N)$.
It remains to show that all projections $\proj{p}{\gt(\varepsilon,\N)}$ are guarded.
\vspace{1ex plus 2pt}

\noindent
\textit{Claim 6:}  Let $\N \models\Live{\J}$ and let $r$ be a location of $\N$.
Then $\proj{r}{\gt(\varepsilon,\N)}$ is guarded, \ie it does not have a subexpression of the form $\rec{X} \PQ$ such that $X$ occurs in
$\PQ$ outside the scope of any subexpression $\send{p}{\lambda}; \PP$ or $\recv{p}{\lambda}; \PP$.
\vspace{1ex plus 2pt}

\noindent
\textit{Proof:}
Suppose, towards a contradiction, that $\proj{r}{\gt(\varepsilon,\N)}$ does have a subexpression of
the form $\rec{X} \PQ$ such that $X$ occurs in $\PQ$ outside the scope of any subexpression
$\send{p}{\lambda}; \PP$ or $\recv{p}{\lambda}; \PP$. By the definition of projection,
this subexpression must have the form $\proj{r}{\G}$, with $\G = \rec{X}\G'$ a subexpression of
$\gt(\varepsilon,\N)$.
Given the algorithm of Figure~\ref{fig:algorithm}, $\G$ must have the form $\gt(h,\M)$ for a history $h$ and
network $\M$ reachable from $\N$. Moreover, $X=X_\M$.

There must be a path in the parse tree of $\PQ$ towards the unguarded occurrence of $X$.
Since the occurrence is unguarded, this path passes only through operators $\mu Y$ and $\merge_{i\in I}$.
Backtracking this path through the projection from $\gt(h,\M)$
yields a path $\rho$ in the parse tree of $\gt(h,\M)$ to a subexpression $\gt(\ell,\M) = X_\M$, with $\ell$ an
extension of $h$. This path $\rho$ passes merely through operators $\bigboxplus_{i\in I} \comm{p}{\lambda_i}{q_i}$
with $r$ not being among $p$ and the $q_i$.

The syntactic path $\rho$ induces a path $\pi'$ in the transition system from $\M$ to $\M$.

Suppose that there is a location $p$ with $\proc(p,\M)= \bigoplus_{i\in I}\send{q_i}{\lambda_i} ; \PP_i$,
no transition in $\pi'$ involves component $p$ or from one of the components $q_i$ for
$i \in I$, and $\M \dgoesto{\comm{p}{\lambda_i}{q_i}~}$ for all $i \in I$.
In that case $p$ is ready in $\M$, and $\ell \upharpoonright \M$ is incomplete for $\M$.
This contradicts the definition of $\gt(\ell,\M)$.

It follows that for each location $p$ with
$\proc(p,\M)= \bigoplus_{i\in I}\send{q_i}{\lambda_i} ; \PP_i$,
and such that no transition in $\pi'$ involves component $p$,
there exists a $k \in I$ such that
either some transitions in $\pi'$ stem from component $q_k$,
or $\M \;\not\!\!\!\dgoesto{\comm{p}{\lambda_k}{q_k}~}$.
Let $\pi''$ be the infinite path obtained from $\pi'$ by transforming all states $\NL$ in
this path in the same way, namely by replacing, for all locations $p$ as above,
$\proc(p,\NL)= \bigoplus_{i\in I}\send{q_i}{\lambda_i} ; \PP_i$ by the appropriate
$\chosen[2pt]{\send{q_k}{\lambda_k} ; \PP_k}$.

Let $\pi$ be the path from $\N$ to $\M$, followed by infinitely many repetitions of the loop $\pi''$.
By construction this path is just.
Past $\M$, $\pi$ contains no transitions involving location $r$.
  Thus, invoking the assumption that $\N \models\Live{\J}$, it follows that $r$ successfully
  terminates on $\pi$, that is, $\proc(r,\M)=\End$.
  The algorithm of Figure~\ref{fig:algorithm} implies that $r \notin \participants{\gt(h,\M)}$.
 Now $\gt(h,\M) = \G = \rec{X}{\G'}$ is not closed, for it it were, that would imply that
  $\proj{r}\G \mathbin= \End$, contradicting the assumption that $\proj{r}\G \mathbin= \rec{X}\PQ$.
So, $\gt(h,\M)$ occurs within a subexpression
$\gt(h',\NL)=\rec{Y}\GH$ of $\gt(\varepsilon,\N)$, with $h'$ a strict prefix of $h$, and such that $Y$
occurs freely in $\gt(h,\M)$. Here $Y$ must have the form $\gt(h'',\NL)$, with $h''$ an extension of $h$.
Thus $\NL$ is reachable from $\M$ and hence $\proc(r,\NL)=\End$.
Again, it follows that $r \notin \participants{\gt(h',\NL)}$ and also $\gt(h',\NL)$ is not closed.
Going on this way, we eventually find a subexpression $\rec{Z}\GH'$ of $\gt(\varepsilon,\N)$ that
is is not closed, but also not inside another expression $\rec{W}\GH''$. This contradicts with
$\gt(\varepsilon,\N)$ being closed.
\hfill \rule{7pt}{7pt}
\end{proof}

%% file: appF.tex
\newcommand{\residual}[2]{{ _{#1}\!\!\setminus\!\! \left( {#2} \right) }}
\newcommand{\psize}[2]{ \left\| #1 \right\|_{#2} }
\newcommand{\psizeSmall}[2]{\| #1 \|_{#2} }
\newcommand{\IN}{\mbox{I\hspace{-1pt}N}}             

\subsection{Proof of Soundness}\label{app:F}

\begin{figure*}
\begin{gather*}
\begin{prooftree}
k\in I
\justifies
\bigboxplus_{i \in I}\comm{p}{\lambda_i}{q_i}; \G_i \dgoesto{\comm{p}{\lambda_k}{q_k}} \G_k
\end{prooftree}
\qquad
\begin{prooftree}
\G_i \dgoesto{\comm{p}{\lambda}{q}} \GH_i ~~\mbox{for}~ i \in I \subseteq J, ~~p,q \notin\{ r,s_i \mid i\in I\}
\justifies
\bigboxplus_{i \in J}\comm{r}{\lambda_i}{s_i}; \G_i \dgoesto{\comm{p}{\lambda}{q}} \bigboxplus_{i \in I}\comm{r}{\lambda_i}{s_i}; \GH_i
\end{prooftree}
\qquad
\begin{prooftree}
\GH\sub{X}{\rec{X}\GH} \dgoesto\alpha \G
\justifies
\rec{X}\GH \dgoesto\alpha \G
\end{prooftree}
\end{gather*}
\vspace{-11pt}
\caption{A transition relation between global types}\label{gt trans}
\end{figure*}

This appendix contains the proof of soundness for guardedly well-typed and race-free networks (Theorem~\ref{thm:guarded typing}).
\vspace{1ex}

\begin{definition}{states}
To type network states we extend our type system with the following rule.
\[
\begin{prooftree}
k \in I \qquad  \PP_k \vdash \PQ_k
\justifies
\chosen[3pt]{\send{q_k}{\lambda_k}}; \PP_{k} \vdash \textstyle{\bigoplus_{i \in I}}\, \send{q_i}{\lambda_i} ; \PQ_i
\end{prooftree}
\vspace{1ex}
\]
\end{definition}

\newcommand{\locs}{\textit{loc}}
\vspace{1ex plus 2pt}

Session fidelity for recursion and internal choice can be proven independently.
These lemmas show that $\tau$--transitions preserve the type of a network.
\vspace{1ex}

{
\begin{lemma}{SR recursion}
If $\N \goesto{\tau} \M$ with $\proc(p, \N) = \rec{X} \PP$ and $\proc(p , \M) = \PP\sub{X}{\rec{X} \PP}$,
and $\N \vdashg \G$ then $\M \vdashg \G$.
\end{lemma}
\begin{proof}
If $\N \vdashg \G$ then  $\rec{X} \PP \vdash \proj{p}{\G}$ and $\proj{p}{\G}$ is guarded.
By the type rules, this can hold only if $\PP\sub{X}{\rec{X} \PP} \vdash \proj{p}{\G}$.
Therefore, using that $\proc(q, \N) = \proc(q, \M)$ for all locations $q \neq p$, $\M \vdashg \G$.
\end{proof}
}

{
Let $\locs(\N)$ denote the set of locations of a network state $\N$.%
\vspace{1ex plus 2pt}

\begin{lemma}{SR internal}
If $\N \goesto{\tau} \M$ with
 $\proc(p, \N) = \bigoplus_{i \in I} \send{q_i}{\lambda_i} ; \PP$, and $\proc(p , \M) = \chosen[3pt]{\send{q_i}{\lambda_i}} ; \PP_i$
for some $i \in I$
and $\N \vdashg \G$ then $\M \vdashg \G$.
\end{lemma}

\begin{proof}
Assume $\N \vdashg \G$. So, $\G$ is closed, $\participants{\G} \subseteq \locs(\G)$,
and, for all $p \in \locs(\N)$, $\proc(p,\N) \vdash \proj{p}\G$ and $\proj{p}\G$ is guarded.
The rules of Figure~\ref{fig:types} imply that when $\bigoplus_{i \in I} \send{q_i}{\lambda_i}; \PP \vdash \GH$,
and $i\mathbin\in I$, certainly also $\chosen[3pt]{\send{q_i}{\lambda_i}} ; \PP_i \vdash \GH$.
As $\proc(q, \N) = \proc(q, \M)$ for all locations $q \neq p$, this implies $\M \vdashg \G$.
\end{proof}
}

{
For transitions of the form $\N \goesto{\comm{p}{\lambda}{q}\,} \M$ we target a session fidelity
result, which is stronger than subject reduction, since it constructs a type for $\M$ 
 from the type of $\N$ that reflects the network transition.
For this, we need a few auxiliary concepts.%
\vspace{1ex plus 2pt}

\begin{definition}{psize}
The \emph{maximum depth} $\psize{\G}{p}$  in the abstract syntax tree of $\G$ of
a communication involving location $p$ is defined as follows:
\[
\begin{array}{@{}r@{\;}l@{}}
\psizeSmall{ \End }{p} &= 0
\\
\psizeSmall{ X }{p} &= \infty
\\
\psizeSmall{ \rec{X} \G }{p} &= \left\{
\begin{array}{@{}ll@{}}
  0 & \begin{array}[t]{@{}l@{}}\mbox{if $p \notin \participants{\G}$}\\
      \mbox{{and $\rec{X}\G$ is closed}}
      \end{array}
\\
  1+ \psizeSmall{ \G }{p} & \mbox{otherwise}
\\
\end{array}
\right.
\\
\psizeSmall{ \bigboxplus_{i \in I} \comm{r}{\lambda_i}{q_i} ; \G_i }{p} &= 
\max\left\{ \psize{ \comm{r}{\lambda_i}{q_i} ; \G_i }{p} \colon i \in I \right\}
\\
\psizeSmall{ \comm{r}{\lambda}{q} ; \G }{p} &= 
{\left\{
\begin{array}{@{}lr@{}}
1 & \mbox{if $p \mathbin= r \vee p \mathbin= q$}
\\
1 {+} 
 \psize{ \G }{p}
&
\mbox{if $p \mathbin{\neq} r \wedge p \mathbin{\neq} q$}
\end{array}
\right.}
\end{array}
\]
\end{definition}

\noindent
Call a projection type $\PP$ \emph{fully guarded}, if it is guarded, and each
occurrence of a variable $X$ within $\PP$ occurs within a subexpression  
$\send{p}{\lambda}; \PQ$ or $\recv{p}{\lambda}; \PQ$ of $\PP$.
\vspace{1ex}
}

\begin{lemma}{measure}
  If $\proj{p}{\G}$ is fully guarded then $\psize{\G}{p}$ is finite.
\end{lemma}
\begin{proof}
  A straightforward structural induction on $\G$, using that
  \begin{itemize}
    \item $\proj{p}{(\rec{X} \G)}$ is fully guarded iff $\proj{p}{\G}$ is fully guarded; 
    \item if $p\neq r,q$ then
           $\proj{p}{(\comm{r}{\lambda}{q} ; \G)} = \proj{p}{\G}$.
\qed
  \end{itemize}
\end{proof}

\begin{corollary}{measure}
If $\N\vdashg\G$ and $p \in \locs(\N)$ then $\psize{\G}{p}$ is finite.
\end{corollary}
\begin{proof}
Let $\N\vdashg\G$ and $p \in \locs(\N)$. By \df{guarded typing}, $\proj{p}\G$ is guarded.
By \df{typing}, $\G$, and hence also $\proj{p}\G$, is closed.
Since a closed projection type is fully guarded iff it is guarded, the result follows from \lem{measure}.
\end{proof}

\begin{lemma}{proj}
If $\GH$ is closed then $\proj{p}{(\G\sub{X}{\GH})} = \proj{p}{\G}\sub{X}{\proj{p}{\GH}}$.
\end{lemma}
\begin{proof}
A trivial structural induction on $\G$.
\end{proof}
\vspace{1ex plus 2pt}

{
\begin{lemma}{guarded1}
$\N \vdashg \rec{X} \G$ iff $\N \vdashg \G\sub{X}{\rec{X} \G}$.
\end{lemma}

\begin{proof}
$\rec{X} \G$ is closed iff $\G\sub{X}{\rec{X} \G}$ is closed.
Moreover, $\participants{\rec{X} \G} = \participants{\G} = \participants{\G\sub{X}{\rec{X} \G}}$.
Pick $p \in \locs(\N)$.
it remains to show that $\proj{p}{(\rec{X} \G)}$ is guarded iff $\proj{p}{\G\sub{X}{\rec{X} \G}}$ is guarded,
and $\proc(p,\N)\vdash\proj{p}{(\rec{X} \G)}$ iff $\proc(p,\N)\vdash\proj{p}{\G\sub{X}{\rec{X} \G}}$.

If $p \notin \participants{\G}$ then $\proj{p}{(\rec{X} \G)} = \End$ and $\proj{p}{(\G\sub{X}{\rec{X} \G})}$
must be $\End$ in the scope of some merge operators only. Both are guarded.
Moreover, $\proc(p,\N) \vdash \End$ iff $\proc(p,\N) \vdash \proj{p}{(\G\sub{X}{\rec{X} \G})}$.

If $p \in \participants{\G}$ then we have $\proj{p}{(\rec{X} \G)} = \rec{X}(\proj{p}{\G})$ and
$\proj{p}{(\G\sub{X}{\rec{X} \G})} = \proj{p}{\G}\sub{X}{\rec{X} (\proj{p}{\G})}$ by \lem{proj}.
Now $\proj{p}{\G}\sub{X}{\rec{X} (\proj{p}{\G})}$ is guarded iff $\rec{X}(\proj{p}{\G})$ is guarded.
Moreover, by the second rule for $\vdash$,
$\proc(p, \N) \vdash\rec{X}(\proj{p}{\G})$ iff $\proc(p, \N) \vdash\proj{p}{\G}\sub{X}{\rec{X} (\proj{p}{\G})}$.
\end{proof}
}

{
\begin{lemma}{guarded2}
If $\G$ is closed then $\psize{\rec{X} \G}{p} > \psize{\G\sub{X}{\rec{X} \G}}{p}$ for all
locations $p\in\participants\G$.
\end{lemma}

\begin{proof}
For any $\G$ and closed $\GH$ we have $\psize{\G}{p} \geq \psize{\G\sub{X}{\GH}}{p}$,
by a trivial induction on the structure of $\G$. Hence\\
$\psize{\rec{X}\G}{p}
     \mathop= 1{+} \psize{\G}{p} 
     \mathop\geq 1 {+} \psize{\G\sub{X}{\rec{X}\G}}{p} 
     \mathop> \psize{\G\sub{X}{\rec{X}\G}}{p}$.\\\mbox{}
     \hfill
\end{proof}
}

\begin{lemma}{participants}
  If $\proc(p,\N) \vdash \proj{p}{\G}$ with $p \in \locs(\N)\setminus \participants{\G}$
  and $\G$ is closed, then $\proc(p,\N) \vdash \End$.
\end{lemma}
\begin{proof}
  A trivial structural induction on $\G$.
\end{proof}

Let $\dgoesto{}$ be the transition relation between global session types defined in Figure~\ref{gt trans}.
In combination with Lemmas~\ref{lem:SR recursion} and~\ref{lem:SR internal}, the following session fidelity result
shows how race-free networks evolve according to the global type.
\vspace{1ex}

\begin{lemma}{session fidelity}
For race-free network states $\N$, if $\N\goesto{\comm{p}{\lambda}{q}\,}\M$
and $\N \vdashg \G$ then there exists $\G'$ such that
$\G \dgoesto{\comm{p}{\lambda}{q}} \G'$ and $\M \vdashg \G'$.
\end{lemma}
\renewcommand{\xi}{\lambda}

\begin{proof}
{By \cor{measure}, $\psize\G{p}$ is finite. The proof proceeds by induction on $\psize\G{p}$.
Note that $\proc(p,\N)$ has the form $\chosen[3pt]{\send{q}{\lambda}} ; \PP$.
Since $\chosen[3pt]{\send{q}{\lambda}} ; \PP \nvdash\End$, we can rule out that $\G=\End$.
}

Let $\G = \rec{X} \GH$.
Since $\N \vdashg \G$, one has $\proc(p,\N) \vdash \proj{p}{\G}$ and $\G$ is closed.
By \lem{participants}, {since $\chosen[3pt]{\send{q}{\lambda}} ; \PP \nvdash\End$, we have }$p \in \participants{\G}$.
By \lem{guarded1}, $\N \vdashg \GH\sub{X}{\rec{X}\GH}$.
By \lem{guarded2}, induction may be applied, so $\GH\sub{X}{\rec{X}\GH}\dgoesto{\comm{p}{\lambda}{q}} \G'$ and $\M \vdashg \G'$.
By the third rule of Figure~\ref{gt trans}, $\G\dgoesto{\comm{p}{\lambda}{q}} \G'$.

The remaining case is that $\G = \bigboxplus_{i\in I}\, \comm{r}{\xi_i}{s_i} ; \G_i$.
By unfolding the rules for transitions, we have 
\begin{itemize}
\item $\proc(p, \N) = \chosen[3pt]{ \send{q}{\lambda} } ; \PP$.
\item $\proc(q, \N) = \sum_{j\in J} \recv{p_j}{\lambda_j} ; \PQ_j$,
where $p = p_h$ and $\lambda \mathbin= \lambda_h$ for some $h \in J$.
\end{itemize}
Furthermore, 
$\proc(p, \M) = \PP$,
$\proc(q, \M) = \PQ_h$ 
and $\proc(u, \M) = \proc(u, \N)$ otherwise.

\vspace{1ex plus 2pt}

First assume that $p=r$.  Then
\(
\proj{p}{\G} = \textstyle{\bigoplus_{i\in I}}\, \send{s_i}{\lambda_i} ; (\proj{p}{\G_i})
\).
Hence, by the type rule for $\bigoplus$, since 
$\chosen[3pt]{\send{q}{\lambda} ; \PP} \vdash \proj{p}{\G}$,
there is a $k\in I$ with $s_k = q$, $\lambda_k = \lambda$ and $\PP \vdash \proj{p}{\G_k}$.
We have that $\G \dgoesto{\comm{p}{\lambda}{q}} \G_k$.
It remains to show that $\M \vdashg \G_k$.

Since $\G$ is closed, so is $\G_k$.
Moreover, $\participants{\G_k} \subseteq \participants{\G} \subseteq \locs (\N)=\locs(\M)$.
Thus it remains to show that for each $u\in\locs(\M)$ one has $\proc(u,\M) \vdash \proj{u}{\G_k}$
and $\proj{u}{\G_k}$ is guarded.
When $u=p$, we have $\proc(p,\M)=\PP \vdash \proj{p}{\G_k}$.

When $u\neq p,q$, we have 
\(
\proj{u}{\G} = \textstyle{\merge_{i\in I}}\, \proj{u}{(\comm{p}{\lambda_i}{s_i} ; \G_i)}
\).
Since $\proc(u,\M) = \proc(u,\N) \vdash \proj{u}{\G}$, by the rule for the merge in Figure~\ref{fig:types},
$\proc(u,\M) \vdash \proj{u}{(\comm{p}{\lambda}{q} ; \G_k)} = \proj{u}{\G_k}$.

Similarly, when $u=q$,  $\proc(q,\N) \vdash \proj{q}{(\comm{p}{\lambda}{q} ; \G_k)} = \recv{p}{\lambda}; \proj{q}{\G_k}$.
As $\proc(q, \N) = \sum_{j\in J} \recv{p_j}{\lambda_j} ; \PQ_j$,\vspace{-2pt} there must be an $l\mathbin\in J$ with
$p_l\mathbin=p$, $\lambda_l\mathbin=\lambda$ and $\PQ_l \vdash \proj{q}{\G_k}$.
Now $\N\goesto{\comm{p}{\lambda}{q}\,}\M'$, where $\proc(q,\M')=\PQ_l$.
Since $\N$ is race-free, $\M'=\M$ and thus $\PQ_l=\PQ_h=\proc(q,\M)$.
Hence $\proc(q,\M) \vdash \proj{q}{\G_k}$.

In all these cases $\proj{u}{\G_k}$ is a simple subterm of $\proj{u}{\G}$, not within a recursion
construct, so $\proj{u}{\G_k}$ is guarded because $\proj{u}{\G}$ is guarded.
\vspace{1ex plus 2pt}

Next assume that $p\neq r$. 
Observe that, for all $i \in I$, we have $s_i \neq p$; otherwise the first actions in $\proc(p, \N)$
would be an external choice of receive actions, which is impossible.

{
\renewcommand{\PR}{{{{\rm V}\hspace{-1pt}_i}}}
\newcommand{\PQi}{{{{\rm U}_i}}}
The thread $\proc(r,\N)$, possibly after unfolding recursion, must have the form
$\bigoplus_{i\in I_0} \send{s_i}{\lambda_i}; \PP_i$ with $\PP_i \vdash \proj{r}{\G_i}$; here $I_0 \mathbin\subseteq I$.
For each $i \mathbin\in I_0$ we define a network state $\N_i$ such that $\N_i \vdashg \G_i$.
Take $\proc(r,\N_i):= \PP_i$.
The thread $\proc(s_i,\N)$, possibly after unfolding recursion, must be of the form $r?\lambda_i; \PR + \PQi$, where
$\PR \vdash \proj{s_i}{\G_i}$; we take $\proc(s_i,\N_i) := \PR$.
For $u\mathop{\neq} r,s_i$ take $\proc(u,\N_i) := \proc(u,\N)$.
Since $\proj{u}\G \mathbin= \merge_{i\in I} \proj{u}{\G_i}$ and $\proc(u,\N_i) \vdash \proj{u}\G$, we have
$\proc(u,\N_i) \vdash \proj{u}{\G_i}$.
Note that $\G_i$ is closed since $\G$ is closed, and
$\participants{\G_i} \subseteq \participants{\G} \subseteq \locs(\N) =: \locs(\N_i)$.
Moreover, for $u\in\locs(\N)$, $\proj{u}{\G_i}$ is guarded since $\proj{u}{\G}$ is guarded.
It follows that indeed $\N_i \vdashg \G_i$.

For each $i \in I_0$ we have $\N\mathrel{\raisebox{0pt}[4pt]{$\goesto\tau$}^*} \goesto{\comm{r}{\lambda_i}{s_i}} \N_i$.
Hence $\N_i$ is race-free. As $p,q \notin \{r,s_i \mid i \mathbin\in I_0\}$, by 
race-freedom of $\N$,
it follows that $\N_i\goesto{\comm{p}{\lambda}{q}\,}\M_i$, where
$\proc(p,\M_i) = \PP$,
$\proc(q,\M_i) \mathbin= \PQ_h$,
$\proc(r,\M_i) \mathbin= \PP_i$,
$\proc(s_i,\M_i) \mathbin= \PR$,
and $\proc(u,\M_i) = \proc(u,\N)$ for all $u \notin \{p,q,r,s_i\}$.
Furthermore $\psize{\G_i}{p} < \psize{\G}{p}$.

By the induction hypothesis there are $\G'_i$, for $i \in I_0$, with
$\G_i \dgoesto{\comm{p}{\lambda}{q}} \G'_i$ and $\M_i \vdashg \G'_i$.
Thus, by the second rule of Figure~\ref{gt trans},
$\G \dgoesto{\comm{p}{\lambda}{q}} \G' := \bigboxplus_{i\in I_0}\, \comm{r}{\xi_i}{s_i} ; \G'_i$.
Trivially, $\G'$ is closed and $\participants{\G'}\subseteq \locs(\M)$.
Moreover, $\proj{u}{\G'}$ is guarded for all $u\in\locs(\N)$.
It remains to show that $\proc(u,\M) \vdash \proj{u}{\G'}$ for all $u \in \locs(\M)$.

We have $\proj{p}{\G'} = \merge_{i\in I_0} \proj{p}{\G'_i}$.
Since $\M_i \vdashg \G'_i$, we have $\proc(p,\M_i) = \PP \vdash \proj{p}{\G'_i}$ for all $i \in I_0$.
Thus, by the typing rule for merge, $\proc(p,\M) = \PP \vdash\proj{p}{\G'}$.

Likewise, $\proj{q}{\G'} = \merge_{i\in I_0} \proj{q}{\G'_i}$,
$\proc(q,\M_i) = \PQ_h \vdash \proj{q}{\G'_i}$ for all $i \in I_0$,
and $\proc(q,\M) = \PQ_h \vdash\proj{q}{\G'}$.

We have $\proj{r}{\G'} = \bigoplus_{i\in I_0} \send{s_i}{\lambda_i}; (\proj{r}{\G'_i})$.
Since $\M_i \vdashg \G'_i$, we have $\proc(r,\M_i) = \PP_i \vdash \proj{r}{\G'_i}$ for all $i \in I_0$.
Thus, by the typing rule for internal choice,
$\bigoplus_{i\in I_0} \send{s_i}{\lambda_i};\PP_i \vdash\proj{r}{\G'}$.
Thus $\proc(r,\M) = \proc(r,\N) \vdash \proj{r}{\G'}$.

For $u \neq p,q,r$ we have $\proj{u}{\G'} = \merge_{i\in I_0} \proj{u}{(\comm{r}{\lambda_i}{s_i}; \G'_i)}$.
Hence we need to show that $\proc(u,\M) \vdash  \proj{u}{(\comm{r}{\lambda_i}{s_i}; \G'_i)}$ for all $i \in I_0$.
So, pick $i\in I_0\subseteq I$.

First suppose $u\neq s_i$.
Since $\proc(u,\M) = \proc(u,\N) = \proc(u,\M_i)$ and $\M_i \vdashg \G_i'$,
we have $\proc(u,\M)\vdash \proj{u}{\G'_i} = \proj{u}{(\comm{r}{\lambda_i}{s_i}; \G'_i)}$.

Finally, suppose $u = s_i$.
As $\M_i \vdashg \G_i'$, we have $\proc(s_i,\M_i) = \PR \vdash \proj{s_i}{\G'_i}$.
By the typing rule for external choice,
$r?\lambda_i; \PR + \PQi \vdash r?\lambda_i; (\proj{s_i}{\G'_i}) = \proj{s_i}{(\comm{r}{\lambda_i}{s_i}; \G'_i)}$.
Hence $\proc(s_i,\M) = \proc(s_i,\N) \vdash \proj{s_i}{(\comm{r}{\lambda_i}{s_i}; \G'_i)}$.
}\end{proof}

{\newcommand{\lambdai}{{a}}
\begin{observation}{fidelity}
The second rule in Figure~\ref{gt trans} allows the index set to be narrowed. 
To understand why,
consider the following global type.
\[
\G \triangleq
\begin{array}[t]{l}
 \left(   \comm{r}{\lambdai}{t} ; \comm{p}{\lambdai}{q} ; \comm{q}{\lambdai}{r} ; \comm{r}{\lambdai}{s} ; \comm{s}{\lambdai}{q}  \right)
\\
    \boxplus
~
     \comm{r}{\lambdai}{s} ; \comm{s}{\lambdai}{q} ; \comm{p}{\lambdai}{q} ; \comm{q}{\lambdai}{r} ; \comm{r}{\lambdai}{t}
\end{array}
\]
Global type $\G$ guardedly types the following network $\N$.
\[
\begin{array}{rl}
& 
\loc{p}{ \send{q}{\lambdai} }
\\
\pipar
&
\loc{q}{(\recv{p}{\lambdai} ; \send{r}{\lambdai} ; \recv{s}{\lambdai})  +  (\recv{s}{\lambdai} ; \recv{p}{\lambdai} ; \send{r}{\lambdai})}
\\
\pipar
&
\loc{r}{ \send{t}{\lambdai} ; \recv{q}{\lambdai} ; \send{s}{\lambdai}  }
\\
\pipar
&
\loc{s}{ \recv{r}{\lambdai} ; \send{q}{\lambdai}  }
\\
\pipar
&
\loc{t}{ \recv{r}{\lambdai}  }
\end{array}
\]
Network $\N$ is race-free and $\N \goesto{\comm{p}{\lambdai}{q}} \M$.
If we insisted that $I = J$ in Fig.~\ref{gt trans}, then there would be no $\G'$ such that 
$\G \dgoesto{\comm{p}{\lambdai}{q}} \G'$
and $\G' \vdashg \M$, as required for session fidelity.
Narrowing of the global type by hiding branches of a choice, as permitted by $I \subseteq J$, is required when we have a race-free network, but the global type is not race-free, as in the above example.
\vspace{1ex plus 2pt}
\end{observation}
}

Using the above, we can prove our soundness result.
This is where we appeal to justness.

\begin{trivlist}
 \item[\hspace{\labelsep}\hspace{10pt}\textit{
{\thm{guarded typing}}:}]
\indent
If $\N$ is guardedly well-typed and race-free, then $\N\models\Live{\J}$.
\end{trivlist}

\begin{proof}
{
Let $\N=\N_0$ be race-free and assume that $\N_0 \vdashg \G_0$. 
Let \(\pi = \N_0 \mathrel{\raisebox{0pt}[4pt]{$\goesto\tau$}^*}\goesto{\comm{p_0}{\lambda_0}{q_0}} \N_1
\mathrel{\raisebox{0pt}[4pt]{$\goesto\tau$}^*}\goesto{\comm{p_1}{\lambda_1}{q_1}} \dots\)
be a path on which some location $p\in\locs(\N)=\locs(\N_i)$ does not successfully terminate, and that contains only finitely many
transitions involving $p$. We aim to show that $\pi$ is not just.

Let $\ell(\pi) \in \IN$ be the index $n$ of the last state
$\N_n$ in this path, or $\ell(\pi)=\infty$ if $\pi$ is infinite.
By Lemmas~\ref{lem:SR recursion},~\ref{lem:SR internal} and~\ref{lem:session fidelity} there is a
sequence $\G_0 \dgoesto{\comm{p_0}{\lambda_0}{q_0}} \G_1 \dgoesto{\comm{p_1}{\lambda_1}{q_1}} \dots$
of length $\ell(\pi)$ such that $\N_i \vdashg \G_i$ for all $i$.

First consider the special case that for some $\G_k$ in this sequence we have $p \notin \participants{\G_k}$.
Since $\G_k$ is closed and $\proc(p,\N_k) \vdash \proj{p}{\G_k}$, \lem{participants} yields that
$\proc(p,\N_k) \vdash \End$. This implies that $\proc(p,\N_k)$ must have the form $\End$ or $\rec{X}\End$ 
As we assumed that $p$ does not successfully terminate on $\pi$, it must stay a $\tau$-transition
away from successful termination. As this $\tau$-transition is local to $p$, it follows that $\pi$
is not just.

Thus we may assume that $p \in \participants{\G_k}$ for all $\G_k$ in the above sequence.
By \cor{measure}, $\psize{\G_i}{p}$ is finite for all $i$.
\vspace{1ex plus 2pt}

\textit{Claim:} When $\N \vdashg \G$, $\G \dgoesto{\comm{t}{\lambda}{q}} \GH$, $p \in \participants{\G}$ and $p\neq t,q$ then
$\psize{\GH}{p}\leq \psize{\G}{p}$.
Moreover, if the transition $\G \dgoesto{\comm{t}{\lambda}{q}} \GH$ is derived without using the second rule in
Figure~\ref{gt trans}, then $\psize{\GH}{p} < \psize{\G}{p}$.
\vspace{1ex plus 2pt}

\textit{Proof:} A trivial induction on the derivation of $\G \dgoesto{\comm{t}{\lambda}{q}} \GH$.
Note that the conclusion $\psize{\GH}{p} < \psize{\G}{p}$ is not warranted when\linebreak\newpage  
\noindent the second rule is
used, due to the possibility that $p =r$ or $p=s_i$, where $r$ and $s_i$ are location variables of that rule.
\hfill \rule{7pt}{7pt}
\vspace{1ex plus 2pt}

\textit{Application of the claim:}
Since we assumed that $\pi$ contains only finitely many transitions involving $p$,
by restricting attention to a suffix of $\pi$ we may just as well assume that
no transition in $\pi$ involves $p$, \ie all $p_i$ and $q_i$ differ from $p$.

As $\psize{\G_i}{p} \geq 0$ for all $i$, there must be a $\G_k$ in the above sequence\vspace{-2pt}
such that $\psize{\G_l}{p} = \psize{\G_k}{p}$ for all $k \leq l \leq \ell(\pi)$, with $l\in\IN$. So, past $\G_k$, all transitions
$\G_l \dgoesto{\comm{p_l}{\lambda_l}{q_l}} \G_{l+1}$ are derived by means of the second rule of Figure~\ref{gt trans}.
Note that $\G_k\neq\End$ since $p \in \participants{\G}$, and $\G_k\neq X$ since $\G_k$ is closed.
Thus, possibly after unfolding recursion, $\G_k$ must have the form
$\bigboxplus_{j\in J}\, \comm{r}{\mu_j}{s_j} ; \GH_j$.
Since $\N_k \vdashg \G_k$, we have $\proc(r,\N_k) \vdash \bigoplus_{j\in J}\, \send{s_j}{\mu_j} ; (\proj{r}{\GH_j})$
and $\proc(s_j,\N) \vdash \recv{r}{\mu_j}; (\proj{s_j}{\GH_j})$ for each $j \in J$.
In case $\proc(r,\N_k)$ never performs the $\tau$-transitions needed to reach a thread state
$\chosen[2pt]{\send{s_j}{\mu_j}; \PP_j}$ with $j\in J$, the path
$\pi$ is not just, and we are done.
Likewise $s_j$ will reach a state where it is ready to receive $\mu_j$ from $r$.
So, for some $k \leq l \leq \ell(\pi)$,  with $l\in\IN$ and $j\in J$,   we have $\N_l \goesto{\comm{r}{\mu_j}{s_j}}$\;.

A straightforward induction on $l \leq m < \ell(\pi)$ shows that
$\G_m$ has the form $\bigboxplus_{j\in J_m}\, \comm{r}{\mu_j}{s_j} ; \GH^m_j$
with $j \in J_m \subseteq J$, so that $p_m \neq r$ and $q_m \neq s_j$.
Here $j\in J_m$ follows since $\N_m \vdashg \G_m$ and thus $\chosen[2pt]{\send{s_j}{\mu_j}; \PP_j} \vdash \proj{r}{\G_m}$,
and $p_m,q_m \neq r,s_j$ follows from the side condition in the second rule of Figure~\ref{gt trans}.
It follows that in the path $\pi$, past $\N_l$ neither location $r$ nor $s_j$ makes progress, and the
transition $\comm{r}{\mu_j}{s_j}$ remains enabled. Hence $\pi$ is not just.
}
\end{proof}